\PassOptionsToPackage{dvipsnames,table}{xcolor}
\documentclass[11pt,letterpaper]{article}

\usepackage[utf8]{inputenc}
\usepackage[T1]{fontenc}
\usepackage[margin=1in]{geometry}
\usepackage{mathpazo}
\usepackage{amsmath,amsfonts,amssymb,amsthm,mathtools}
\usepackage{microtype}
\usepackage{xcolor}

\usepackage{graphicx}
\usepackage{tikz}
\usetikzlibrary{arrows.meta,calc}
\usepackage[font=small,labelfont=bf]{caption}
\usepackage{placeins}
\usepackage{enumitem}
\usepackage{etoolbox}
\definecolor{constructionblue}{RGB}{31,87,130}
\definecolor{constructionorange}{RGB}{178,94,35}

\definecolor{niceRed}{RGB}{190,38,38}
\definecolor{niceYellow}{HTML}{f5b400}
\definecolor{blueGrotto}{HTML}{059DC0}
\definecolor{royalBlue}{HTML}{057DCD}
\definecolor{navyBlue}{HTML}{0B579C}
\definecolor{yaleBlue}{HTML}{00356b}
\definecolor{limeGreen}{HTML}{81B622}
\definecolor{nicePurple}{HTML}{9c27b0}
\definecolor{lightRoyalBlue}{HTML}{def2ff}
\definecolor{gold}{HTML}{ffa300}

\usepackage{hyperref}
\hypersetup{
  colorlinks=true,
  urlcolor=blueGrotto,
  linkcolor=royalBlue,
  citecolor=navyBlue
}

\makeatletter
\renewcommand\paragraph{\@startsection{paragraph}{4}{\z@}%
  {\smallskipamount}%
  {-1em}%
  {\normalfont\normalsize\bfseries}}
\makeatother

\AtBeginDocument{%
  \AfterEndEnvironment{theorem}{\noindent\ignorespaces}%
  \AfterEndEnvironment{lemma}{\noindent\ignorespaces}%
  \AfterEndEnvironment{proposition}{\noindent\ignorespaces}%
  \AfterEndEnvironment{corollary}{\noindent\ignorespaces}%
  \AfterEndEnvironment{remark}{\noindent\ignorespaces}%
  \AfterEndEnvironment{proof}{\noindent\ignorespaces}%
}

\usepackage{soul}
\setul{2pt}{0.4pt}
\usepackage[framemethod=TikZ]{mdframed}
\mdfsetup{
  backgroundcolor=royalBlue!0,
  roundcorner=4pt,
  skipabove=5pt,
  skipbelow=5pt,
  linewidth=1pt
}

\newenvironment{outputlist}[1]
  {\begin{list}{}{%
     \settowidth{\labelwidth}{\textnormal{#1}}%
     \setlength{\leftmargin}{\labelwidth}%
     \addtolength{\leftmargin}{\labelsep}%
     \addtolength{\leftmargin}{0.5em}}}
  {\end{list}}

\usepackage[
  backref=true,
  backend=biber,
  natbib=true,
  style=alphabetic,
  sorting=alphabeticlabel,
  sortcites=true,
  minbibnames=3,
  maxbibnames=999,
  mincitenames=4,
  maxcitenames=4,
  minalphanames=4,
  maxalphanames=4,
  doi=false,
  isbn=false
]{biblatex}
\DeclareSortingTemplate{alphabeticlabel}{
  \sort[final]{\field{labelalpha}}
  \sort{\field{year}}
  \sort{\field{title}}
}
\AtEveryBibitem{\clearname{editor}\clearlist{location}}

\AtBeginRefsection{\GenRefcontextData{sorting=ynt}}
\AtEveryCite{\localrefcontext[sorting=ynt]}

\usepackage[capitalise,noabbrev,nameinlink,sort]{cleveref}
\crefname{section}{Section}{Sections}
\crefname{theorem}{Theorem}{Theorems}
\crefname{lemma}{Lemma}{Lemmas}
\crefname{proposition}{Proposition}{Propositions}
\crefname{corollary}{Corollary}{Corollaries}
\crefname{remark}{Remark}{Remarks}
\crefname{equation}{Equation}{Equations}
\crefname{appendix}{Section}{Sections}
\crefname{figure}{Figure}{Figures}

\usepackage{xcolor,umoline}
\definecolor{reviewdeleted}{RGB}{125,125,125}
\definecolor{reviewadded}{RGB}{0,125,55}
\makeatletter
\let\review@originalUMOspace\UMOspace
\protected\def\UMOspace{\review@originalUMOspace}
\let\review@originalUMOnewline\UMOnewline
\protected\def\UMOnewline{\review@originalUMOnewline}

\newsavebox{\review@strikebox}
\DeclareRobustCommand{\reviewdelmath}[1]{\mathpalette\review@delmath{#1}}
\newcommand{\review@delmath}[2]{\begingroup\color{reviewdeleted}%
  \sbox{\review@strikebox}{$\m@th#1#2$}%
  \rlap{\raisebox{.45ex}{\rule{\wd\review@strikebox}{.4pt}}}%
  \usebox{\review@strikebox}\endgroup}

\AtBeginDocument{\@ifpackageloaded{hyperref}{%
  \pdfstringdefDisableCommands{%
    \def\reviewdelmath#1{#1}%
    \def\UMOspace#1[#2]{}\def\UMOnewline{}}}{} }
\makeatother

\graphicspath{{Draft2/}{./}}

\newtheorem{theorem}{Theorem}[section]
\newtheorem{lemma}[theorem]{Lemma}
\newtheorem{corollary}[theorem]{Corollary}
\newtheorem{proposition}[theorem]{Proposition}
\theoremstyle{remark}
\newtheorem{remark}[theorem]{Remark}
\numberwithin{equation}{section}

\newcommand{\CLS}{\mathsf{CLS}}

\theoremstyle{plain}
\newtheorem{definition}{Definition}

\title{On the Complexity of Finding Fixed Points for\\ Set-Valued Contractions}
\author{Emmanouil-Vasileios Vlatakis-Gkaragkounis\thanks{Department of Computer Sciences, University of Wisconsin-Madison, email: \texttt{vlatakis@wisc.edu}} \and Pucheng Xiong\thanks{Department of Computer Sciences, University of Wisconsin-Madison, email: \texttt{pxiong79@wisc.edu}}}
\date{}

\begin{document}
\maketitle
\begin{abstract}
We study the computational complexity of finding fixed points of set-valued contractions. We first formulate a computational problem for Nadler's fixed-point theorem, \textsc{Projected-Nadler}, and prove that it is \textsf{CLS}-complete by showing its equivalence to \textsc{Continuous-LocalOpt}. We then establish a stronger converse to Nadler's theorem, which can be applied as a tool to analyze the convergence rate of the set-valued basic iteration procedure. Finally, we reduce the large-margin triplet stationarity problem to \textsc{Projected-Nadler}. Combining the argument for \textsf{CLS}-hardness of exact triplet stationarity problem \cite{YanEtAlTriplet} with the
\textsf{CLS}-hardness of approximate KKT points of quadratic programs
\cite{Fearnley2025KKT}, we obtain \textsf{CLS}-completeness of large-margin
triplet stationarity at any prescribed accuracy.
\end{abstract}
\newpage
{\small \tableofcontents}
\newpage

\section{Introduction}
\label{sec:introduction}

Many computations are iterations: gradient descent, power iteration,
expectation--maximization, alternating minimization, and coordinate descent
all apply an update rule repeatedly in the hope that the orbit settles
at a fixed point. For a single-valued map $f$ on a complete metric space,
Banach's contraction principle (also known as Banach's fixed point theorem) gives the cleanest possible certificate: if
$d(f(x),f(y))\leq c\,d(x,y)$ for some $c<1$, then $f$ has a unique fixed
point and every orbit converges to it geometrically. The converse of the Banach contraction principle \cite{Bessaga1959,Meyers1967} shows that this certificate is surprisingly universal: assuming global convergence to the unique fixed point, one can change the metric so that the dynamics
become a contraction. Its computational counterpart is equally sharp:
finding the fixed point of a succinctly represented contraction map is complete
for continuous local search
\cite{DaskalakisPapadimitriou2011,DaskalakisTzamosZampetakis2018}.

\paragraph{From functions to correspondences.}
Many algorithms are not intrinsically single-valued. Their next state may
depend on a tie, an arrival order, a coordinate permutation, an equilibrium
selection, or an adversarially chosen minimizer. Fixing one selection in
advance hides the robustness question that matters: does convergence hold for
\emph{every} admissible sequence of choices? Random-permutation cyclic
coordinate descent is a representative example: once the distribution over
permutations is discarded, one epoch sends the current point to the finite
set of all possible epoch outputs, and each trajectory of this correspondence
is a possibly adaptive sequence of coordinate orders. Correspondences also
arise directly in equilibrium computation: best-response maps,
feasible-action maps, generalized and quasi-variational inequalities are
set-valued, and when the feasible region depends on the current point,
replacing the model by a predetermined selection changes the problem
\cite{VI_and_NE,GQVI_chan,complexity_of_QVI,quasivariational-inequalities-local}.
The relevant fixed-point condition is $x\in F(x)$. A useful geometric picture
comes from iterated function systems: if $w_1,\ldots,w_m$ are contractions,
the Hutchinson operator $\mathcal H(S)=\bigcup_{i=1}^m w_i(S)$ is a
contraction on the hyperspace of nonempty compact sets under the Hausdorff
metric, and its fixed point is a self-similar attractor
\cite{Hutchinson1981,BarnsleyDemko1985}. Strictly speaking, $\mathcal H$ is
a single-valued Banach contraction whose \emph{points} are compact sets, but
Figure~\ref{fig:hutchinson-fractal} makes the Hausdorff geometry behind
Nadler's theorem concrete.

\begin{figure}[t]
\centering
\includegraphics[width=0.24\linewidth]{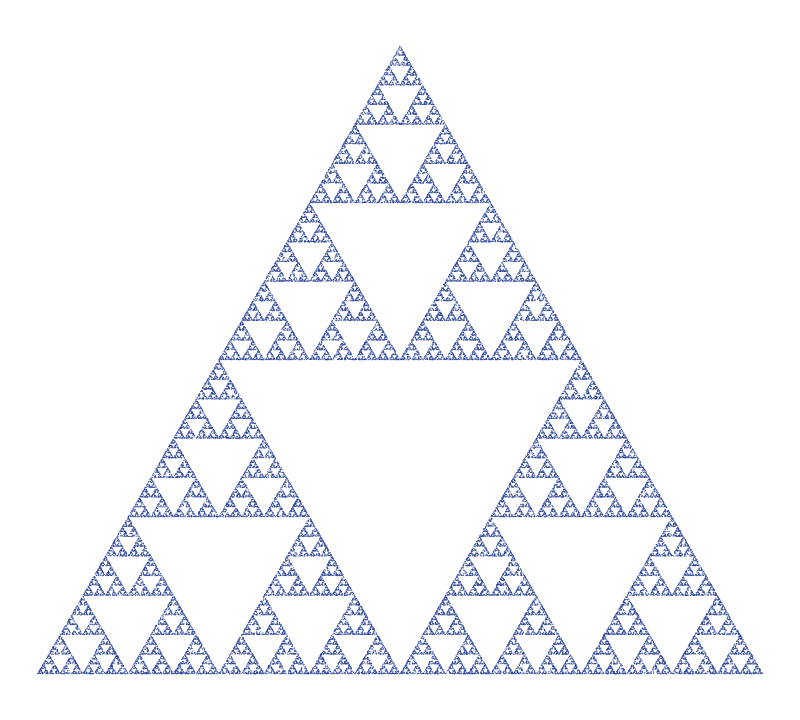}
\caption{A Hausdorff-metric fixed point: the Sierpi\'nski attractor is the
unique fixed point of the Hutchinson operator
$\mathcal H(S)=w_1(S)\cup w_2(S)\cup w_3(S)$ on the hyperspace of nonempty
compact subsets of $\mathbb R^2$, illustrating the metric geometry shared by
Banach and Nadler contractions.}
\label{fig:hutchinson-fractal}
\end{figure}

\paragraph{The missing computational counterpart.}
Nadler's theorem is the contraction principle for correspondences: if
$(X,d)$ is complete, each $F(x)$ is nonempty, closed, and bounded, and
$H_d(F(x),F(y))\leq k\,d(x,y)$ for some $k<1$, then $F$ has a fixed point
\cite{Nadler1969}. The theorem sits naturally between two pillars of
computational fixed-point theory. It is set-valued, like Kakutani's theorem,
whose computational formulations lead to $\mathsf{PPAD}$ \cite{PapadimitriouVlatakisGkaragkounisZampetakis2023}, and contractive,
like Banach's theorem, whose computational formulation leads to $\CLS$
\cite{DaskalakisTzamosZampetakis2018}.
Yet its computational content has not received a corresponding
classification. This paper asks:
\begin{center}
\setlength{\fboxsep}{7pt}
\fbox{\begin{minipage}{0.9\linewidth}
\textit{Can Nadler's fixed-point guarantee be captured by a natural total
search problem? Conversely, when can robust global convergence of a
set-valued iteration be certified as Hausdorff contraction after a change of
metric?}
\end{minipage}}
\end{center}
A satisfactory formulation must provide succinct access to the values $F(x)$, make
approximate fixedness verifiable, and return checkable witnesses when a
claimed contraction or regularity property is false. It must also allow the
metric to be part of the instance: the energy metrics and potential
ultrametrics that reveal contraction below are not interchangeable with the
ambient Euclidean metric.

\paragraph{The projection barrier.}
The most tempting representation is a nearest-point circuit $\Pi_F(x,z)$
and its diagonal selection $P(x)=\Pi_F(x,x)$. A subtlety arises immediately:
Hausdorff–Lipschitz correspondence with convex values does \emph{not} generally
make $P$ Lipschitz. Even for exact projections, the correct moving-set
estimate is only square-root,
\[
 \|\Pi_A(z)-\Pi_B(z)\|_2^2
 \leq
 d_{\mathrm H}(A,B)
 \bigl(\operatorname{dist}(z,A)+\operatorname{dist}(z,B)\bigr),
\]
so on a bounded domain a Hausdorff--Lipschitz correspondence yields a
H\"older-$1/2$ diagonal projection in general, and the exponent is sharp.
Projection-based formulations of computational Kakutani face the same obstacle; any proof that silently promotes the
projection to a Lipschitz map loses a square root. Our formulation
accommodates this geometry rather than assuming it away: we fix a rational
exponent $\beta\in(0,1]$, with $\beta=1/2$ as the canonical convex Euclidean
case, and treat a claimed $\beta$-H\"older bound for the implemented map $P$
as a separately verifiable condition. The price is that $P$ falls outside
the Lipschitz regime of standard continuous local search; the key to
recovering it is a succinct Kuhn-simplicial interpolant, which needs only
$n+1$ grid values per simplex instead of the $2^n$ for multilinear
interpolation. At the same time, the dynamical heart of Nadler's theorem
needs no continuous selector at all: for $r(x)=d(x,P(x))$, exact
nearest-point responses give $r(P(x))\leq k\,r(x)$. We do not assume
exactness. Instead, the contraction certificate tests this inequality, up
to a small additive tolerance, at $(x,P(x),x)$: either $r$ decreases
sufficiently along $P$, making it a local-search potential, or the triple
is itself a certificate. This works for arbitrary input circuits, and a
failed test reflects a genuine contraction violation whenever the
circuits' errors fit within the tolerance.

\subsection{Our contributions and techniques}

\paragraph{1. A computational formulation of Nadler's theorem.}
For every fixed rational $\beta\in(0,1]$, we formulate the total search
problem $\beta$-\textsc{Projected-Nadler}. The input comprises a complete
metric $d$, a correspondence $F\colon X\rightrightarrows X$ represented by an approximate
nearest-point projection circuit $\Pi_F(x,z)$ that returns an approximate closest point to
$z$ in $F(x)$, a claimed Hausdorff contraction factor $k\in(0,1)$, and
claimed H\"older/metric regularity bounds. The diagonal selection
$P(x)=\Pi_F(x,x)$ induces the residual $r(x)=d(x,P(x))$. The representation
parallels computational Kakutani fixed points
\cite{PapadimitriouVlatakisGkaragkounisZampetakis2023}, but deliberately
omits small-volume outputs, which would trivialize singleton correspondences
(Section~\ref{sec:nadler-continuous-local-opt}). A valid solution is an
approximate fixed point, a projection-based certificate that the Hausdorff
contraction fails, or a succinct violation of a regularity bound. Our first
results establish that
\[
 \boxed{\quad
 \beta\text{-}\textsc{Projected-Nadler}
 \ \equiv_p\
 \beta\text{-}\textsc{H\"older-LocalOpt}
 \ \equiv_p\
 \textsc{Continuous-LocalOpt},\quad}
\]
so $\beta$-\textsc{Projected-Nadler} is $\CLS$-complete
(Theorem~\ref{thm:nadler-clo-equivalence}) relative to the canonical complete
problem for $\CLS$ \cite{DaskalakisPapadimitriou2011}. The problem $\beta$-\textsc{H\"older-LocalOpt} is
\textsc{Continuous-LocalOpt} with its Lipschitz conditions weakened to $\beta$-H\"older conditions. We need it because $P$ only satisfies H\"older condition in general, as discussed in the projection barrier. We show equivalence between \textsc{Continuous-LocalOpt} and $\beta$-\textsc{H\"older-LocalOpt} through a H\"older-to-Lipschitz
interpolation theorem that operates in polynomial dimension, preserves succinctness of representations,
and loses only polynomial precision for fixed $\beta$, which we mentioned when describing the projection barrier. For hardness, we reduce from \textsc{Continuous-LocalOpt}. We encode a continuous local-search instance $(f,p)$ as the singleton correspondence
$F_f(x)=\{f(x)\}$ under a complete ultrametric whose nonzero distances track
the potential values of their endpoints; a contraction violation then exposes
a point where $p$ fails to decrease sufficiently. Since this ultrametric is
not topologically equivalent to the Euclidean metric, the metric must be part
of the input of our computational formulation of Nadler's theorem.

\paragraph{2. A quantitative converse to Nadler's theorem.}
Classical converse results for Banach's fixed point theorem show that globally
convergent single-valued iterations can be made contractive after a suitable
change of metric \cite{Meyers1967,DaskalakisTzamosZampetakis2018}. For correspondences,
the hypothesis must cover \emph{every} branch of the iteration at once,
since convergence of one selected trajectory says nothing about the rest. We
therefore study correspondences for which every admissible trajectory
converges robustly, and locally uniformly, to an endpoint $x^*$. This is stronger than $x^*\in F(x^*)$ and
is necessary because a Nadler contraction can otherwise have many fixed points.
Under these hypotheses, for every contraction factor $c\in(0,1)$ and
accuracy scale $\varepsilon>0$ there is a complete metric $D_{c,\varepsilon}$
that generates the original topology, makes $F$ a Hausdorff $c$-contraction,
and transfers small distances back to the original metric
(Theorem~\ref{thm:daskalakis-nadler-converse}). This strengthens the set-theoretic converses to Fryszkowski's problem
\cite{Comaneci2017,Luchian2018}, reviewed in
Appendix~\ref{subsec:fryszkowski}. Those results characterize when some complete contracting metric exists, but the metric they produce need not generate the original topology, and closeness in it need not carry a quantitative guarantee about closeness in the original metric. These are
precisely the two properties that an algorithmic convergence analysis needs,
and our construction provides both, following the framework of the converse of Banach's theorem developed in \cite{DaskalakisTzamosZampetakis2018}. The construction performs the potential-to-metric construction of our complexity result in reverse. The convergence hypothesis manufactures a \emph{rank}, which is the depth of a point in the nested global images of
$F$. Every admissible step lands in a deeper image, so the rank increases
along every branch. We then enlarge the original metric to an equivalent
one under which $F$ is Hausdorff-nonexpansive, weight it by
$c^{\mathrm{rank}}$, and pass to shortest-chain (geodesic) distances. Each step now
contracts by the factor $c$ while the topology is unchanged, and the
residual $\operatorname{dist}_{D_{c,\varepsilon}}(x,F(x))$ is comparable
to the distance from $x^*$.

\paragraph{3. Robust convergence of coordinate descent under arbitrary
orderings.}
As an illustration of the converse, we study random-permutation cyclic
coordinate descent on strongly convex quadratics, precisely the kind of
correspondence motivated above, whose next state depends on an ``adaptive'' 
ordering choice. One coordinate at a time may not be the right object: the
one-coordinate correspondence has branches that revisit the same nonoptimal
point forever, since a coordinate update fixes an entire hyperplane rather
than only the minimizer. The full-epoch correspondence removes this
obstruction, since every coordinate is updated once per transition, and we
show that it contracts uniformly over \emph{all} $n!$ orderings in the energy
norm induced by the Hessian (Proposition~\ref{prop:rpcd-nadler-contraction}). This yields a pathwise
convergence theorem, not merely one in expectation, valid for deterministic,
adaptive, randomized, or even adversarial sequences of orderings
(Corollaries~\ref{cor:rpcd-all-order-convergence}
and~\ref{cor:rpcd-residual-bound}). A two-dimensional example shows that the
choice of metric is not incidental: the same epoch map can strictly
\emph{expand} the Euclidean norm even though every ordering converges to the
unique minimizer (Appendix~\ref{sec:appendix-rpcd}). In summary, the energy metric makes the epoch correspondence a strict Hausdorff contraction, illustrating the role of non-Euclidean metrics in our quantitative converse.

\paragraph{4. Complexity of Large margin triplet stationarity.}
Finally, we show that $\beta$-\textsc{Projected-Nadler} is not only an abstract
characterization but a computational primitive for a concrete, nonconvex
learning objective: finding an approximate first-order stationary point of a
weighted triplet loss.  The triplet loss, introduced in FaceNet
\cite{FaceNetTripletLoss}, is now one of the most prominent contrastive losses. In the large-margin regime,
where the margin $\alpha$ is at least the embedding dimension
$d_{\mathrm{emb}}$, every hinge function (i.e. $x \mapsto \max\{0,x\}$) agrees on the feasible box with its
underlying quadratic, so the objective is globally smooth even though it need
not be convex. We reduce this problem to $\beta$-\textsc{Projected-Nadler}
(Theorem~\ref{thm:large-margin-triplet-to-nadler}) using the singleton
correspondence $F_{\mathcal L}(x)=\{T(x)\}$ induced by the projected-gradient
update $T$, together with an ultrametric built from the normalized loss. This
is the construction from our hardness proof, applied to a concrete objective.  An approximate Nadler fixed point is then a fixed point
of $T$, hence an exact stationary point; a contraction violation identifies
a witness point with insufficient potential decrease, which a standard
projected-gradient estimate converts into an approximate stationary point.
 The remaining regularity outputs are ruled out by explicit Lipschitz
bounds on $T$ and on the metric. Yan et al.~\cite{YanEtAlTriplet} prove $\CLS$-hardness for exact
stationarity. We obtain \textsf{CLS}-hardness at a prescribed accuracy $\varepsilon>0$ by
applying their reduction to approximate
\textsc{QuadraticProgram-KKT}~\cite{Fearnley2025KKT} and converting an
approximate stationary point into an approximate KKT point with one
projected-gradient step. This yields $\CLS$-completeness of
\textsc{Large-Margin-Triplet-FOSP}
(Corollary~\ref{cor:large-margin-triplet-cls-complete}).

\subsection{Organization, acknowledgment, and AI disclosure}
Section~\ref{sec:preliminaries} introduces the notions used throughout.
Section~\ref{sec:nadler-continuous-local-opt} formulates
$\beta$-\textsc{Projected-Nadler}, whose \textsf{CLS}-completeness is proved
in Appendix~\ref{sec:appendix-nadler-complexity}.
Section~\ref{sec:converse-nadler} states the new converse to Nadler's
theorem; Appendix~\ref{sec:appendix-converse-nadler} reviews the
set-theoretic converses surrounding Fryszkowski's problem and gives the
proofs, and Appendix~\ref{sec:appendix-rpcd} applies the converse to
coordinate descent. Section~\ref{sec:large-margin-triplet-nadler}
introduces large-margin triplet stationarity, with proofs in
Appendix~\ref{sec:appendix-triplet-loss}. Appendix~\ref{sec:topological-preliminaries}
collects additional background for Appendices~\ref{sec:appendix-rpcd},
\ref{sec:appendix-nadler-complexity}, and~\ref{sec:appendix-converse-nadler}.

\paragraph{Acknowledgment and AI disclosure.}
The authors thank Ioannis Panageas (UC Irvine) for helpful advice at an early stage of this project. We used generative-AI systems  (ChatGPT 5.6 sol, ChatGPT 6-Astra, Fable 5.1, and Opus 5.5) for polishing the language and exposition, navigating related literature for Sections~\ref{sec:converse-nadler} and~\ref{sec:large-margin-triplet-nadler} and Appendix~\ref{sec:appendix-rpcd}, drawing figures,
and auditing proofs. For mathematical work, we used ChatGPT 6-Astra and Fable 5.1 to help finalize the details based on authors' original proofs, primarily for the proof of Theorem \ref{thm:daskalakis-nadler-converse} in Section \ref{sec:appendix-converse-nadler}. We also used ChatGPT 6-Astra to help with the algebraic and numerical computations, primarily in Section \ref{sec:appendix-rpcd}. The Lemma \ref{lem:equiv-relation} and Proposition \ref{prop:rpcd-explicit-rate} were discovered by ChatGPT 6-Astra as improvements to the authors' existing work. The authors then developed the corresponding proofs with assistance from ChatGPT 6-Astra. The authors verified and take responsibility for every
mathematical statement, proof, and citation.

\section{Preliminaries}
\label{sec:preliminaries}

\paragraph{Notation and computational model.}
For a positive integer $n$, let $[n]=\{1,\ldots,n\}$. We write
$\mathbb{R}_{\geq 0}$ for the nonnegative reals, $\langle\cdot,\cdot\rangle$
for the Euclidean inner product, $\|\cdot\|_p$ for the $\ell_p$ norm, and
$\operatorname{diam}_d(A)=\sup_{x,y\in A}d(x,y)$ for a metric $d$ and a set
$A$. Unless stated otherwise, all numerical parameters are rational and
represented in binary. Functions and metrics are represented by
polynomial-size arithmetic circuits over rational constants and the
operations $\{+,-,\times,\min,\max,>\}$, where the comparison gate returns
one or zero. Structural properties such as metricity and completeness are
treated as promises, while the regularity properties needed by the
reductions admit explicit violation outputs.
Throughout, we use polynomial-time many-one reductions: a reduction maps each
source instance to a target instance in polynomial time and provides a
polynomial-time decoder that maps every valid solution of the target
instance to a valid solution of the source instance. We write
$\mathcal A\leq_p\mathcal B$ for such a reduction and
$\mathcal A\equiv_p\mathcal B$ when reductions hold in both directions.

\paragraph{Well-behaved arithmetic circuits.}
To ensure efficient exact evaluation, we use well-behaved arithmetic
circuits as defined in
\cite[Section~3.1.3]{FearnleyGoldbergHollenderSavani2023}. A multiplication
gate is a \emph{true multiplication} if neither input is a rational constant
gate. Denoting $S$ as the binary description length of a circuit $C$,
including its rational constants, $C$ is \emph{well-behaved} if every
directed path ending at an output contains at most $\log_2 S$ true
multiplication gates. Multiplication by an explicit rational constant does
not count. This syntactic restriction is checkable in polynomial time and rules out
the exponential growth of encoding lengths possible under repeated squaring:
on a rational input $x$, all intermediate values have polynomial encoding
length, and $C(x)$ can be evaluated exactly in time
$\operatorname{poly}(S,\operatorname{size}(x))$
\cite[Lemma~3.3 and Remark~1]{FearnleyGoldbergHollenderSavani2023}.

\paragraph{Metric spaces.}
A metric on a set $X$ is a function
$d\colon X\times X\to\mathbb{R}_{\geq 0}$ such that, for all $x,y,z\in X$,
$d(x,y)=0$ if and only if $x=y$, $d(x,y)=d(y,x)$, and
$d(x,z)\leq d(x,y)+d(y,z)$. It is an ultrametric if it satisfies the
stronger triangle inequality $d(x,z)\leq\max\{d(x,y),d(y,z)\}$. A sequence
$(x_t)$ is $d$-Cauchy if, for every $\varepsilon>0$, all sufficiently late
pairs satisfy $d(x_s,x_t)<\varepsilon$, and $(X,d)$ is complete if every
$d$-Cauchy sequence converges to a point of $X$.

\paragraph{Correspondences and Hausdorff distance.}
A correspondence, or set-valued map, from $X$ to $Y$ is a map
$F\colon X\rightrightarrows Y$ assigning a subset $F(x)\subseteq Y$ to every
$x\in X$. We write $\mathrm{CB}(X)$ for the family of nonempty closed and
bounded subsets of $X$ and $\mathcal K(X)$ for the family of nonempty
compact subsets. For a nonempty set $A\subseteq X$, the point-to-set
distance is $\operatorname{dist}_d(z,A)=\inf_{a\in A}d(z,a)$, and for
nonempty closed and bounded $A,B\subseteq X$ the Hausdorff distance induced
by $d$ is
\[
H_d(A,B)
=
\max\left\{
\sup_{a\in A}\operatorname{dist}_d(a,B),
\sup_{b\in B}\operatorname{dist}_d(b,A)
\right\}.
\]
We also write $d_H$ when the underlying metric is clear. A correspondence
$F\colon X\to\mathcal K(X)$ is Hausdorff-continuous if $d(x_t,x)\to 0$
implies $H_d(F(x_t),F(x))\to 0$, and Hausdorff $L$-Lipschitz if
$H_d(F(x),F(y))\leq Ld(x,y)$ for all $x,y\in X$. A point $x^*$ is a fixed
point of $F$ if $x^*\in F(x^*)$ and an endpoint if $F(x^*)=\{x^*\}$. Every
endpoint is a fixed point, but not conversely.

\paragraph{Contractions and Nadler's theorem.}
A correspondence $F\colon X\rightrightarrows X$ is a Hausdorff
$k$-contraction if $k\in[0,1)$ and $H_d(F(x),F(y))\leq k\,d(x,y)$ for all
$x,y\in X$.

\begin{theorem}[Nadler's fixed-point theorem {\cite{Nadler1969}}]
Let $(X,d)$ be complete and let
$F\colon X\to\mathrm{CB}(X)$ be a Hausdorff $k$-contraction for some
$k<1$. Then $F$ has a fixed point.
\end{theorem}

Unlike the single-valued Banach theorem, Nadler's theorem does not in general
guarantee uniqueness. This distinction motivates the endpoint assumption in
the converse theorem of Section~\ref{sec:converse-nadler}.

\paragraph{Basic iteration procedure for correspondences.}
For a continuous self-map $f\colon X\to X$, the iteration $x_{t+1}=f(x_t)$
generates a unique orbit, and any convergent orbit has a fixed-point limit.
For a correspondence $F\colon X\rightrightarrows X$, each step instead
permits an arbitrary choice $x_{t+1}\in F(x_t)$, so an initial point
generates a family of admissible trajectories, whose exact-time reachable
sets are $F^{[0]}(x)=\{x\}$ and
$F^{[t+1]}(x)=\bigcup_{y\in F^{[t]}(x)}F(y)$. If $F$ has a closed graph, the
limit of any convergent admissible trajectory is a fixed point, but this says
nothing about the other branches, so the natural notion of global convergence
requires $H_d(F^{[t]}(x),\{x^*\})\to 0$ for every initial point $x$.
Moreover, $x^*\in F(x^*)$ does not make $x^*$ absorbing. This requires the
stronger condition that $F(x^*)=\{x^*\}$. A selection such as $P(x)=\Pi_F(x,x)$
restores a deterministic iteration, but describes only one branch of the
full correspondence.

\paragraph{Smoothness, projections, and stationarity.}
Let $K\subseteq\mathbb{R}^q$ be nonempty, closed, and convex. The Euclidean
projection $\Pi_K(z)=\arg\min_{u\in K}\|z-u\|_2$ is single-valued and
nonexpansive. A differentiable function $\mathcal L\colon K\to\mathbb{R}$ is
$\Lambda$-smooth if
$\|\nabla\mathcal L(x)-\nabla\mathcal L(y)\|_2\leq\Lambda\|x-y\|_2$ for all
$x,y\in K$, and its projected-gradient update with step size
$\eta>0$ is $T_\eta(z)=\Pi_K(z-\eta\nabla\mathcal L(z))$.
The normal cone of $K$ at $z\in K$ is
$N_K(z)=\{n\in\mathbb{R}^q:\langle n,u-z\rangle\leq 0\text{ for every }u\in K\}$,
and projection optimality gives $y=\Pi_K(x)$ if and only if
$x-y\in N_K(y)$. A point $z\in K$ is an $\varepsilon$-first-order stationary
point ($\varepsilon$-FOSP) of $\mathcal L$ over $K$ if
\[
\langle u-z,\nabla\mathcal L(z)\rangle
\geq
-\varepsilon
\qquad
\text{for every }u\in K.
\]
For $\varepsilon=0$, this is equivalent to
$0\in\nabla\mathcal L(z)+N_K(z)$; in particular, a fixed point of the
projected-gradient map $z=\Pi_K(z-\eta\nabla\mathcal L(z))$ is an
exact first-order stationary point.

\section{Nadler's Problem and Its Complexity}
\label{sec:nadler-continuous-local-opt}

In this section, we formulate a computational version of Nadler's theorem
and characterize its complexity. Each value of the correspondence is
represented by a projection circuit, paralleling the projection-based
formulation of computational Kakutani fixed points
\cite{PapadimitriouVlatakisGkaragkounisZampetakis2023}, but with the
Hausdorff-Lipschitz condition replaced by the strict contraction required by
Nadler's theorem. We intentionally omit the small-volume output used in
weak-separation-oracle formulations of Kakutani
\cite{PapadimitriouVlatakisGkaragkounisZampetakis2023}: a singleton
correspondence $F(x)=\{f(x)\}$ has $\operatorname{vol}(F(x))=0$ everywhere,
so such an output would solve every singleton instance trivially and
invalidate the singleton construction in our hardness reduction. Our main
result is a polynomial-time equivalence between the computational formulation
of Nadler's theorem and \textsc{Continuous-LocalOpt}, a standard complete
problem for \textsf{CLS}. We take $X=[0,1]^n$ as the domain; this choice is
inessential, and the usual variants over rational polytopes and
high-dimensional boxes do not affect our arguments.

\begin{definition}[$\beta$-\textsc{Projected-Nadler}]
\label{def:projected-nadler}
An instance consists of well-behaved circuits $\Pi,d$, where $\Pi(x,y) \in F(x)$ is intended as an (approximate) nearest point projection from $y$ to $F(x)$ and $d: X \times X \rightarrow [0,D]$ is intended as the metric, and rational parameters
$\varepsilon_1,\varepsilon_2,L_P,L_d,D>0,$ and $k \in (0,1)$. We promise that $0 < \varepsilon_2 < \varepsilon_1(1-k)$ and $(X, d)$ is a complete metric space.
Set $P(x)=\Pi(x,x)$ and $r(x)=d(x,P(x))$. A valid output is one of the following:
\begin{outputlist}{(N4)}
\item[\textnormal{(N1)}] A point $x\in X$ such that
$r(x) \leq\varepsilon_1$.

\item[\textnormal{(N2)}] Points $x,y,z\in X$ such that, for
$u=\Pi(x,z), v=\Pi(y,u)$ or $u=\Pi(y,z), v=\Pi(x,u)$,
\begin{equation}
 d(u,v)>k d(x,y)+\varepsilon_2.
\label{eq:pn-buffered-output}
\end{equation}

\item[\textnormal{(N3)}] Points $x,y\in X$ such that
$\|P(x)-P(y)\|_1>L_P\|x-y\|_1^\beta$.

\item[\textnormal{(N4)}] Points $a,b,a',b'\in X$, with $a\neq b$ and $a'\neq b'$, such that
\begin{equation}
|d(a,b)-d(a',b')|
>L_d\bigl(\|a-a'\|_1+\|b-b'\|_1\bigr).
\label{eq:pn-distance-violation}
\end{equation}    
\end{outputlist}
\end{definition}

\begin{remark}
\label{rem:pn-promises}
The numerical promise $0<\varepsilon_2<\varepsilon_1(1-k)$ can be checked in
polynomial time; therefore it can be considered part of the input format. The other promises (that $d$ is a metric and $(X, d)$ is complete) are what connect a solution to Nadler's theorem, but the totality of our problem does
not depend on them. The only promise (other than the numerical one) used by our $\CLS$-membership proof 
(Lemma~\ref{lem:nadler-to-holder-localopt}) is that $d(x,x)=0$, and this can actually be avoided (see Remark~\ref{rem:pn-diagonal} for details). Conversely, the
instances built in our hardness reduction satisfy every promise.
\end{remark}

The circuits evaluate their own rational outputs exactly on rational inputs;
approximation refers to their agreement with an underlying precise metric
or the nearest-point projection map. For example, an exact Euclidean projection can be
irrational even at a rational query and need not be representable by the
allowed rational gates.

\paragraph{The contraction certificate.}
An output of type \textnormal{(N2)} can be checked in polynomial time by
evaluating $\Pi$ and $d$ at the rational witness points. At a practical level, we consider the circuit $d$ itself as the metric and $\Pi$ as an $\eta$-approximate nearest-point circuit, which means
$\Pi(w,t)\in F(w)$ and
$d(t,\Pi(w,t))\le\operatorname{dist}_d(t,F(w))+\eta$ for all $w,t\in X$.
For $u=\Pi(x,z)\in F(x)$ and $v=\Pi(y,u)$ (the other orientation is
symmetric),
\[
H_d\bigl(F(x),F(y)\bigr)
\ \ge\ \operatorname{dist}_d\bigl(u,F(y)\bigr)
\ \ge\ d(u,v)-\eta
\ >\ k\,d(x,y)+\varepsilon_2-\eta,
\]
so whenever $\eta\le\varepsilon_2$ an \textnormal{(N2)} output certifies
$H_d(F(x),F(y))>k\,d(x,y)$. In the reverse direction, the same estimate shows that a
genuine $k$-contraction admits no \textnormal{(N2)} witness at all.

More generally, these input circuits may only approximate an underlying geometry.
Let $\rho$ be a metric on $X$ with respect to which $F$ has nonempty closed
and bounded values, and suppose that $|d(a,b)-\rho(a,b)|\le\nu$ for all
$a,b\in X$, while every response $q=\Pi(w,t)$ satisfies
$\operatorname{dist}_{\rho}(q,F(w))\le\tau$ and
$\rho(t,q)\le\operatorname{dist}_{\rho}(t,F(w))+\eta$. The same argument,
now absorbing both approximation errors, gives
\begin{align*}
    H_{\rho}\bigl(F(x),F(y)\bigr)
    &\geq \operatorname{dist}_{\rho}\bigl(u,F(y)\bigr)-\tau
    \geq \rho(u,v)-\tau-\eta\\
    &\geq d(u,v)-\tau-\eta-\nu
    > k\,\rho(x,y)+\varepsilon_2-\tau-\eta-(1+k)\nu,
\end{align*}
so \textnormal{(N2)} certifies a genuine violation of Hausdorff contraction
as long as $\tau+\eta+(1+k)\nu\le\varepsilon_2$. The previous practical consideration
is the case where $\rho=d$ and $\tau=\nu=0$, and exact projection gives
$\eta=0$ as well. Two caveats remain. An approximate response need not
satisfy $\Pi(x,u)=u$ for $u\in F(x)$, and because of the positive tolerance
$\varepsilon_2$, the witness format does not claim to detect every strict
violation of Hausdorff contraction. Moreover, these accuracy assumptions are
hypotheses beyond Definition~\ref{def:projected-nadler}: without them,
\textnormal{(N2)} reflects a contraction violation only up to the additive
slack $\tau+\eta+(1+k)\nu-\varepsilon_2$. Promising the bound
$\tau+\eta+(1+k)\nu\le\varepsilon_2$ does not affect the complexity of the
problem, since every instance produced by our reductions uses exact
circuits, but it narrows the admissible inputs: together with
$\varepsilon_2<(1-k)\varepsilon_1$ it forces
$\tau+\eta+(1+k)\nu<(1-k)\varepsilon_1$, so the circuits' inaccuracy must be bounded by $(1-k)\varepsilon_1$.

\paragraph{The regularity certificates.}
The outputs \textnormal{(N3)} and \textnormal{(N4)} expose failures of the
claimed $\beta$-H\"older regularity of $P$ and off-diagonal Lipschitz
regularity of $d$, respectively. Beyond continuity failures, they also
certify irregular approximation error of the circuits. For example, let
$X=[0,1]$ and $F(w)=[0,1/4]$, whose exact projection is
$\Pi_*(w,z)=\min\{z,1/4\}$, and for a rational $0<\eta<1/4$ define
$\Pi(w,z)=\max\{0,\min\{z,1/4\}-\eta\mathbf{1}[w>3/4]\}$. This circuit
always returns a point in $F(w)$ and satisfies
$|\Pi(w,z)-\Pi_*(w,z)|\le\eta$, but its diagonal map $P(x)=\Pi(x,x)$ jumps
by $\eta$ at $x=3/4$: for $x=3/4$ and $y=3/4+\eta^2/4$, we have
$|P(x)-P(y)|=\eta>\eta/2=|x-y|^{1/2}$, so $(x,y)$ is a valid
\textnormal{(N3)} answer for $\beta=1/2$ and $L_P=1$, even though the exact
diagonal projection is $1$-Lipschitz. The regular approximation
$P_{\mathrm{reg}}(x)=\min\{x,1/4-\eta\}$, by contrast, is $1$-Lipschitz with
the same maximum pointwise error $\eta$ (Figure~\ref{fig:reg-vs-irreg}).
These are properties of the implemented maps; an \textnormal{(N3)} answer
need not identify a defect of the underlying correspondence.

\begin{figure}[t]
    \centering
    \includegraphics[width=0.62\linewidth]{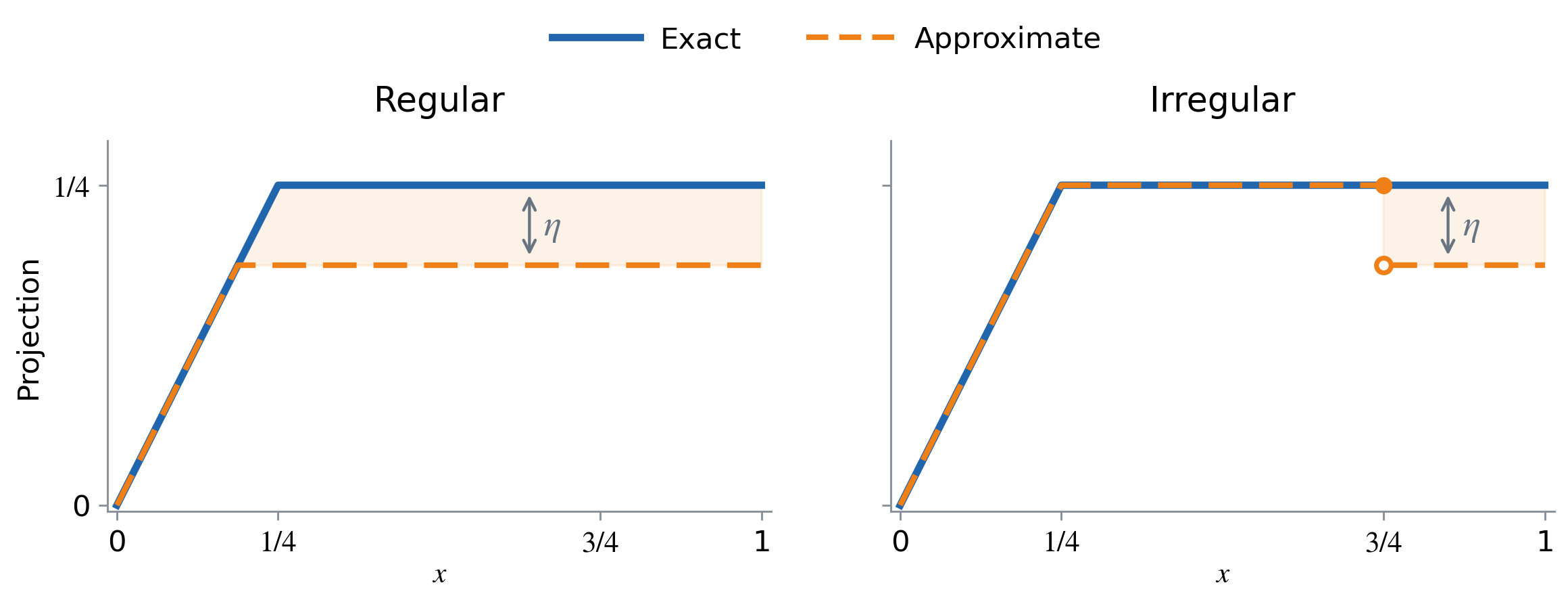}
    \caption{%
        Regular (left) and irregular (right) projection approximations. Blue solid curves represent the exact projection, and orange dashed curves represent its approximations. Both have maximum pointwise error $\eta$, but the regular approximation is Lipschitz, whereas the irregular approximation has a jump discontinuity.
    }
    \label{fig:reg-vs-irreg}
\end{figure}

\paragraph{Residual decrease.}
The residual decrease follows from condition \textnormal{(N2)}. At a point
$x$, test the ordered triple $(x,P(x),x)$ in the first orientation, which
gives $u=P(x)$ and $v=P(P(x))$. Either this is an \textnormal{(N2)} answer,
or $r(P(x))\leq k r(x)+\varepsilon_2$, and hence
\[
r(x)-r(P(x))\geq(1-k)r(x)-\varepsilon_2.
\]
On the second branch, $r(x)>\varepsilon_1$ implies a strictly positive
decrease because $\varepsilon_2<(1-k)\varepsilon_1$. The decrease is enforced by the certificate, not derived from Hausdorff
contraction; when $\Pi$ is exact and $d$ is the contracting metric,
contraction gives $r(P(x))\leq k\,r(x)$ directly. As $k$ approaches $1$, the guaranteed decrease
shrinks, so the same target accuracy $\varepsilon_1$ requires a smaller
additive tolerance. This is compatible with finite-precision computation: for
every rational $k<1$, the choice $\varepsilon_2=(1-k)\varepsilon_1/2$ has
binary encoding length polynomial in those of $k$ and $\varepsilon_1$, so
although $\varepsilon_2$ tends to zero as $k\uparrow 1$, every individual
instance uses a strictly positive, finitely represented tolerance.

\paragraph{Continuous local search.}
To determine the complexity of $\beta$-\textsc{Projected-Nadler}, we recall
the standard complete problem for $\CLS$
\cite{DaskalakisTzamosZampetakis2018}.

\begin{definition}[\textsc{Continuous-LocalOpt}]
\label{def:continuous-local-opt}
An instance consists of well-behaved arithmetic circuits
$f\colon[0,1]^3\to[0,1]^3$ and $p\colon[0,1]^3\to[0,1]$, together with
positive rational parameters $\delta$ and $\lambda$. A valid output is one of the following.
\\
\begin{outputlist}{(CO3)}
    \item[\textnormal{(CO1)}] A point $x\in[0,1]^3$ satisfying
    $p(f(x))\geq p(x)-\delta$.

    \item[\textnormal{(CO2)}] Points $x,y\in[0,1]^3$ satisfying
    $\|f(x)-f(y)\|_1>\lambda\|x-y\|_1$.

    \item[\textnormal{(CO3)}] Points $x,y\in[0,1]^3$ satisfying
    $|p(x)-p(y)|>\lambda\|x-y\|_1$.
\end{outputlist}
\end{definition}

The class $\CLS$ consists of the total search problems that admit a
polynomial-time many-one reduction to \textsc{Continuous-LocalOpt}. As
discussed in \cite{DaskalakisPapadimitriou2011}, the choice of $[0,1]^3$ and
$\ell_1$ norms is inessential: high-dimensional polytopes (e.g., $[0,1]^n$) and other $\ell_p$
norms can be used without affecting the complexity. Since the
$\beta$-H\"older condition \textnormal{(N3)} could fail to recover the
Lipschitz conditions \textnormal{(CO2)} or \textnormal{(CO3)} in a
reduction, we also define $\beta$-\textsc{H\"older-LocalOpt}, which replaces
the Lipschitz conditions of \textsc{Continuous-LocalOpt} by H\"older
conditions (e.g., $\|f(x)-f(y)\|_1>\lambda\|x-y\|^\beta_1$), and show that
it is \textsf{CLS}-complete for every fixed rational $\beta\in(0,1]$ by
proving its polynomial-time equivalence to \textsc{Continuous-LocalOpt}. A
reduction from \textsc{Continuous-LocalOpt} and a reduction to
$\beta$-\textsc{H\"older-LocalOpt} then yield the following theorem.

\begin{theorem}
\label{thm:nadler-clo-equivalence}
The problem $\beta$-\textsc{Projected-Nadler} is $\CLS$-complete for any fixed rational $\beta \in (0,1]$.
\end{theorem}

Our hardness reduction follows the proof ideas of the reductions for
\textsc{Banach} in \cite{DaskalakisTzamosZampetakis2018}, with modifications
that accommodate the different problem formulation. The details and proofs
are given in Appendix~\ref{sec:appendix-nadler-complexity}.

\section{Converse Nadler's Fixed Point Theorem}
\label{sec:converse-nadler}

Classical converses to Banach's contraction principle ask when a self-map
can be made contractive by replacing the ambient metric. Fryszkowski posed
the corresponding remetrization problem for set-valued maps, and
Com\u{a}neci and Luchian answered it at a set-theoretic level: under an
endpoint assumption, a complete bounded contracting metric exists precisely when a
   bounded Lyapunov-type rank function decreasing along every admissible successor exists
   and, equivalently, when the nested global images of the correspondence
   collapse to the endpoint \cite{Jachymski2000,Comaneci2017,Luchian2018}.
Appendix~\ref{subsec:fryszkowski} surveys these results. However, these remetrization results do not begin
with a reference metric, and the constructed metric need neither generate a
prescribed topology nor be comparable with any pre-existing distance, so
convergence in the remetrized space need not translate into convergence in
the metric relevant to an iterative algorithm. This section proves a stronger
converse that requires both topological equivalence and quantitative error
transfer, in the spirit of the single-valued result of Daskalakis, Tzamos,
and Zampetakis \cite{DaskalakisTzamosZampetakis2018}. For a correspondence,
the correct analogue of global convergence must control all admissible
choices, not merely one selected trajectory.

Let $(X,d)$ be compact and let $F:X\to\mathcal{K}(X)$ be
Hausdorff-continuous. Define $F^{[0]}(x):=\{x\}$ and
$F^{[t+1]}(x):=\widehat F(F^{[t]}(x))$, where
$\widehat F(A):=\bigcup_{x\in A}F(x)$. Thus $F^{[t]}(x)$ is the set of all
states reachable after exactly $t$ admissible updates from $x$. We say that
$F$ is \emph{robustly globally convergent} to an endpoint $x^*$ if
\begin{equation}
    \lim_{t\to\infty}
    H_d\!\left(F^{[t]}(x),\{x^*\}\right)=0
    \qquad\text{for every }x\in X.
    \label{eq:pointwise-robust-convergence}
\end{equation}
We additionally require local uniformity: there is an open neighborhood
$U$ of $x^*$ such that
\begin{equation}
    \lim_{t\to\infty}
    \sup_{x\in U}
    H_d\!\left(F^{[t]}(x),\{x^*\}\right)=0.
    \label{eq:local-uniform-robust-convergence}
\end{equation}
Condition \eqref{eq:pointwise-robust-convergence} says that every admissible
trajectory converges uniformly over all choices made at each time, which is
strictly stronger than the existence of one convergent selection. Under the
compactness and Hausdorff-continuity assumptions above,
Lemma~\ref{lem:equiv-relation} shows that
\eqref{eq:pointwise-robust-convergence} and
\eqref{eq:local-uniform-robust-convergence}, together with the endpoint
condition, are equivalent to \emph{global-image collapse},
$\bigcap_{t\ge0}\widehat F^{\,t}(X)=\{x^*\}$, and also to convergence that is
uniform over all initial states; the proof appears in
Appendix~\ref{subsec:condition-equivalence}. For related discussions of
set-valued topological contraction and endpoints, see \cite{TARAFDAR199579}.

\begin{theorem}[strong converse to Nadler's theorem]
\label{thm:daskalakis-nadler-converse}
Let $(X,d)$ be a compact metric space and let
$F:X\to\mathcal{K}(X)$ be continuous with respect to the Hausdorff metric
$H_d$. Suppose that $F$ has an endpoint $x^*$ and satisfies
\eqref{eq:pointwise-robust-convergence} and
\eqref{eq:local-uniform-robust-convergence}. Then, for every contraction
factor $c\in(0,1)$ and every accuracy scale $\varepsilon>0$, there exists a
metric $D_{c,\varepsilon}$ on $X$ with the following properties.
\begin{enumerate}
    \item The metrics $D_{c,\varepsilon}$ and $d$ generate the same topology.
    In particular, $(X,D_{c,\varepsilon})$ is compact and complete.

    \item The correspondence is a Hausdorff $c$-contraction:
    \begin{equation}
        H_{D_{c,\varepsilon}}(F(x),F(y))
        \leq
        cD_{c,\varepsilon}(x,y)
        \qquad\text{for all }x,y\in X.
        \label{eq:constructed-nadler-contraction}
    \end{equation}

    \item Approximation in the constructed metric transfers to the original
    metric:
    \begin{align}
        D_{c,\varepsilon}(x,x^*)\leq\varepsilon
        &\quad\Longrightarrow\quad
        d(x,x^*)\leq2\varepsilon,
        \label{eq:endpoint-transfer}\\
        D_{c,\varepsilon}(x,y)\leq\varepsilon
        &\quad\Longrightarrow\quad
        \min\!\left\{
            d(x,y),d(x,x^*),d(y,x^*)
        \right\}
        \leq2\varepsilon.
        \label{eq:pairwise-transfer}
    \end{align}
\end{enumerate}
\end{theorem}

The proof, given in Appendix~\ref{sec:appendix-converse-nadler}, follows the
structure of \cite[Theorem 1]{DaskalakisTzamosZampetakis2018} and Meyers'
theorem \cite{Meyers1967}. The theorem immediately yields convergence and
residual bounds analogous to
\cite[Corollary 1]{DaskalakisTzamosZampetakis2018}.

\begin{corollary}[Global convergence and residual error bound]
\label{cor:robust-global-convergence}
Under the assumptions of
Theorem~\ref{thm:daskalakis-nadler-converse}, let
$D=D_{c,\varepsilon}$ and define
$r_D(x):=\operatorname{dist}_D(x,F(x))$. Then:
\begin{enumerate}
    \item $x^*$ is the unique fixed point of $F$.

    \item Every admissible trajectory $x_{t+1}\in F(x_t)$ satisfies
    \begin{equation}
        D(x_t,x^*)
        \leq
        c^tD(x_0,x^*).
        \label{eq:trajectory-linear-convergence}
    \end{equation}

    \item The residual is equivalent, up to constants depending only on
    $c$, to the distance from the endpoint:
    \begin{equation}
        (1-c)D(x,x^*)
        \leq
        r_D(x)
        \leq
        (1+c)D(x,x^*).
        \label{eq:residual-error-bound}
    \end{equation}

    \item Whenever $r_D(x_0)>0$, every integer
    \begin{equation}
        t
        \geq
        \max\left\{
            0,
                \frac{
                    \log\!\left(
                        r_D(x_0)/((1-c)\varepsilon)
                    \right)
                }{
                    \log(1/c)
                }
        \right\}
        \label{eq:iteration-complexity-bound}
    \end{equation}
    guarantees $d(x_t,x^*)\leq2\varepsilon$.
\end{enumerate}
\end{corollary}

As an example, Appendix~\ref{sec:appendix-rpcd} studies random-permutation
cyclic coordinate descent (RPCD) on strongly convex quadratic objectives.
With $G_\pi$ denoting one complete coordinate sweep under permutation $\pi$
and $\mathfrak S_n$ denoting the set of permutations over $[n]$, we prove directly that
$\widetilde F(x)=\{G_\pi(x):\pi\in\mathfrak S_n\}$ is a strict Hausdorff
contraction in the energy metric induced by the quadratic objective. The
resulting convergence bound holds pathwise for every sequence of coordinate
orderings, including adaptive or adversarial choices, and on a compact
invariant energy sublevel set the example also satisfies the hypotheses of
Theorem~\ref{thm:daskalakis-nadler-converse}.

\section{The Complexity of Large-Margin Triplet Loss Minimization}
\label{sec:large-margin-triplet-nadler}

In this section, we construct a reduction from large-margin triplet-loss
stationarity to $\beta$-\textsc{Projected-Nadler} over a
polynomial-dimensional box. We begin with a formal introduction of the
problem. Let $d_{\mathrm{emb}}$ be a positive integer,
$V \subset \mathbb{R}^{d_{\mathrm{emb}}}$ be a finite set of movable points,
and $A, B$ be two fixed pivots with rational embeddings
$a_A,a_B\in[0,1]^{d_{\mathrm{emb}}}$. Each movable point $v\in V$ is assigned
to an embedding $z_v\in[0,1]^{d_{\mathrm{emb}}}$, and the movable embeddings are
collected into a vector $z\in K:=[0,1]^q$, where $q=d_{\mathrm{emb}}|V|$.
For notational uniformity, set $z_A:=a_A$ and $z_B:=a_B$. These two vectors
are fixed and are not optimization variables. Let $\mathcal C$ be a
collection of ordered triplets $t=(i_t,j_t,k_t)$ over $V\cup\{A,B\}$, each
specifying that the anchor $i_t$ should be closer to $j_t$ than to $k_t$. Each triplet carries a rational weight $w_t \ge 0$ and a rational margin $\alpha>0$. Defining $[a]_+:=\max\{a,0\}$, the weighted triplet loss is
\begin{equation}
    \mathcal L(z)
    :=
    \sum_{t\in\mathcal C} w_t
    \left[
        \|z_{i_t}-z_{j_t}\|_2^2
        -
        \|z_{i_t}-z_{k_t}\|_2^2
        +
        \alpha
    \right]_+.
    \label{eq:large-margin-triplet-loss}
\end{equation}
The hinge function $[\,\cdot\,]_+$ for triplet $t$ vanishes precisely when the
squared distance from $i_t$ to $k_t$ exceeds the squared distance from $i_t$
to $j_t$ by at least $\alpha$. We focus on the large-margin regime
$\alpha \ge d_{\mathrm{emb}}$, which allows us to treat the loss on the
feasible box as a smooth function. Since all embeddings lie in
$[0,1]^{d_{\mathrm{emb}}}$, for every $z\in K$ and every triplet $t$,
\[
\Delta_t(z)
:= \|z_{i_t}-z_{j_t}\|_2^2-\|z_{i_t}-z_{k_t}\|_2^2+\alpha
\ge \alpha-d_{\mathrm{emb}}\ge 0 .
\]
Thus $[\Delta_t(z)]_+=\Delta_t(z)$ on $K$, and $\mathcal{L}$ agrees there
with the quadratic polynomial
$\widetilde{\mathcal L}(z):=\sum_{t\in\mathcal C}w_t\Delta_t(z)$. In the
stationarity condition below, $\nabla\mathcal L$ is the gradient of $\widetilde{\mathcal L}$, including at the boundary points of $K$. This gradient
is canonical, since any two continuously differentiable extensions agreeing
with $\mathcal L$ on $K$ have identical gradients on its interior and, by continuity, on its boundary. It determines the first-order variation of
$\mathcal L$ along every feasible direction. The convention also covers
$\alpha=d_{\mathrm{emb}}$. When $\alpha=d_{\mathrm{emb}}$,  $\Delta_t(z)$ can be zero at some boundary points of $K$, so the
hinge term $[\Delta_t(z)]_+$ is not differentiable at those points as a function
on $\mathbb{R}^q$. Nevertheless, this causes no difficulty: the hinge term coincides
with the polynomial $\Delta_t$ on $K$, and only the smoothness of $\mathcal L$ restricted to $K$ is ever used.

\begin{definition}[\textsc{Large-Margin-Triplet-FOSP}]
\label{def:large-margin-triplet-fosp}
An instance consists of the sets and embeddings described above, the
triplet collection $\mathcal C$ and its rational weights, a rational margin
$\alpha \ge d_{\mathrm{emb}}$, and a rational accuracy parameter
$\varepsilon>0$. The task is to output a point $z^*\in K$ such that
\begin{equation}
    \langle u-z^*,\nabla\mathcal L(z^*)\rangle
    \geq
    -\varepsilon
    \qquad
    \text{for every }u\in K.
    \label{eq:triplet-fosp-definition}
\end{equation}
Such a point is called an $\varepsilon$-first-order stationary point, or
$\varepsilon$-FOSP. For $\varepsilon=0$, condition
\eqref{eq:triplet-fosp-definition} is the usual variational-inequality
formulation of first-order stationarity over the box $K$.
If $V=\varnothing$, then $q=0$ and $K=[0,1]^0$ consists of the empty
tuple. Its variational-inequality gap is the empty sum, so the empty tuple
is an exact FOSP.
\end{definition}

Yan et al. prove \textsf{CLS}-hardness of exact FOSP computation for triplet loss minimization through a
reduction from exact \textsc{QuadraticProgram-KKT} \cite{YanEtAlTriplet}.
That result does not, by itself, give hardness at a prescribed positive accuracy,
so we instead reduce from the approximate \textsc{QuadraticProgram-KKT} problem,
which is also \textsf{CLS}-complete \cite{Fearnley2025KKT}. The following recalls the definition of the approximate \textsc{QuadraticProgram-KKT} problem.

\begin{definition}[Approximate \textsc{QuadraticProgram-KKT}]
\label{def:quadratic-program-kkt}
An instance consists of a symmetric matrix
$Q \in \mathbb{Q}^{n \times n}$, a vector
$b \in \mathbb{Q}^{n}$ and a rational accuracy parameter $\varepsilon > 0$. These coefficients define the quadratic program
\begin{equation}
    \min_{x \in [0,1]^n}
    p(x)
    =
    x^\top Qx + b^\top x.
    \label{eq:quadratic-program-kkt}
\end{equation}
The goal is to output an $\varepsilon$-KKT point $x^* \in [0,1]^n$ for \eqref{eq:quadratic-program-kkt}. A point $x\in [0,1]^n$ is an $\varepsilon$-KKT point for \eqref{eq:quadratic-program-kkt} if $x_i > 0$ implies $
\frac{\partial p}{\partial x_i}(x) \le \varepsilon$ and $x_i < 1$ implies $\frac{\partial p}{\partial x_i}(x) \ge -\varepsilon$.
\end{definition}

We state the following complexity results for
\textsc{Large-Margin-Triplet-FOSP}, with proofs given in full in
Appendix~\ref{sec:appendix-triplet-loss}.

\begin{theorem}
\label{thm:large-margin-triplet-to-nadler}
There is a polynomial-time many-one reduction from \textsc{Large-Margin-Triplet-FOSP} to $\beta$-\textsc{Projected-Nadler} with fixed parameter $\beta=1$.
\end{theorem}

Theorem~\ref{thm:large-margin-triplet-to-nadler} places
\textsc{Large-Margin-Triplet-FOSP} in \textsf{CLS}, via the
\textsf{CLS}-membership of $\beta$-\textsc{Projected-Nadler} established
earlier. Combined with the argument of \cite[Theorem 4.1]{YanEtAlTriplet}
for \textsf{CLS}-hardness of the exact triplet-loss stationarity problem,
this yields the following corollary.

\begin{corollary}
\label{cor:large-margin-triplet-cls-complete}
The problem \textsc{Large-Margin-Triplet-FOSP} is $\mathsf{CLS}$-complete.
\end{corollary}

\section{Conclusion and Further Directions}
\label{sec:conclusion}

We formulated a computational version of Nadler's fixed-point theorem,
showed it is $\CLS$-complete even under the H\"older-regularity relaxation, proved a quantitative converse that certifies
robust convergence of a set-valued iteration by an explicit change of metric,
and instantiated both directions on random-permutation coordinate descent and
large-margin triplet loss. The following questions remain open.

\begin{enumerate}
\item \textbf{Weak projection oracles.} Our formulation lets the projection
circuit be inexact in distance, but its responses always lie in $F(x)$. Under
weak separation oracle access of a convex compact domain, as for the computational formulation of Kakutani
\cite{PapadimitriouVlatakisGkaragkounisZampetakis2023}, a response is only
guaranteed to lie in the outer $\gamma$-neighborhood
$\overline{B}(F(x),\gamma)=\bigcup_{z\in F(x)}\{y:d(y,z)\le\gamma\}$. Does the
$\CLS$-completeness of $\beta$-\textsc{Projected-Nadler} survive when feasibility
holds only up to such a tolerance $\gamma$ in the weak separation oracle model?

    \item \textbf{Beyond Hausdorff contractions.} Generalized, \'Ciri\'c-type,
    and Meir--Keeler-type multivalued contractions weaken
    $H_d(F(x),F(y))\le k\,d(x,y)$ to asymptotic or case-by-case decrease.
    Fryszkowski's remetrization results suggest that some generality survives
    set-theoretically \cite{Comaneci2017,Luchian2018}, but since our
    reduction needs \emph{uniform} geometric decrease of the residual, it is
    open which weakenings preserve $\CLS$-membership and which push toward
    $\mathsf{PPAD}$ or full $\mathsf{FIXP}$.

    \item \textbf{Set-valued equilibrium computation.} Best-response maps
    and quasi-variational inequalities with point-dependent feasible regions
    \cite{VI_and_NE,GQVI_chan,complexity_of_QVI,quasivariational-inequalities-local}
    motivated correspondences in our introduction but were not pursued
    computationally. Under what proximal or Tikhonov-type regularization does
    a generalized Nash best-response correspondence become a genuine
    Hausdorff contraction, and is the resulting equilibrium $\CLS$-complete
    via $\beta$-\textsc{Projected-Nadler}, a contraction-based counterpart
    to $\mathsf{PPAD}$-completeness in the convex-valued case
    \cite{PapadimitriouVlatakisGkaragkounisZampetakis2023}?

    \item \textbf{Beyond triplet loss.} The reduction needs only two facts:
    large margin forces global smoothness, and one projected-gradient step
    supplies potential decrease. Quadruplet losses, margin-based pairwise
    ranking losses, and localized softmax contrastive losses
    \cite{YanEtAlTriplet} share both, so the same
    singleton-correspondence/potential-ultrametric construction should transfer; the open task is to find the weakest margin condition for each objective.

\end{enumerate}

\newpage
\printbibliography
\newpage
\appendix
\section{Additional Preliminaries}
\label{sec:topological-preliminaries}

Here, we introduce additional topological notions and concepts for later proofs.

\paragraph{Topological spaces.}
A \emph{topology} on a set $X$ is a collection
$\tau\subseteq\mathcal{P}(X)$ such that $\varnothing,X\in\tau$, arbitrary
unions of sets in $\tau$ belong to $\tau$, and finite intersections of sets
in $\tau$ belong to $\tau$. The pair $(X,\tau)$ is called a
\emph{topological space}, and the elements of $\tau$ are called open sets.
A set is closed if its complement is open. A neighborhood of $x\in X$ is a
set containing an open set that contains $x$. For $A\subseteq X$, its closure
$\overline{A}$ is the smallest closed set containing $A$.

A topological space is \emph{Hausdorff} if any two distinct points admit
disjoint open neighborhoods. A bijection $f\colon X\to Y$ is a
\emph{homeomorphism} if both $f$ and $f^{-1}$ are continuous. In this case,
$X$ and $Y$ have the same topological structure.

\paragraph{Metrics and pseudometrics.}
A \emph{pseudometric} on $X$ is a function
$p\colon X\times X\to\mathbb{R}_{\geq 0}$ satisfying
$p(x,x)=0$, symmetry, and the triangle inequality. Unlike a metric, a
pseudometric may satisfy $p(x,y)=0$ for distinct points $x\neq y$. Thus, a
metric is precisely a pseudometric that separates distinct points.

Every metric $d$ on $X$ induces a topology whose open sets are unions of
open balls $B_d(x,r):=\{y\in X:d(x,y)<r\}$. Two metrics $d$ and $D$ on the
same set are \emph{topologically equivalent} if they induce the same
topology. Equivalently, the identity map
$\operatorname{id}\colon (X,d)\to(X,D)$ is continuous and has a continuous inverse. That is, it is a homeomorphism.

A metric or pseudometric $p$ is said to be continuous with respect to a
given topology on $X$ if $(x,y)\mapsto p(x,y)$ is continuous on $X\times X$
with its product topology. When $X$ is metrized by $d$, the product topology
is induced, for example, by
\[
d_{\times}\bigl((x,y),(x',y')\bigr)
    :=d(x,x')+d(y,y').
\]
If $g\colon X\to(Z,\delta)$ is any map into a metric space, then
$p_g(x,y):=\delta(g(x),g(y))$ is a pseudometric on $X$. This observation
will be applied with $Z$ equal to a hyperspace of compact sets.

\paragraph{Continuity and compactness.}
A map $f\colon X\to Y$ between topological spaces is continuous if
$f^{-1}(V)$ is open in $X$ for every open set $V\subseteq Y$. For metric
spaces $(X,d)$ and $(Y,p)$, continuity at $x\in X$ is equivalently
expressed as follows: for every $\varepsilon>0$, there exists $\delta>0$ such
that
\[
d(x,y)<\delta
    \quad\Longrightarrow\quad
p\bigl(f(x),f(y)\bigr)<\varepsilon.
\]

It is also equivalent to sequential continuity:
$x_n\to x$ implies $f(x_n)\to f(x)$. The map is uniformly continuous if
$\delta$ may be chosen independently of $x$.

A topological space is compact if every open cover has a finite subcover.
We use the following standard facts. Every compact metric space is
sequentially compact and complete; every continuous map from a compact
metric space to a metric space is uniformly continuous; and every
continuous real-valued function on a compact space attains its minimum and
maximum. Moreover, a continuous bijection from a compact space to a
Hausdorff space is a homeomorphism. Metric spaces are Hausdorff and regular;
in particular, if $x\in U$ and $U$ is open, then there exists an open set
$V$ such that $x\in V$ and $\overline{V}\subseteq U$.

\paragraph{Hausdorff distance and Hausdorff continuity.}
Let $(Y,\delta)$ be a metric space, and let $\mathcal{K}(Y)$ denote the
family of nonempty compact subsets of $Y$. For $z\in Y$ and a nonempty set
$A\subseteq Y$, define
$\operatorname{dist}_{\delta}(z,A):=\inf_{a\in A}\delta(z,a)$. For
$A,B\in\mathcal{K}(Y)$, recall the notion of Hausdorff distance:
\[
H_{\delta}(A,B)
:=
\max\left\{
    \sup_{a\in A}\operatorname{dist}_{\delta}(a,B),
    \sup_{b\in B}\operatorname{dist}_{\delta}(b,A)
\right\}.
\]
The function $H_{\delta}$ is a metric on $\mathcal{K}(Y)$. Compactness also
ensures that each point-to-set infimum above is attained. Let $(X,d)$ be another metric space. A compact-valued correspondence
$F\colon X\rightrightarrows Y$ is \emph{Hausdorff-continuous} if the
ordinary map $x\longmapsto F(x)$
from $(X,d)$ to $(\mathcal{K}(Y),H_{\delta})$ is continuous. Equivalently,
whenever $x_n\to x$ in $d$, one has
$H_{\delta}(F(x_n),F(x))\to 0$. It is \emph{Hausdorff
$L$-Lipschitz} if
$H_{\delta}(F(x),F(y))\leq Ld(x,y)$ for all $x,y\in X$.

For a sequence of compact-valued correspondences $F_t$, we say that
$F_t$ converges uniformly to $F$ on a set $E\subseteq X$ if
\[
\sup_{x\in E}H_{\delta}\bigl(F_t(x),F(x)\bigr)\longrightarrow 0.
\]
A set $W\subseteq X$ is \emph{forward invariant} under
$F\colon X\rightrightarrows X$ if $F(x)\subseteq W$ for every $x\in W$.


\section{Example: Random-Permutation Cyclic Coordinate Descent}
\label{sec:appendix-rpcd}
The strong converse proved in the previous section shows that robust convergence of all branches of a set-valued iteration can be certified by a suitable contracting metric. We illustrate this principle using random-permutation cyclic coordinate descent (RPCD), a standard variant of coordinate descent in which a fresh permutation of the coordinates is selected at the beginning of each epoch. Various research studies have been conducted to analyze the convergence behavior of RPCD \cite{ChangHsiehLin2008, OswaldZhou2017, WrightLee2019, WrightLee2020, SunLuoYe2020, GurbuzbalabanOzdaglarVanliWright2020, KimLeeYun2025}. For strongly convex quadratic objectives, this method is equivalent to applying the Gauss--Seidel method with a randomly chosen coordinate ordering.

The randomness itself will play no role in our analysis. Instead, we discard
the probability distribution and retain the set of all possible outputs of one epoch. This produces a finite-valued correspondence. A trajectory of this correspondence represents coordinate descent under an arbitrary and possibly adversarial sequence of permutations.

\paragraph{The quadratic problem and the epoch correspondence.}

Let $A\in\mathbb{R}^{n\times n}$ be symmetric and positive definite, let
$b\in\mathbb{R}^{n}$, and consider the strongly convex quadratic objective
\begin{equation}
    \Phi(x)
    :=
    \frac{1}{2}x^\top A x-b^\top x.
    \label{eq:rpcd-quadratic-objective}
\end{equation}
Its unique minimizer is
$x^*=A^{-1}b$. Moreover,
\begin{equation}
    \Phi(x)-\Phi(x^*)
    =
    \frac{1}{2}(x-x^*)^\top A(x-x^*).
    \label{eq:rpcd-objective-gap}
\end{equation}
Let $e_i$ denote the $i$th standard basis vector. Exact minimization of
$\Phi$ along coordinate $i$ gives the update
\begin{equation}
    C_i(x)
    :=
    x-
    \frac{e_i^\top(Ax-b)}{A_{ii}}e_i.
    \label{eq:rpcd-coordinate-update}
\end{equation}
Indeed, the right-hand side is the unique minimizer of
$\Phi(x+\tau e_i)$ over $\tau\in\mathbb{R}$.

Let $\mathfrak{S}_n$ be the set of permutations of $[n]$. For a permutation
$\pi=(\pi_1,\ldots,\pi_n)\in\mathfrak{S}_n$, define the corresponding
one-epoch map by
\begin{equation}
    G_\pi
    :=
    C_{\pi_n}\circ\cdots\circ\, C_{\pi_1}.
    \label{eq:rpcd-epoch-map}
\end{equation}
Thus, if $\pi_t$ is the permutation selected at epoch $t$, the usual
random-permutation coordinate-descent iteration can be written as
\begin{align}
    x^{t,0}
    &:=x^t,
    \nonumber\\
    x^{t,j}
    &:=
    C_{\pi_{t,j}}\bigl(x^{t,j-1}\bigr),
    \qquad (j\in[n])
    \nonumber\\
    x^{t+1}
    &:=
    x^{t,n}
    =
    G_{\pi_t}(x^t),
    \label{eq:rpcd-inner-iteration}
\end{align}
where $\pi_{t,j}$ denotes the $j$th coordinate of $\pi_t$. We associate with the algorithm the epoch correspondence
\begin{equation}
    \widetilde{F}(x)
    :=
    \left\{
        G_\pi(x):
        \pi\in\mathfrak{S}_n
    \right\}.
    \label{eq:rpcd-correspondence}
\end{equation}
The values of $\widetilde{F}$ are nonempty and finite, and hence compact.
An admissible trajectory
\begin{equation}
    x^{t+1}\in \widetilde{F}(x^t)
\end{equation}
is exactly an execution of coordinate descent in which an arbitrary
permutation may be selected at each epoch. This representation places the iteration within the framework of
algorithms generated by finite unions of single-valued maps
studied by Tam~\cite{Tam2018}. Here, every permutation is admissible
at each point. 

\paragraph{Why an epoch, rather than one coordinate update, is necessary.}

One might instead consider the one-coordinate correspondence
\begin{equation}
    \mathcal{R}(x)
    :=
    \left\{
        C_i(x):
        i\in[n]
    \right\}.
    \label{eq:rpcd-atomic-correspondence}
\end{equation}
This correspondence does not satisfy the all-branches convergence hypothesis
of the strong converse. For every coordinate $i$, every point satisfying
\begin{equation}
    e_i^\top(Ax-b)=0
\end{equation}
is fixed by the branch $C_i$. When $n\geq2$, this affine hyperplane contains
points other than $x^*$. By selecting the same coordinate repeatedly, one
therefore obtains an admissible trajectory that remains forever at a
nonoptimal point.

This obstruction is specific to the set-valued setting. Randomized coordinate
descent may converge probabilistically under weaker coordinate-sampling
conditions, but robust convergence of the full correspondence must hold for
every admissible sequence of choices. Grouping the updates into complete
epochs removes this obstruction because every coordinate is processed once
during every transition of $\widetilde{F}$.

\paragraph{Euclidean distance need not contract.}

Although every complete epoch converges toward $x^*$, its map need not be
contracting under every distance metric. Here we show that $\widetilde{F}$ is not contracting under Euclidean distance. Consider the two-dimensional instance
\begin{equation}
    A
    =
    \begin{pmatrix}
        1 & 2\\
        2 & 5
    \end{pmatrix},
    \qquad
    b=0.
    \label{eq:rpcd-euclidean-counterexample-matrix}
\end{equation}
The matrix is positive definite because its leading principal minors are
positive. The minimizer is $x^*=0$. For later use, define the linear error-update matrices
\begin{equation}
    Q_i
    :=
    I-\frac{1}{A_{ii}}e_i e_i^\top A.
    \label{eq:rpcd-coordinate-error-matrix}
\end{equation}
Since $Ax^*=b$, equation
\eqref{eq:rpcd-coordinate-update} gives $  C_i(x)-x^*=Q_i(x-x^*).$
For the matrix in \eqref{eq:rpcd-euclidean-counterexample-matrix},
\begin{equation}
    Q_1
    =
    \begin{pmatrix}
        0 & -2\\
        0 & 1
    \end{pmatrix},
    \qquad
    Q_2
    =
    \begin{pmatrix}
        1 & 0\\
        -\frac{2}{5} & 0
    \end{pmatrix}.
\end{equation}
For the ordering that updates coordinate one followed by coordinate two,
\begin{equation}
    Q_2Q_1
    =
    \begin{pmatrix}
        0 & -2\\
        0 & \frac{4}{5}
    \end{pmatrix}.
\end{equation}
Consequently,
\begin{equation}
    \left\|
        G_{(1,2)}(e_2)
    \right\|_2
    =
    \left\|
        Q_2Q_1e_2
    \right\|_2
    =
    \left\|
        \begin{pmatrix}
            -2\\
            \frac{4}{5}
        \end{pmatrix}
    \right\|_2
    =
    \frac{\sqrt{116}}{5}
    >
    1
    =
    \|e_2\|_2.
    \label{eq:rpcd-euclidean-expansion}
\end{equation}
Because $\widetilde{F}(0)=\{0\}$, it follows that
\begin{equation}
    H_{\|\cdot\|_2}
    \left(
        \widetilde{F}(e_2),
        \widetilde{F}(0)
    \right)
    >
    \|e_2-0\|_2.
    \label{eq:rpcd-not-euclidean-nonexpansive}
\end{equation}
Thus, the epoch correspondence can expand Euclidean distances even though
every admissible sequence of epochs converges to the unique minimizer.

\paragraph{The energy metric and convergence analysis.}

Now we present a metric where $\widetilde{F}$ is contracting. Define the energy
norm and its associated metric by
\begin{equation}
    \|z\|_A
    :=
    \sqrt{z^\top A z},
    \qquad
    d_A(x,y)
    :=
    \|x-y\|_A.
    \label{eq:rpcd-energy-metric}
\end{equation}
Since $A$ is positive definite,
\begin{equation}
    \sqrt{\lambda_{\min}(A)}\,\|x-y\|_2
    \leq
    d_A(x,y)
    \leq
    \sqrt{\lambda_{\max}(A)}\,\|x-y\|_2.
    \label{eq:rpcd-metric-equivalence}
\end{equation}
Hence $d_A$ is (metrically) equivalent to Euclidean distance, and
$(\mathbb{R}^n,d_A)$ is complete. For a linear map $M$, write
\begin{equation}
    \|M\|_A
    :=
    \sup_{z\neq0}
    \frac{\|Mz\|_A}{\|z\|_A}=\max_{\|z\|_A=1}\|Mz\|_A.
    \label{eq:rpcd-induced-energy-norm}
\end{equation}
For each permutation $\pi\in\mathfrak{S}_n$, define its epoch error matrix by $ M_\pi:= Q_{\pi_n}\cdots Q_{\pi_1},$
and define the worst-order epoch factor
\begin{equation}
    \widetilde{\rho}
    :=
    \max_{\pi\in\mathfrak{S}_n}
    \|M_\pi\|_A.
    \label{eq:rpcd-contraction-factor}
\end{equation}
Now we prove that $\widetilde{\rho}$ is a contracting factor namely $\widetilde{\rho} \in [0,1)$. Notice that exact coordinate minimization can be considered a special case of successive
subspace correction, with each correction performed on
$\operatorname{span}\{e_i\}$. Based on this viewpoint, the following energy-norm argument
is consistent with the classical analysis of such methods
\cite{XuZikatanov2002}. We include a direct self-contained proof for our setting.
\begin{lemma}[Uniform contraction of complete coordinate sweeps]
\label{lem:rpcd-complete-sweep-contraction}
The worst-order epoch factor satisfies
\begin{equation}
    0\leq \widetilde{\rho}<1.
    \label{eq:rpcd-rho-less-than-one}
\end{equation}
\end{lemma}

\begin{proof}
Let $\alpha=  ({e_i^\top A z})/{A_{ii}}$. For every $z\in\mathbb{R}^n$, direct expansion of
\eqref{eq:rpcd-coordinate-error-matrix} gives
\begin{equation}
 \label{eq:rpcd-pythagorean-identity}
\begin{aligned}
\|Q_i z\|_A^2
    &=
    (z-\alpha e_i)^\top A(z-\alpha e_i) \\
    &=
    z^\top A z
    -
    2\alpha e_i^\top A z
    +
    \alpha^2 e_i^\top A e_i\\ 
    &=
    \|z\|_A^2
    -
    \frac{(e_i^\top A z)^2}{A_{ii}}.    
\end{aligned}
\end{equation}
Thus, $Q_i$ is nonexpansive in the energy norm. In fact, it is the
$A$-orthogonal projection onto the subspace
\begin{equation}
    \mathcal{S}_i
    :=
    \left\{
        z\in\mathbb{R}^n:
        e_i^\top A z=0
    \right\}.
\end{equation}
Equality in \eqref{eq:rpcd-pythagorean-identity} holds precisely when
$z\in\mathcal{S}_i$, in which case $Q_i z=z$. Fix a permutation $\pi$ and a nonzero vector $z$. Set $z_0:=z$ and
$z_j:=Q_{\pi_j}z_{j-1}$ for $j\in[n]$. From \eqref{eq:rpcd-pythagorean-identity}, we know that $\|M_\pi z\|_A \le \|z\|_A.$ Thus, suppose, toward a contradiction, that
\begin{equation}
    \|M_\pi z\|_A
    =
    \|z\|_A.
\end{equation}
Since the sequence
\begin{equation}
    \|z_0\|_A,\,\|z_1\|_A,\,\ldots\,,\|z_n\|_A
\end{equation}
is nonincreasing, equality between its first and last terms implies equality
at every coordinate update. Equation
\eqref{eq:rpcd-pythagorean-identity} therefore gives
\begin{equation}
    e_{\pi_j}^\top A z_{j-1}=0
    \qquad
    \text{for every }j\in[n].
\end{equation}
It also follows that $z_j=z_{j-1}$ at every step. Hence $z_j=z$ for every
$j$, and
\begin{equation}
    e_{\pi_j}^\top A z=0
    \qquad
    \text{for every }j\in[n].
\end{equation}
Because $\pi$ contains every coordinate exactly once, this implies $Az=0$.
Positive definiteness of $A$ then gives $z=0$, a contradiction. We have shown that $\|M_\pi z\|_A<\|z\|_A$ for every nonzero $z$. The
$A$-unit sphere $S_A = \{z: \|z\|_A = 1\}$  is compact, so
$\|M_\pi\|_A<1$. Finally, $\mathfrak{S}_n$ is finite, and therefore the
maximum in \eqref{eq:rpcd-contraction-factor} is also strictly smaller than
one.
\end{proof}

\begin{proposition}[RPCD as a Nadler contraction]
\label{prop:rpcd-nadler-contraction}
The correspondence $\widetilde{F}$ is a Hausdorff
$\widetilde{\rho}$-contraction under the energy metric:
\begin{equation}
    H_{d_A}
    \left(
        \widetilde{F}(x),
        \widetilde{F}(y)
    \right)
    \leq
    \widetilde{\rho}d_A(x,y)
    \qquad
    \text{for all }x,y\in\mathbb{R}^n.
    \label{eq:rpcd-hausdorff-contraction}
\end{equation}
Moreover, $\widetilde{\rho}$ is the smallest global contraction factor
for $\widetilde{F}$ under $d_A$.
\end{proposition}

\begin{proof}
We start from expanding the numerator of coordinate map $C_i$:
\begin{align}
    C_i(x)
    =
    x
    -
    \frac{e_i^\top Ax}{A_{ii}}e_i
    +
    \frac{e_i^\top b}{A_{ii}}e_i 
    =
    Q_i x+c_i,
\end{align}
where $Q_i$ is defined in \eqref{eq:rpcd-coordinate-error-matrix} and $ c_i
=\frac{e_i^\top b}{A_{ii}}e_i.$ Consequently, for any two points $u,v\in\mathbb{R}^n$,
\[
\begin{aligned}
    C_i(u)-C_i(v)
    =
    (Q_i u+c_i)-(Q_i v+c_i) 
    =
    Q_i(u-v).
\end{aligned}
\]
The constant vector $c_i$ is the same in both updates because both applications use the same coordinate $i$, matrix $A$, and vector $b$; hence, it cancels upon subtraction. Now apply the same permutation $\pi=(\pi_1,\ldots,\pi_n)$ to both $x$ and $y$. Repeatedly applying the preceding identity yields
\begin{equation}
\label{eq:rpcd-epoch-difference}    
\begin{aligned}
    \nonumber
    G_\pi(x)-G_\pi(y)
    &=
    C_{\pi_n}\circ\cdots\circ\, C_{\pi_1}(x)
    -
    C_{\pi_n}\circ\cdots\circ\, C_{\pi_1}(y) \\
    \nonumber
    &=
    Q_{\pi_n}
    \left(
        C_{\pi_{n-1}}\circ\cdots\circ\, C_{\pi_1}(x)
        -
        C_{\pi_{n-1}}\circ\cdots\circ\, C_{\pi_1}(y)
    \right) \\
    \nonumber
    &=
    \cdots \\
    \nonumber
    &=
    Q_{\pi_n}\cdots Q_{\pi_1}(x-y) \\
    &=
    M_\pi(x-y).  
\end{aligned}
\end{equation}

Consequently, for every permutation $\pi$,
\begin{equation}
    d_A
    \left(
        G_\pi(x),
        G_\pi(y)
    \right)
    \leq
    \widetilde{\rho}d_A(x,y).
    \label{eq:rpcd-branch-contraction}
\end{equation}
Fix $u\in \widetilde{F}(x)$. There is a permutation $\pi$ such that
$u=G_\pi(x)$. Pairing $u$ with the point $G_\pi(y)\in
\widetilde{F}(y)$ gives
\begin{equation}
    \operatorname{dist}_{d_A}
    \left(
        u,\widetilde{F}(y)
    \right)
    \leq
    \widetilde{\rho}d_A(x,y).
\end{equation}
Taking the supremum over $u\in \widetilde{F}(x)$ gives one directed
Hausdorff bound. Interchanging $x$ and $y$ gives the other, proving
\eqref{eq:rpcd-hausdorff-contraction}. It remains to prove sharpness. Since the set of permutations is finite and
the $A$-unit sphere is compact, there are a permutation $\pi^*$ and a
vector $z^*$ satisfying
\begin{equation}
    \|z^*\|_A=1,
    \qquad
    \|M_{\pi^*}z^*\|_A
    =
    \widetilde{\rho}.
\end{equation}
Let $x=x^*+z^*$ and $y=x^*$. Every coordinate update fixes
$x^*$, and hence $ \widetilde{F}(x^*)=\{x^*\}.$
It follows that
\begin{align}
    H_{d_A}
    \left(
        \widetilde{F}(x),
        \widetilde{F}(x^*)
    \right)
    =
    \max_{\pi\in\mathfrak{S}_n}
    \|M_\pi z^*\|_A
    =
    \widetilde{\rho}
    =
    \widetilde{\rho}\,d_A(x,x^*).
\end{align}
Thus no smaller global contraction factor is possible under $d_A$.
\end{proof}
In fact, we can prove an explicit upper bound estimate for $\widetilde{\rho}$, using the framework in \cite[Appendix A.1]{SunYe2021} for cyclic coordinate descent.
\begin{proposition}[Explicit upper bound for $\widetilde\rho$]
\label{prop:rpcd-explicit-rate}
Let $A\succ0$ be the Hessian of the quadratic objective $\Phi$, and define
\begin{equation}
\label{eq:rpcd-parameters}
\begin{aligned}
D_A:=\operatorname{diag}(A_{11},\ldots,A_{nn}),\quad
 B:=D_A^{-1/2}AD_A^{-1/2},\quad
\mu:=\lambda_{\min}(B),\quad
 S:=\sum_{i<j}\frac{A_{ij}^2}{A_{ii}A_{jj}}.
\end{aligned}
\end{equation}
Then the factor $\widetilde\rho$ satisfies
\begin{equation}
\label{eq:rpcd-factor-bound}
\widetilde\rho
=\max_{\pi\in \mathfrak{S}_n}\|M_\pi\|_A
\le q:=\sqrt{\frac{S}{S+\mu}}<1.
\end{equation}
\end{proposition}

\begin{proof}
Let $x^\ast$ be the unique minimizer of $\Phi$.
We will bound the decrease of the error in energy norm (or "error energy" in short):
\begin{equation}
\label{eq:rpcd-objective-energy}
\|x-x^\ast\|_A^2
=2\bigl(\Phi(x)-\Phi(x^\ast)\bigr)
\end{equation}
after every coordinate has been minimized once. The proof uses two
identities: the total energy decrease equals the sum of the squared
normalized coordinate changes, and the final gradient is the strictly
upper triangular part of the normalized Hessian applied to those
changes. See \eqref{eq:rpcd-total-energy-decrease} and
\eqref{eq:rpcd-final-gradient}. Together they bound the remaining energy
by a fixed multiple of the energy lost.

Let $e_i$ denote the $i$th standard basis vector and let $I$ be
the $n\times n$ identity matrix.
Exact minimization along coordinate $i$ changes the current error
$v$ into $Q_i v$, where
\begin{equation}
\label{eq:rpcd-coordinate-matrix}
v:=x-x^\ast,
\qquad
Q_i:=I-\frac{1}{A_{ii}}e_i e_i^\top A.
\end{equation}
For a fixed ordering $\pi=(\pi_1,\ldots,\pi_n)$, recall that
\begin{equation}
\label{eq:rpcd-epoch-matrix}
M_\pi=Q_{\pi_n}\cdots Q_{\pi_1}.
\end{equation}
We analyze the updates in \eqref{eq:rpcd-epoch-matrix} after
rescaling and reordering the coordinates. Define
\begin{equation}
\label{eq:rpcd-permuted-hessian}
P_\pi:=(e_{\pi_1},\ldots,e_{\pi_n}),
\qquad
B_\pi:=P_\pi^\top BP_\pi.
\end{equation}
Multiplication by $P_\pi^\top$ places coordinate $\pi_j$ in
position $j$. The definitions in \eqref{eq:rpcd-parameters}
and \eqref{eq:rpcd-permuted-hessian} give
\begin{equation}
\label{eq:rpcd-unit-diagonal}
B_{ii}=1,
\qquad
(B_\pi)_{ii}=1
\qquad (i=1,\ldots,n).
\end{equation}
Consequently, $B_\pi$ is the Hessian in coordinates where the
updates occur in the order $1,\ldots,n$ and every coordinate
has unit curvature. For an arbitrary initial error $v\in\mathbb R^n$, let
\begin{equation}
\label{eq:rpcd-initial-error}
z^{(0)}:=P_\pi^\top D_A^{1/2}v,
\qquad
z^{(0)}_j=\sqrt{A_{\pi_j\pi_j}}\,v_{\pi_j}.
\end{equation}
Thus $z^{(0)}$ represents the same initial error in the rescaled
and reordered coordinates. The change of variables in
\eqref{eq:rpcd-initial-error} preserves the error energy:
\begin{equation}
\label{eq:rpcd-initial-energy}
\|z^{(0)}\|_{B_\pi}^2=\|v\|_A^2.
\end{equation}
Moreover, $P_\pi$ is orthogonal, so
\eqref{eq:rpcd-permuted-hessian} implies that $B_\pi$ and $B$
have the same eigenvalues. Hence
\begin{equation}
\label{eq:rpcd-spectral-lower-bound}
\lambda_{\min}(B_\pi)=\mu>0.
\end{equation}
Define the strictly lower triangular matrix $L_\pi$ by
\begin{equation}
\label{eq:rpcd-lower-triangle}
(L_\pi)_{ij}:=
\begin{cases}
(B_\pi)_{ij}, & i>j,\\
0, & i\le j.
\end{cases}
\end{equation}
By symmetry of $B_\pi$, together with
\eqref{eq:rpcd-unit-diagonal} and
\eqref{eq:rpcd-lower-triangle},
\begin{equation}
\label{eq:rpcd-triangular-decomposition}
B_\pi=I+L_\pi+L_\pi^\top.
\end{equation}
The entries in row $j$ of $L_\pi$ describe how changes to the
previously updated coordinates $1,\ldots,j-1$ affect the
derivative in coordinate $j$.

We now track each update explicitly. For $j=1,\ldots,n$, let
$z^{(j)}$ be the transformed error after coordinates
$1,\ldots,j$ have been updated, and let $s_j$ be the change made
to coordinate $j$.
At update $j$, we minimize the quadratic objective by changing
only coordinate $j$. For an increment $\alpha\in\mathbb R$,
\eqref{eq:rpcd-unit-diagonal} gives
\begin{equation}
\label{eq:rpcd-coordinate-objective}
\begin{aligned}
\frac12
(z^{(j-1)}+\alpha e_j)^\top B_\pi
(z^{(j-1)}+\alpha e_j)
=\frac12\|z^{(j-1)}\|_{B_\pi}^2+\alpha e_j^\top B_\pi z^{(j-1)}
+\frac{\alpha^2}{2}.
\end{aligned}
\end{equation}
Setting the derivative of \eqref{eq:rpcd-coordinate-objective}
with respect to $\alpha$ equal to zero gives the unique minimizing
increment. Thus,
\begin{equation}
\label{eq:rpcd-normalized-update}
s_j:=-e_j^\top B_\pi z^{(j-1)},
\qquad
z^{(j)}:=z^{(j-1)}+s_j e_j.
\end{equation}
The updates in \eqref{eq:rpcd-normalized-update} are precisely
the original coordinate updates expressed in the new coordinates.
Indeed, \eqref{eq:rpcd-coordinate-matrix} and
\eqref{eq:rpcd-permuted-hessian} imply
\begin{equation}
\label{eq:rpcd-conjugated-update}
P_\pi^\top (D_A^{1/2}Q_{\pi_j}D_A^{-1/2}) P_\pi
=I-e_j e_j^\top B_\pi.
\end{equation}
Combining \eqref{eq:rpcd-epoch-matrix},
\eqref{eq:rpcd-initial-error}, and
\eqref{eq:rpcd-conjugated-update} therefore gives
\begin{equation}
\label{eq:rpcd-final-error}
\begin{aligned}
z^{(n)}=P_\pi^\top D_A^{1/2}M_\pi v,\qquad \|z^{(n)}\|_{B_\pi}^2=\|M_\pi v\|_A^2.
\end{aligned}
\end{equation}
Now collect the coordinate changes into
\begin{equation}
\label{eq:rpcd-increments}
s:=(s_1,\ldots,s_n)^\top.
\end{equation}
Since every coordinate is updated exactly once,
\eqref{eq:rpcd-normalized-update} and
\eqref{eq:rpcd-increments} yield
\begin{equation}
\label{eq:rpcd-accumulated-updates}
\begin{aligned}
z^{(j-1)}=z^{(0)}+\sum_{i<j}s_i e_i,\qquad z^{(n)}=z^{(0)}+s.
\end{aligned}
\end{equation}
Substituting \eqref{eq:rpcd-accumulated-updates} into
\eqref{eq:rpcd-normalized-update} gives
\begin{equation}
\label{eq:rpcd-coordinate-optimality}
s_j+\sum_{i<j}(B_\pi)_{ji}s_i
=-(B_\pi z^{(0)})_j.
\end{equation}
By \eqref{eq:rpcd-lower-triangle}, the equations in
\eqref{eq:rpcd-coordinate-optimality} can be written as
\begin{equation}
\label{eq:rpcd-triangular-system}
(I+L_\pi)s=-B_\pi z^{(0)}.
\end{equation}
Using \eqref{eq:rpcd-accumulated-updates},
\eqref{eq:rpcd-triangular-system}, and
\eqref{eq:rpcd-triangular-decomposition}, we obtain
\begin{equation}
\label{eq:rpcd-final-gradient}
\begin{aligned}
B_\pi z^{(n)}
=B_\pi z^{(0)}+B_\pi s
=-(I+L_\pi)s+B_\pi s
=L_\pi^\top s.
\end{aligned}
\end{equation}
Equation \eqref{eq:rpcd-final-gradient} describes the derivative
remaining after all updates. The derivative in coordinate $j$
becomes zero when that coordinate is minimized, but subsequent
changes to coordinates $j+1,\ldots,n$ can make it nonzero again.
Those subsequent contributions form the $j$th component of
$L_\pi^\top s$.

Next, we relate the coordinate changes to the energy decrease.
Using \eqref{eq:rpcd-normalized-update} and
\eqref{eq:rpcd-unit-diagonal}, we obtain
\begin{equation}
\label{eq:rpcd-one-update-energy}
\begin{aligned}
\|z^{(j)}\|_{B_\pi}^2
&=\|z^{(j-1)}\|_{B_\pi}^2
  +2s_j e_j^\top B_\pi z^{(j-1)}+s_j^2\\
&=\|z^{(j-1)}\|_{B_\pi}^2-s_j^2.
\end{aligned}
\end{equation}
Thus \eqref{eq:rpcd-one-update-energy} shows that $s_j^2$ is
exactly the energy lost at update $j$.
Summing \eqref{eq:rpcd-one-update-energy} over $j=1,\ldots,n$
and using \eqref{eq:rpcd-increments} yields
\begin{equation}
\label{eq:rpcd-total-energy-decrease}
\|z^{(0)}\|_{B_\pi}^2-\|z^{(n)}\|_{B_\pi}^2
=\sum_{j=1}^n s_j^2
=\|s\|_2^2,
\end{equation}
where $\|\cdot\|_2$ denotes the Euclidean norm.

To bound the remaining energy, first observe that
\eqref{eq:rpcd-lower-triangle},
\eqref{eq:rpcd-permuted-hessian}, and
\eqref{eq:rpcd-parameters} give
\begin{equation}
\label{eq:rpcd-lower-triangle-size}
\begin{aligned}
\|L_\pi\|_F^2:=\sum_{i,j}(L_\pi)_{ij}^2
=\sum_{i>j}(B_\pi)_{ij}^2
=\sum_{i<j}B_{ij}^2
=S.
\end{aligned}
\end{equation}
Here $\|\cdot\|_F$ denotes the Frobenius norm.
The penultimate equality in \eqref{eq:rpcd-lower-triangle-size}
holds because every symmetric off-diagonal pair contributes
exactly one entry below the diagonal, regardless of the permutation.
Thus \eqref{eq:rpcd-lower-triangle-size} identifies $S$ as the
total squared size of these interactions.

Using \eqref{eq:rpcd-final-gradient}, we have
\begin{equation}
\label{eq:rpcd-final-energy-bound}
\begin{aligned}
\|z^{(n)}\|_{B_\pi}^2
=(L_\pi^\top s)^\top B_\pi^{-1}(L_\pi^\top s)
\le \frac{1}{\mu}\|L_\pi^\top s\|_2^2
\le \frac{\|L_\pi\|_F^2}{\mu}\|s\|_2^2
=\frac{S}{\mu}\|s\|_2^2.
\end{aligned}
\end{equation}
The first inequality in \eqref{eq:rpcd-final-energy-bound}
follows from \eqref{eq:rpcd-spectral-lower-bound},
the second follows from Cauchy--Schwarz, and the final equality
uses \eqref{eq:rpcd-lower-triangle-size}.
Substituting \eqref{eq:rpcd-total-energy-decrease} into
\eqref{eq:rpcd-final-energy-bound} gives
\begin{equation}
\label{eq:rpcd-energy-comparison}
\|z^{(n)}\|_{B_\pi}^2
\le
\frac{S}{\mu}
\left(
\|z^{(0)}\|_{B_\pi}^2-\|z^{(n)}\|_{B_\pi}^2
\right).
\end{equation}
Rearranging \eqref{eq:rpcd-energy-comparison} and applying
\eqref{eq:rpcd-initial-energy} and \eqref{eq:rpcd-final-error}
yields
\begin{equation}
\label{eq:rpcd-original-energy-bound}
\|M_\pi v\|_A^2
\le \frac{S}{S+\mu}\|v\|_A^2.
\end{equation}
Dividing \eqref{eq:rpcd-original-energy-bound} by $\|v\|_A^2$
and taking the supremum over nonzero initial errors gives
\begin{equation}
\label{eq:rpcd-operator-bound}
\|M_\pi\|_A^2
=\sup_{v\ne0}\frac{\|M_\pi v\|_A^2}{\|v\|_A^2}
\le \frac{S}{S+\mu}.
\end{equation}
Since \eqref{eq:rpcd-operator-bound} holds for every permutation,
taking square roots and maximizing over $\pi$ proves
\begin{equation}
\label{eq:rpcd-conclusion}
\widetilde\rho
=\max_{\pi\in \mathfrak{S}_n}\|M_\pi\|_A
\le \sqrt{\frac{S}{S+\mu}}
=q<1.
\end{equation}

The strict inequality in \eqref{eq:rpcd-conclusion} follows from
\eqref{eq:rpcd-spectral-lower-bound} and the nonnegativity of
the summands defining $S$ in \eqref{eq:rpcd-parameters}.
In particular, if $S=0$, then $B=I$ and one complete sweep reaches
the minimizer exactly.
\end{proof}

\begin{corollary}[Uniform convergence over coordinate orderings]
\label{cor:rpcd-all-order-convergence}
For every $x^0\in\mathbb{R}^n$ and every $t\geq 0$,
\begin{equation}
    H_{d_A}
    \left(
        \widetilde{F}^{[t]}(x^0),
        \{x^*\}
    \right)
    \leq
    \widetilde{\rho}^{\,t}d_A(x^0,x^*).
    \label{eq:rpcd-reachable-set-convergence}
\end{equation}
Consequently, every trajectory
$x^{t+1}=G_{\pi_t}(x^t)$, under an arbitrary sequence of coordinate
orderings $(\pi_t)_{t\geq 0}$, satisfies
\begin{align}
    \Phi(x^t)-\Phi(x^*)
    &\leq
    \widetilde{\rho}^{\,2t}
    \bigl(\Phi(x^0)-\Phi(x^*)\bigr),
    \label{eq:rpcd-objective-convergence}
\end{align}
and
\begin{align}
    \|x^t-x^*\|_2
    &\leq
    \sqrt{\frac{\lambda_{\max}(A)}
        {\lambda_{\min}(A)}}\,
    \widetilde{\rho}^{\,t}
    \|x^0-x^*\|_2.
    \label{eq:rpcd-euclidean-convergence}
\end{align}
\end{corollary}

\begin{proof}
Since every coordinate update fixes the minimizer $x^*$, we have
\begin{equation}
    G_\pi(x^*)=x^*
    \qquad
    \text{for every }\pi\in\mathfrak{S}_n,
\end{equation}
and therefore $\widetilde{F}(x^*)=\{x^*\}.$ Let $x^{t+1}=G_{\pi_t}(x^t)$ for an arbitrary sequence of permutations.
Because $x^{t+1}\in\widetilde{F}(x^t)$, Proposition
\ref{prop:rpcd-nadler-contraction} gives
\begin{align}
    d_A(x^{t+1},x^*)
    \leq
    H_{d_A}
    \left(
        \widetilde{F}(x^t),
        \widetilde{F}(x^*)
    \right)
    \leq
    \widetilde{\rho}\,d_A(x^t,x^*).
    \label{eq:rpcd-one-step-convergence}
\end{align}
Iterating \eqref{eq:rpcd-one-step-convergence} yields
\begin{equation}
    d_A(x^t,x^*)
    \leq
    \widetilde{\rho}^{\,t}d_A(x^0,x^*).
    \label{eq:rpcd-pathwise-energy-convergence}
\end{equation}

Every point $y\in\widetilde{F}^{[t]}(x^0)$ is the endpoint of some
admissible sequence of $t$ coordinate orderings. Hence
\eqref{eq:rpcd-pathwise-energy-convergence} holds uniformly over all such
$y$. Since the Hausdorff distance from a nonempty set to a singleton is
the largest distance to that singleton, we obtain
\begin{align}
    H_{d_A}
    \left(
        \widetilde{F}^{[t]}(x^0),
        \{x^*\}
    \right)
    =
    \sup_{y\in\widetilde{F}^{[t]}(x^0)}
    d_A(y,x^*)
    \leq
    \widetilde{\rho}^{\,t}d_A(x^0,x^*).
\end{align}
This proves \eqref{eq:rpcd-reachable-set-convergence}. The quadratic objective satisfies
\begin{equation}
    \Phi(x)-\Phi(x^*)
    =
    \frac{1}{2}d_A(x,x^*)^2.
\end{equation}
Therefore, by \eqref{eq:rpcd-pathwise-energy-convergence}, we have
\begin{align}
    \Phi(x^t)-\Phi(x^*)
    =
    \frac{1}{2}d_A(x^t,x^*)^2
    \leq
    \frac{1}{2}
    \widetilde{\rho}^{\,2t}
    d_A(x^0,x^*)^2
    =
    \widetilde{\rho}^{\,2t}
    \bigl(\Phi(x^0)-\Phi(x^*)\bigr),
\end{align}
which proves \eqref{eq:rpcd-objective-convergence}. Finally, positive definiteness of $A$ gives
\begin{equation}
    \sqrt{\lambda_{\min}(A)}\,\|z\|_2
    \leq
    \|z\|_A
    \leq
    \sqrt{\lambda_{\max}(A)}\,\|z\|_2.
\end{equation}
Combining these inequalities with
\eqref{eq:rpcd-pathwise-energy-convergence}, we obtain
\begin{align}
    \|x^t-x^*\|_2
    \leq
    \frac{1}{\sqrt{\lambda_{\min}(A)}}
    d_A(x^t,x^*)
    \leq
    \frac{
        \widetilde{\rho}^{\,t}
    }{
        \sqrt{\lambda_{\min}(A)}
    }
    d_A(x^0,x^*)
    \leq
    \sqrt{
        \frac{\lambda_{\max}(A)}
        {\lambda_{\min}(A)}
    }\,
    \widetilde{\rho}^{\,t}
    \|x^0-x^*\|_2.
\label{eq:rpcd-eigenvalue-bound-of-distance}
\end{align}
This is precisely
\eqref{eq:rpcd-euclidean-convergence}.
\end{proof}

Qualitative convergence under arbitrary sequences of complete
sweeps also follows from the common-fixed-point theorem for
paracontractions~\cite[Theorem~1]{ElsnerKoltrachtNeumann1992}.
Indeed, the coordinate maps are orthogonal projections in the
energy inner product, their common fixed point is $x^*$,
and every coordinate is updated infinitely often.
The estimate above additionally provides a geometric rate
that is uniform over all sequences of complete sweeps.

\begin{corollary}[Residual error bound]
\label{cor:rpcd-residual-bound}
Define the epoch residual by
\begin{equation}
    r_A(x)
    :=
    \operatorname{dist}_{d_A}
    \left(
        x,\widetilde{F}(x)
    \right)
    =
    \min_{\pi\in\mathfrak{S}_n}
    d_A\left(x,G_\pi(x)\right).
    \label{eq:rpcd-residual}
\end{equation}
Then, for every $x\in\mathbb{R}^n$,
\begin{equation}
    (1-\widetilde{\rho})\,d_A(x,x^*)
    \leq
    r_A(x)
    \leq
    (1+\widetilde{\rho})\,d_A(x,x^*).
    \label{eq:rpcd-residual-error-bound}
\end{equation}
In particular, $r_A(x)=0$ if and only if $x=x^*$, so $x^*$ is the
unique fixed point of $\widetilde{F}$. Moreover, every admissible trajectory
satisfies
\begin{equation}
    \|x^t-x^*\|_2
    \leq
    \frac{
        \widetilde{\rho}^{\,t}
    }{
        (1-\widetilde{\rho})
        \sqrt{\lambda_{\min}(A)}
    }
    r_A(x^0).
    \label{eq:rpcd-residual-convergence}
\end{equation}
\end{corollary}

\begin{proof}
Fix $x\in\mathbb{R}^n$ and $u\in\widetilde{F}(x)$. Since
$\widetilde{F}(x^*)=\{x^*\}$, Proposition
\ref{prop:rpcd-nadler-contraction} implies
\begin{align}
    d_A(u,x^*)
    \leq
    H_{d_A}
    \left(
        \widetilde{F}(x),
        \widetilde{F}(x^*)
    \right)
    \leq
    \widetilde{\rho}\,d_A(x,x^*).
    \label{eq:rpcd-successor-endpoint-bound}
\end{align}
The reverse triangle inequality and
\eqref{eq:rpcd-successor-endpoint-bound} give
\begin{align}
    d_A(x,u)
    \geq
    d_A(x,x^*)-d_A(u,x^*)
    \geq
    (1-\widetilde{\rho})\,d_A(x,x^*).
    \label{eq:rpcd-residual-lower-pointwise}
\end{align}
Similarly, the triangle inequality gives
\begin{align}
    d_A(x,u)
    \leq
    d_A(x,x^*)+d_A(u,x^*)
    \leq
    (1+\widetilde{\rho})\,d_A(x,x^*).
    \label{eq:rpcd-residual-upper-pointwise}
\end{align}
Taking the infimum over
$u\in\widetilde{F}(x)$ in
\eqref{eq:rpcd-residual-lower-pointwise} and
\eqref{eq:rpcd-residual-upper-pointwise} yields
\begin{equation}
    (1-\widetilde{\rho})\,d_A(x,x^*)
    \leq
    r_A(x)
    \leq
    (1+\widetilde{\rho})\,d_A(x,x^*).
\end{equation}

Because $\widetilde{\rho}<1$, the lower bound shows that
$r_A(x)=0$ only if $x=x^*$. Conversely,
$\widetilde{F}(x^*)=\{x^*\}$ implies $r_A(x^*)=0$. Thus,
\begin{equation}
    r_A(x)=0
    \quad\Longleftrightarrow\quad
    x=x^*.
\end{equation}
In particular, if $x$ is a fixed point of $\widetilde{F}$, then
$r_A(x)=0$, and hence $x=x^*$. Therefore $x^*$ is the unique fixed
point. Now let $(x^t)_{t\geq0}$ be any admissible trajectory. By
Corollary~\ref{cor:rpcd-all-order-convergence},
\begin{equation}
    d_A(x^t,x^*)
    \leq
    \widetilde{\rho}^{\,t}d_A(x^0,x^*).
\end{equation}
The lower residual bound at $x^0$ gives
\begin{equation}
    d_A(x^0,x^*)
    \leq
    \frac{r_A(x^0)}{1-\widetilde{\rho}}.
\end{equation}
Consequently, we have
\begin{equation}
    d_A(x^t,x^*)
    \leq
    \frac{
        \widetilde{\rho}^{\,t}
    }{
        1-\widetilde{\rho}
    }
    r_A(x^0).
\end{equation}
Finally, using \eqref{eq:rpcd-eigenvalue-bound-of-distance} from Corollary \ref{cor:rpcd-all-order-convergence}, we get
\begin{equation}
    \|x^t-x^*\|_2
    \leq
    \frac{d_A(x^t,x^*)}
    {\sqrt{\lambda_{\min}(A)}}.
\end{equation}
And we conclude that
\begin{equation}
    \|x^t-x^*\|_2
    \leq
    \frac{
        \widetilde{\rho}^{\,t}
    }{
        (1-\widetilde{\rho})
        \sqrt{\lambda_{\min}(A)}
    }
    r_A(x^0),
\end{equation}
as claimed.
\end{proof}

\paragraph{Relation to the strong converse (Theorem \ref{thm:daskalakis-nadler-converse}).}

For any $R>0$, consider the ball under $d_A$:
\begin{equation}
    X_R
    :=
    \left\{
        x\in\mathbb{R}^n:
        d_A(x,x^*)\leq R
    \right\}.
\end{equation}
Every coordinate update never increases $d_A$-distance relative to $x^*$, so for $x \in X_R$ and arbitrary permutation $\pi$, $d_A(G_\pi(x),x^*) \le d_A(x,x^*) \le R$ and $X_R$ is forward invariant under every branch of
$\widetilde{F}$. The set $X_R$ is compact in the Euclidean topology,
and $\widetilde{F}$ is Hausdorff-continuous by Proposition \ref{prop:rpcd-nadler-contraction}. Equation
\eqref{eq:rpcd-reachable-set-convergence} gives uniform convergence of all
reachable sets to the endpoint $x^*$, since as $t \rightarrow \infty$,
\[
\sup_{x\in X_R}
H_{d_A}\bigl(\widetilde F^{[t]}(x),\{x^*\}\bigr)
\leq
\sup_{x\in X_R}\widetilde{\rho}^{\,t}d_A(x,x^*)
\leq
\widetilde{\rho}^{\,t}R
\longrightarrow 0.
\]
Thus, the correspondence
restricted to $X_R$ satisfies the hypotheses of the strong converse.


\begin{remark}
Although the epoch correspondence is motivated by random-permutation cyclic
coordinate descent, the convergence result above is not probabilistic.  The
distribution over permutations is discarded, and every permutation is instead
treated as an admissible branch of the correspondence $\widetilde F$.
Consequently, the estimate
\[
    d_A(x^t,x^*)
    \leq
    \widetilde{\rho}^{\,t}d_A(x^0,x^*)
\]
holds pathwise for every deterministic, adaptive, or even adversarial sequence of
coordinate orderings.  The quantity
\[
    \widetilde{\rho}
    =
    \max_{\pi\in\mathfrak S_n}\|M_\pi\|_A
\]
is a sharp one-epoch Hausdorff contraction factor of the full correspondence
under the energy metric, and it yields the convergence estimate with factor
$\widetilde{\rho}^{\,2t}$.  Thus, the result should be interpreted as a certificate for robust convergence (or as a topological perspective for the convergence), rather than as an improvement of the convergence rate for RPCD. Since it maximizes over
all permutations and repeatedly applies the worst one-step factor, the bound
may be more conservative than previous asymptotic analysis.
\end{remark}

\section{Omitted Proofs in Section \ref{sec:nadler-continuous-local-opt}}
\label{sec:appendix-nadler-complexity}
This section proves the reductions among
\textsc{Continuous-LocalOpt}, $\beta$-\textsc{H\"older-LocalOpt}, and $\beta$-\textsc{Projected-Nadler}.  Figure~\ref{fig:problem-reductions}
shows the organization of the constructions.

\begin{figure}[t]
\centering
\begin{tikzpicture}[
    problem/.style={
        draw=blue!45!black,
        fill=blue!4,
        rounded corners=3pt,
        line width=0.75pt,
        text width=4.5cm,
        minimum height=1.05cm,
        inner xsep=5pt,
        inner ysep=7pt,
        align=center,
        font=\small
    },
    reduction/.style={
        -{Latex[length=2.5mm,width=1.8mm]},
        line width=0.85pt,
        draw=black!80
    },
    edge label/.style={
        font=\small,
        fill=white,
        inner sep=3pt,
        align=center
    }
]

\node[problem] (co) at (-4,3)
    {\textsc{Continuous-LocalOpt}};

\node[problem] (ho) at (4,3)
    {$\beta$-\textsc{H\"older-LocalOpt}};

\node[problem, fill=blue!8] (pn) at (0,0)
    {$\beta$-\textsc{Projected-Nadler}};

\draw[reduction]
    ([yshift=2.2mm]co.east) --
    node[edge label, above=1pt]
        {Theorem~\ref{thm:holder-localopt-cls-hard}}
    ([yshift=2.2mm]ho.west);

\draw[reduction]
    ([yshift=-2.2mm]ho.west) --
    node[edge label, below=1pt]
        {Theorem~\ref{thm:holder-localopt-cls}}
    ([yshift=-2.2mm]co.east);

\draw[reduction]
    (co.south) --
    node[edge label, left=9pt]
        {Lemma~\ref{lem:clo-to-nadler}}
    (pn.north west);

\draw[reduction]
    (pn.north east) --
    node[edge label, right=9pt]
        {Lemma~\ref{lem:nadler-to-holder-localopt}}
    (ho.south);

\end{tikzpicture}
\caption{{Polynomial-time many-one reductions among
{\textsc{Continuous-}}{\textsc{LocalOpt}},
{$\beta$}-{\textsc{Projected-}}{\textsc{Nadler}}, and
{$\beta$}-{\textsc{H\"older-}}{\textsc{LocalOpt}}, for every fixed rational
{$\beta\in(0,1]$}. Each arrow points from the source problem
to the target problem and is labeled by the theorem or
lemma establishing the reduction.}}
\label{fig:problem-reductions}
\end{figure}

\subsection{Hölder-continuous local search}
\label{subsec:appendix-holder-localopt}

We introduce a Hölder-continuous analogue of
\textsc{Continuous-LocalOpt}. Throughout this subsection, fix a rational
constant $\beta\in(0,1]$. The exponent $\beta$ is not part of the input.
We show the three-dimensional formulation first and explicitly use their
$n$-dimensional variants in the constructions. The norm is $\ell_1$ throughout.

\begin{definition}[$\beta$-\textsc{H\"older-LocalOpt}]
\label{def:holder-localopt}
An instance consists of well-behaved arithmetic circuits
$f\colon[0,1]^3\to[0,1]^3$ and
$p\colon[0,1]^3\to[0,1]$, together with positive rational parameters
$\delta$ and $\lambda$. A valid output is one of the following.
\begin{outputlist}{(HO3)}
\item[\textnormal{(HO1)}] A point $x\in[0,1]^3$ satisfying
$p(f(x))\geq p(x)-\delta$.

\item[\textnormal{(HO2)}] Points $x,y\in[0,1]^3$ satisfying
$\|f(x)-f(y)\|_1>\lambda\|x-y\|_1^\beta$.

\item[\textnormal{(HO3)}] Points $x,y\in[0,1]^3$ satisfying
$|p(x)-p(y)|>\lambda\|x-y\|_1^\beta$.    
\end{outputlist}
\end{definition}

We claim that the problem is total. If either asserted Hölder bound fails, then an output of type (HO2) or (HO3) exists. Otherwise, suppose for contradiction that
no output of type (HO1) exists. Starting from an arbitrary $x_0\in[0,1]^3$,
define $x_{t+1}:=f(x_t)$. Then
$p(x_{t+1})<p(x_t)-\delta$ for every $t$, and hence
$p(x_T)<p(x_0)-T\delta$. Choosing an integer $T>1/\delta$ gives
$p(x_T)<0$, contradicting the fact that $p$ takes values in $[0,1]$.

We next give the interpolation lemma used in the membership reduction. Here, we use the simplex-based interpolation method from \cite[Section 1.2]{Davies1996}. However, we record the following implementation and quantitative properties because, in addition
to the usual interpolation facts, our reduction requires an explicit
polynomial-size arithmetic circuit and a Lipschitz bound with respect to
the $\ell_1$ norm. Figure~\ref{fig:simplex-interpolation-2d} gives a
two-dimensional example of the construction.

\begin{lemma}[Efficient simplex-based interpolation]
\label{lem:pn-kuhn}
Let $g:[0,1]^n\to[0,1]^r$ be represented by a well-behaved circuit,
and let $h=2^{-m}$ for an integer $m\geq1$. One can construct a
well-behaved circuit $I_hg:[0,1]^n\to[0,1]^r$, in time polynomial
in $n,r,m$ and the description length of $g$, with the following
properties.
\begin{enumerate}
\item[(i)] At each input $x$, $I_hg(x)$ is a convex combination of
the values of $g$ at $n+1$ grid vertices, each at $\ell_1$-distance
at most $nh$ from $x$. These vertices and weights are computable
within the same bound.

\item[(ii)] For all $x,y$, one has
$\|I_hg(x)-I_hg(y)\|_1\leq(r/h)\|x-y\|_1$.

\item[(iii)] If every vertex $v$ used at $x$ satisfies
$\|g(x)-g(v)\|_1\leq\lambda\|x-v\|_1^\beta$, then
$\|g(x)-I_hg(x)\|_1\leq\lambda(nh)^\beta$.
\end{enumerate}
\end{lemma}

\begin{proof}
We start by introducing the basic construction for the interpolation.
\paragraph{Construction and part (i).}
Partition $[0,1]^n$ into grid cubes of side length $h$. For an input
$x$, choose a cube containing $x$ as follows: in each coordinate, a
positive grid boundary belongs to the interval immediately to its left,
and zero belongs to the first interval. Write $c$ for the lower corner
of this cube and set $t_i=(x_i-c_i)/h$. Thus $t_i\in[0,1]$ records the
position of $x_i$ within its chosen interval.

Choose a permutation $\pi$ satisfying
$t_{\pi(1)}\geq\cdots\geq t_{\pi(n)}$, breaking ties by increasing index.
Starting from $c$, increase one coordinate at a time by $h$, in this
order, to obtain the vertices
\begin{equation}
v_0=c,\qquad
v_j=c+h\sum_{\ell=1}^j e_{\pi(\ell)}
\quad(1\leq j\leq n).
\label{eq:pn-kuhn-vertices}
\end{equation}
The $e_j$ denotes the $j$-th standard basis vector. Their convex hull is the selected simplex. Define the weights
\begin{equation}
w_0=1-t_{\pi(1)},\qquad
w_j=t_{\pi(j)}-t_{\pi(j+1)}\ (1\leq j<n),\qquad
w_n=t_{\pi(n)}.
\label{eq:pn-kuhn-weights}
\end{equation}
The ordering makes these weights nonnegative, and their sum telescopes
to one. Moreover, coordinate $\pi(\ell)$ has been increased precisely
at the vertices $v_\ell,\ldots,v_n$, whose total weight is
$\sum_{j=\ell}^n w_j=t_{\pi(\ell)}$. Consequently,
$x=\sum_{j=0}^n w_jv_j$.

Set $I_hg(x)=\sum_{j=0}^n w_jg(v_j)$. This is a convex combination of
$n+1$ values in $[0,1]^r$, so it also belongs to $[0,1]^r$.
Every $v_j$ lies in the same cube as $x$, and hence
$\|x-v_j\|_1\leq nh$. This proves the geometric claims in part (i);
the computational claim is proved below.

\paragraph{Continuity of $I_hg$.}
Within a fixed simplex, the vertices are fixed and the weights are
affine functions of $x$. Thus the interpolation formula is affine
there. It remains to check that the formulas agree when $x$ lies on a
boundary between simplices or between grid cubes.

First consider different simplices in the same cube. Their formulas
can meet only where some normalized coordinates are equal. Suppose
$t_{\pi(j)}=t_{\pi(j+1)}$. Exchanging these two coordinates in the order
changes only $v_j$: all earlier vertices are unchanged, and after both
coordinates have been increased, the later vertices are unchanged as
well. The weight of the changed vertex is
$w_j=t_{\pi(j)}-t_{\pi(j+1)}=0$. Therefore the interpolated value is
unchanged. Any two orderings of tied coordinates are related by such
adjacent exchanges, so all formulas within the cube agree on their
common boundaries.

Next consider two cubes sharing a face perpendicular to coordinate
$i$. On this face, the normalized coordinate is $t_i=1$ in the cube
on the left and $t_i=0$ in the cube on the right; the other normalized
coordinates agree. By the preceding tie argument, we may place
coordinate $i$ first in the left cube's ordering and last in the right
cube's ordering. In the left cube, only $v_0$ lies outside the shared
face, and its weight is $w_0=1-t_i=0$. In the right cube, only $v_n$
lies outside that face, and its weight is $w_n=t_i=0$.
After discarding these zero-weight vertices, both formulas start at
the lower corner of the shared face and increase the remaining
coordinates in exactly the same order. Their weights are also the
same, since they are formed from the same ordered coordinates.
Thus the two formulas agree on the shared face.

Repeating this argument handles intersections of several grid faces.
The finitely many affine formulas therefore agree on all overlaps of
the closed simplices, and together define a continuous function on
$[0,1]^n$. Notice that this argument does not require $g$ itself to be
continuous.

\paragraph{Part (ii): the Lipschitz bound.}
We first prove the bound when $x$ and $y$ belong to the same closed
simplex. By the boundary agreement just proved, both values can be
computed using this simplex's vertices $v_0,\ldots,v_n$ and coordinate
order $\pi$, even if the boundary convention selects another simplex.
Collecting the coefficients of the vertex values in the weighted sum gives
\begin{equation}
I_hg(x)=g(v_0)+\sum_{j=1}^n
\frac{x_{\pi(j)}-c_{\pi(j)}}{h}
\bigl(g(v_j)-g(v_{j-1})\bigr).
\label{eq:pn-kuhn-affine}
\end{equation}
Each coordinate of $g(v_j)-g(v_{j-1})$ lies in $[-1,1]$, so
$\|g(v_j)-g(v_{j-1})\|_1\leq r$. Subtracting the two affine formulas
and applying the triangle inequality yields
\begin{equation*}
\begin{aligned}
\|I_hg(x)-I_hg(y)\|_1
&\leq \frac1h\sum_{j=1}^n
|x_{\pi(j)}-y_{\pi(j)}|
\,\|g(v_j)-g(v_{j-1})\|_1\\
&\leq \frac rh\sum_{j=1}^n|x_{\pi(j)}-y_{\pi(j)}|
=\frac rh\|x-y\|_1.
\end{aligned}
\end{equation*}

For arbitrary $x,y$, subdivide the straight segment from $x$ to $y$
into finitely many subsegments, each contained in a closed simplex.
Such a subdivision exists because there are finitely many simplices
and the intersection of a segment with a convex simplex is an interval,
a point, or the empty set. Write the subdivision points as
$z_k=x+\theta_k(y-x)$, where
$0=\theta_0<\cdots<\theta_M=1$. Continuity ensures that the formula from
either neighboring simplex gives the same value at a subdivision point.
Applying the preceding bound to each subsegment gives
\begin{equation*}
\begin{aligned}
\|I_hg(x)-I_hg(y)\|_1
&\leq\sum_{k=1}^M\|I_hg(z_k)-I_hg(z_{k-1})\|_1\\
&\leq\frac rh\sum_{k=1}^M\|z_k-z_{k-1}\|_1
=\frac rh\|x-y\|_1.
\end{aligned}
\end{equation*}
The final equality follows from
$z_k-z_{k-1}=(\theta_k-\theta_{k-1})(y-x)$ and
$\sum_{k=1}^M(\theta_k-\theta_{k-1})=1$.
This proves part (ii), including at simplex boundaries.

\paragraph{Part (iii): the approximation bound.}
Fix $x$ and suppose that the stated comparisons hold at its selected
vertices. Since the weights are nonnegative and sum to one,
\begin{equation*}
\begin{aligned}
\|g(x)-I_hg(x)\|_1
\leq\sum_{j=0}^n w_j\|g(x)-g(v_j)\|_1
\leq\lambda\sum_{j=0}^n w_j\|x-v_j\|_1^\beta
\leq\lambda(nh)^\beta.
\end{aligned}
\end{equation*}
This proves part (iii).

\paragraph{Polynomial-time circuit construction.}
Let $S$ be the description length of the circuit for $g$.
The circuit computes only the simplex selected by its input; it does
not enumerate the grid or all possible coordinate orderings. To locate the cube, process each coordinate $x_i$ by $m$ successive
bisections. Initialize $u_0=x_i$ and, for $\ell=1,\ldots,m$, compute
\begin{equation*}
b_\ell=\mathbf 1[u_{\ell-1}>1/2],\qquad
u_\ell=2u_{\ell-1}-b_\ell.
\end{equation*}
Here $\mathbf 1[\cdot]$ denotes the indicator of the stated condition.
The bit $b_\ell$ selects a half-interval, and $u_\ell$ gives the
rescaled position within that half. Inductively, $u_\ell\in[0,1]$ and
$x_i=\sum_{a=1}^{\ell}2^{-a}b_a+2^{-\ell}u_\ell$.
Thus $c_i=\sum_{\ell=1}^m2^{-\ell}b_\ell$ and $t_i=u_m$.
The strict comparison assigns a midpoint to its left half, exactly
implementing the chosen grid-boundary convention. This stage uses
$O(nm)$ gates.

To determine the coordinate order, compute the rank of each $t_i$:
\begin{equation*}
\rho_i=1+\sum_{k<i}\mathbf 1[t_k\geq t_i]
          +\sum_{k>i}\mathbf 1[t_k>t_i].
\end{equation*}
This counts the coordinates preceding $i$ in decreasing order, with
smaller indices preceding larger ones in a tie. Hence the ranks are
exactly $1,\ldots,n$. All ranks can be computed using $O(n^2)$ gates.
Non-strict comparisons and equality tests are also available using a
constant number of strict comparisons and arithmetic gates. The sorted coordinates and the vertices can now be obtained without
explicitly outputting the permutation:
\begin{equation*}
t_{\pi(j)}=\sum_{i=1}^n t_i\mathbf 1[\rho_i=j],\qquad
(v_j)_i=c_i+h\mathbf 1[\rho_i\leq j].
\end{equation*}
The first formula applies for $1\leq j\leq n$, and the second for
$0\leq j\leq n$. These formulas and
\eqref{eq:pn-kuhn-weights} compute all vertices and weights using
$O(n^2)$ additional gates. Finally, evaluate $n+1$ copies of $g$ in
parallel at these vertices and form the weighted sum, coordinate by
coordinate. The resulting circuit has
$O(nm+n^2+(n+1)S+nr)$ gates. All new rational constants have
$O(m+\log n)$-bit encodings, so its description and construction time
are polynomial in $n,r,m,S$. In particular, the cost depends on $m$, but not on the $2^{mn}$ grid cubes.


For completeness, consider the well-behavedness condition. Cube selection and rank computation
use no true multiplications. Selecting the sorted coordinates adds one
layer of true multiplications, and forming the weighted sum adds one
more. Since the copies of $g$ are parallel, the true-multiplication
depth is at most that of $g$ plus two. Under the logarithmic
true-multiplication depth convention for well-behaved circuits,
padding the circuit by a constant factor with redundant additions of
zero therefore preserves well-behavedness and the polynomial size
bound. For a rational input of encoding length $B$, the vertices and weights
have encoding length $\mathrm{poly}(n,m,B)$. The well-behaved copies of
$g$ can therefore be evaluated exactly, and the entire computation takes
time in $\mathrm{poly}(n,r,m,S,B)$. This proves the computational claim in
part (i) and completes the construction.
\end{proof}

\begin{figure}[t]
    \centering
    \includegraphics[width=0.5\linewidth]{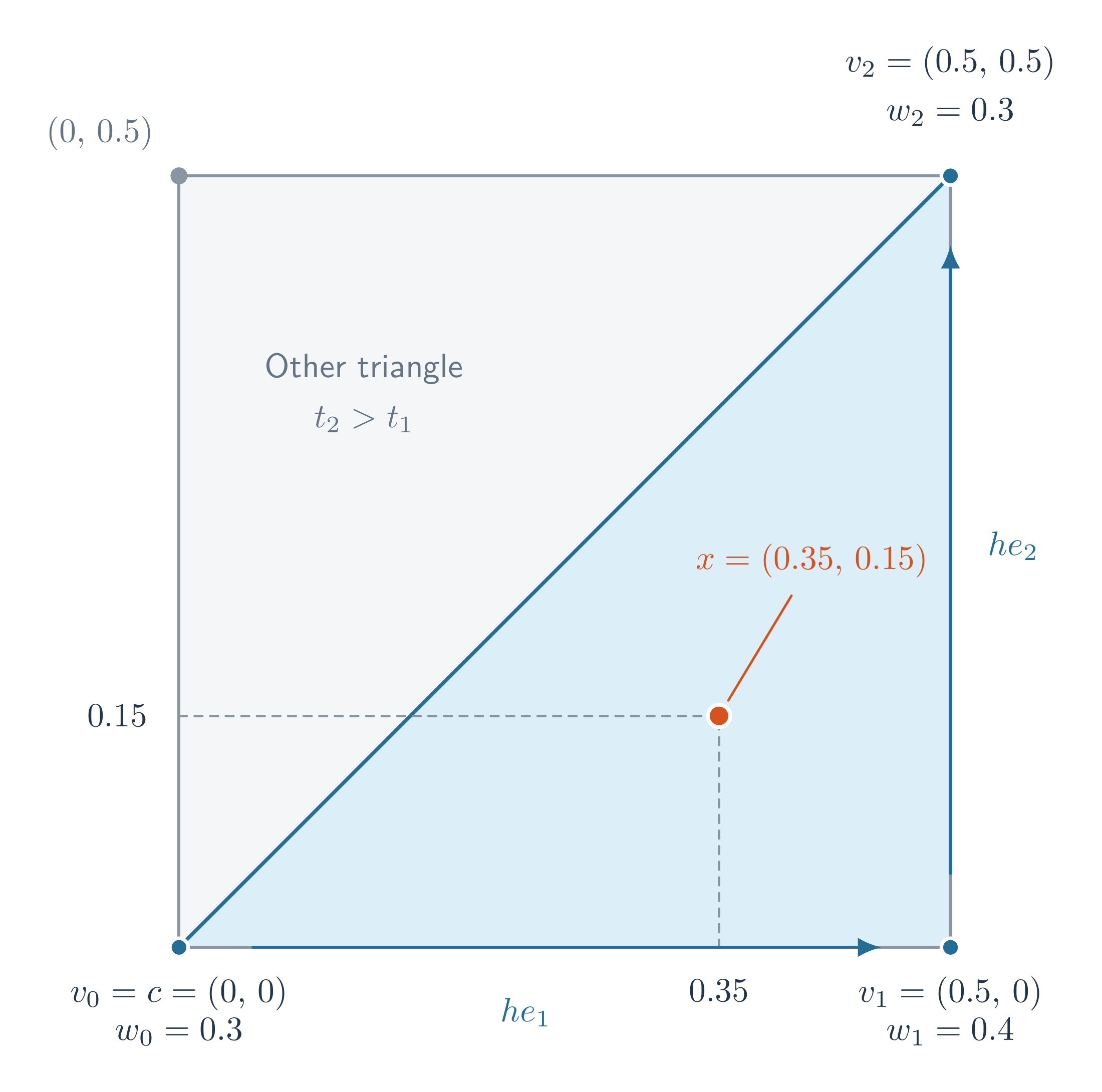}
    \caption{%
        Simplex-based interpolation in two dimensions.
        A grid square with lower corner $c=(0,0)$ and side length
        $h=1/2$ is divided into two triangles.
        The orange point $x=(0.35,0.15)$ lies in the blue triangle
        with vertices $v_0=c$, $v_1=c+he_1$, and
        $v_2=c+h(e_1+e_2)$, where $e_1,e_2$ are the standard
        basis vectors.
        Its normalized coordinates are
        $(t_1,t_2)=(x-c)/h=(0.7,0.3)$, giving weights
        $w_0=1-t_1=0.3$, $w_1=t_1-t_2=0.4$, and $w_2=t_2=0.3$.
        These weights express $x=0.3v_0+0.4v_1+0.3v_2$
        and define the interpolated value
        $I_hg(x)=0.3g(v_0)+0.4g(v_1)+0.3g(v_2)$.
    }
    \label{fig:simplex-interpolation-2d}
\end{figure}
Now we prove the {$\mathsf{CLS}$}-completeness of {$\beta$}-{\textsc{H\"older-LocalOpt}}. In the proof, we assume both {\textsc{Continuous-LocalOpt}} and {$\beta$}-{\textsc{H\"older-LocalOpt}} are in the polynomial dimensional box {$[0,1]^n$} for generality.

\begin{theorem}
\label{thm:holder-localopt-cls}
For every fixed rational $\beta\in(0,1]$, there is a polynomial-time reduction from
$\beta$-\textsc{H\"older-LocalOpt} to  \textsc{Continuous-LocalOpt}.
\end{theorem}

\begin{proof}
Let $(f,p,\delta,\lambda)$ be an instance of
$\beta$-\textsc{H\"older-LocalOpt} on $[0,1]^n$, with $f$ and $p$
represented by well-behaved circuits. If $\delta\geq1$, the point $0$
satisfies \textnormal{(HO1)}, since
$p(f(0))\geq0\geq p(0)-\delta$. Hence assume $0<\delta<1$.

Let $\bar\lambda=\max\{1,\lambda\}$ and choose a grid length
$h=2^{-m}$, with $m\geq1$, such that $nh \le 1$ and
\begin{equation}
3\bar\lambda^2(nh)^{\beta^2}\leq\delta/2.
\label{eq:pn-holder-mesh}
\end{equation}
The second inequality reserves an error allowance of $\delta/2$ for
transferring a solution from the interpolated functions back to $f$
and $p$. We verify below that this grid length can be chosen in polynomial
time.

Set $\widetilde f=I_hf$ and $\widetilde p=I_hp$. By
Lemma~\ref{lem:pn-kuhn}, these functions have the required ranges,
and their Lipschitz constants are at most $n/h$ and $1/h$,
respectively. Construct the \textsc{Continuous-LocalOpt} instance
$(\widetilde f,\widetilde p,\delta/2,n/h)$. Neither
\textnormal{(CO2)} nor \textnormal{(CO3)} is a valid answer to this
instance, because both interpolants satisfy the stated Lipschitz
bound, regardless of the regularity of $f$ and $p$.
Thus every target answer is a point $x$ satisfying
$\widetilde p(\widetilde f(x))\geq\widetilde p(x)-\delta/2$.
Fix such an $x$ and write $z=\widetilde f(x)$.

For a point $u$, let $V(u)$ denote the $n+1$ grid vertices selected
by the interpolation construction at $u$. The vertices depend only
on $u$ and the grid, not on the function being interpolated.
We check the following inequalities:
\begin{enumerate}
    \item $\|f(x)-f(v)\|_1 \leq\lambda\|x-v\|_1^\beta\ \text{for every }v\in V(x),$
    \item $|p(x)-p(v)| \leq\lambda\|x-v\|_1^\beta\ \text{for every }v\in V(x),$
    \item $|p(z)-p(v)| \leq\lambda\|z-v\|_1^\beta\ \text{for every }v\in V(z),$
    \item $|p(f(x))-p(z)|\leq\lambda\|f(x)-z\|_1^\beta.$
\end{enumerate}
If the first comparison fails, return $(x,v)$ as an
\textnormal{(HO2)} answer. Any other failed comparison gives
an \textnormal{(HO3)} answer, using the corresponding pair
$(x,v)$, $(z,v)$, or $(f(x),z)$.

The last comparison is needed in addition to the vertex comparisons:
it controls the change in the original potential $p$ when its argument
is changed from $f(x)$ to the interpolated update $z$.
We now assume that every comparison passes and show that $x$
satisfies \textnormal{(HO1)}. Importantly, we assume only these
finitely many inequalities, not a global H\"older condition.

Let $\eta=\lambda(nh)^\beta$. By Lemma~\ref{lem:pn-kuhn}(iii),
the comparisons for $f$ give $\|f(x)-z\|_1\leq\eta$.
By the same lemma, the second family gives
$|p(x)-\widetilde p(x)|\leq\eta$, and the third gives
$|p(z)-\widetilde p(z)|\leq\eta$.
Finally, the last comparison gives
$|p(f(x))-p(z)|\leq\lambda\eta^\beta$.
These are the three potential errors that must be accounted for:
interpolation at $x$, interpolation at $z$, and changing the update
from $f(x)$ to $z$.

Because $nh\leq1$ and $\beta^2\leq\beta$, we have
$(nh)^\beta\leq(nh)^{\beta^2}$. Also,
$\lambda\leq\bar\lambda^2$ and
$\lambda^{1+\beta}\leq\bar\lambda^2$, since
$\bar\lambda\geq\max\{1,\lambda\}$ and $1+\beta\leq2$.
Consequently, the total error satisfies
\begin{equation}
\begin{aligned}
&|p(x)-\widetilde p(x)|+|p(z)-\widetilde p(z)|
  +|p(f(x))-p(z)|\\
&\qquad\leq 2\eta+\lambda\eta^\beta
 =2\lambda(nh)^\beta
   +\lambda^{1+\beta}(nh)^{\beta^2}
\leq 3\bar\lambda^2(nh)^{\beta^2}
 \leq\delta/2.
\end{aligned}
\label{eq:pn-holder-errors}
\end{equation}
We can therefore transfer the target inequality back to the original
functions one term at a time:
\begin{equation*}
\begin{aligned}
p(f(x))
&\geq p(z)-\lambda\eta^\beta\\
&\geq\widetilde p(z)-\eta-\lambda\eta^\beta\\
&\geq\widetilde p(x)-\delta/2-\eta-\lambda\eta^\beta\\
&\geq p(x)-\delta/2-2\eta-\lambda\eta^\beta\\
&\geq p(x)-\delta.
\end{aligned}
\end{equation*}
The third line uses the target answer and $z=\widetilde f(x)$;
the last line uses \eqref{eq:pn-holder-errors}.
Thus $x$ is an \textnormal{(HO1)} answer.

Now we show that the grid length can be chosen in polynomial time. Write the fixed exponent as $\beta=a/b$, where $a,b$ are fixed
positive integers. A grid length satisfying \eqref{eq:pn-holder-mesh} can be
found by increasing $m$ until both inequalities hold. This requires
only $m=O_\beta(\log n+\log\bar\lambda+\log(1/\delta))$.
Inequality \eqref{eq:pn-holder-mesh} is equivalent to the rational comparison
\begin{equation*}
(3\bar\lambda^2)^{b^2}(n2^{-m})^{a^2}
\leq(\delta/2)^{b^2}.
\end{equation*}
Its exponents are constants, so it can be checked exactly in polynomial
time by \cite[Proposition 3.7]{Etessami2014}. Lemma~\ref{lem:pn-kuhn} then constructs both interpolant circuits
in polynomial time. The parameter $n/h=n2^m$ also has polynomial
encoding length, namely $O(m+\log n)$ bits.

The decoder uses only the vertices in $V(x)$ and $V(z)$, hence
$O(n)$ circuit evaluations and comparisons; it never searches the
whole grid. For any nonnegative rational numbers $A,B$, a comparison
$A\leq\lambda B^\beta$ is equivalent to
$A^b\leq\lambda^bB^a$, so all the H\"older tests can likewise be
performed exactly using fixed integer powers. Well-behavedness and
Lemma~\ref{lem:pn-kuhn} ensure that $f(x)$, $z$, the selected vertices,
and all evaluated values have encoding lengths polynomial in the
instance size and the encoding length of $x$.
Thus the instance construction and the solution decoder are both
polynomial-time. Every target solution yields an answer of type
\textnormal{(HO1)}, \textnormal{(HO2)}, or \textnormal{(HO3)}, proving
membership in $\mathsf{CLS}$.
\end{proof}

\begin{theorem}
\label{thm:holder-localopt-cls-hard}


For every fixed rational $\beta\in(0,1]$, there is a polynomial-time reduction from 
\textsc{Continuous-LocalOpt} to
 $\beta$-\textsc{H\"older-LocalOpt}.
\end{theorem}

\begin{proof}
Let $\mathcal{I}_{\mathrm{CO}}=(f,p,\delta,\lambda)$ be an instance of
\textsc{Continuous-LocalOpt}. We construct the
$\beta$-\textsc{H\"older-LocalOpt} instance $\mathcal{I}_{\mathrm{HO}}
    :=
    (f,p,\delta,n\lambda).$
This construction clearly has polynomial size. An output of type (HO1) is exactly an output of type (CO1). Now suppose
that $(x,y)$ is an output of type (HO2), and set
$t:=\|x-y\|_1$. Since $x,y\in[0,1]^n$, one has $0\leq t\leq n$. For
$t>0$,
$t^{1-\beta}\leq n^{1-\beta}\leq n$, and hence $t\leq nt^\beta$.
Therefore,
\[
    \|f(x)-f(y)\|_1
    >
    n\lambda t^\beta
    \geq
    \lambda t,
\]
so $(x,y)$ is an output of type (CO2). The case $t=0$ cannot satisfy the
strict (HO2) inequality. The same argument shows that every output of
type (HO3) is an output of type (CO3).

Thus every solution of the constructed
{$\beta$}-{\textsc{H\"older-LocalOpt}} instance can be decoded in polynomial
time into a solution of the original {\textsc{Continuous-LocalOpt}}
instance. Since {\textsc{Continuous-LocalOpt}} is {$\mathsf{CLS}$}-complete,
{$\beta$}-{\textsc{H\"older-LocalOpt}} is {$\mathsf{CLS}$}-hard.
\end{proof}


By Theorem {\ref{thm:holder-localopt-cls}} and {\ref{thm:holder-localopt-cls-hard}}, we can conclude the {\textsf{CLS}}-completeness of {$\beta$}-{\textsc{H\"older-}}{\textsc{LocalOpt}}:



\begin{theorem}
\label{thm:holder-localopt-cls-complete}For every fixed rational {$\beta \in (0,1]$}, the problem {$\beta$}-{\textsc{H\"older-LocalOpt}} is {$\mathsf{CLS}$}-complete.


\end{theorem}
\begin{remark}[The exponent as part of the input]
\label{rem:holder-exponent-input}
Our reductions fix {$\beta$} independently
of the input. The same proof allows {$\beta$} to be supplied with the
instance provided that {$1/\beta$} and both its numerator and denominator are polynomially bounded in the input
length. If a binary-encoded exponent is allowed to be
exponentially small, the grid depth $m$ required by
Equation~{\eqref{eq:pn-holder-mesh}} may be exponential, so the proof
above would no longer establish a polynomial-time reduction.


\end{remark}
\subsection{The \textsf{CLS}-membership of $\beta$-\textsc{Projected-Nadler}}
\label{subsec:nadler-to-clo}

 Again, the reduction in this (and the following) section works for any polynomial dimension $n$.
\begin{lemma}
\label{lem:nadler-to-holder-localopt}
For every fixed rational $\beta\in(0,1]$, $\beta$-\textsc{Projected-Nadler} admits a
polynomial-time reduction to $\beta$-\textsc{H\"older-LocalOpt}.
\end{lemma}

\begin{proof}
Consider a well-formed instance of $\beta$-\textsc{Projected-Nadler}.
Let $\Delta=(1-k)\varepsilon_1-\varepsilon_2>0$ and set
\begin{equation}
f_H(x)=P(x),\qquad
p_H(x)=\frac{r(x)}D,\qquad
\delta_H=\frac{\Delta}{D},\qquad
\lambda_H=
\max\left\{L_P,\frac{L_d(n+L_P)}D\right\}.
\label{eq:pn-local-search-parameters}
\end{equation}
The construction ensures $f_H(X)\subseteq X$ and
$p_H(X)\subseteq[0,1]$.
A constant number of circuit compositions is used, so the circuits
remain well-behaved after polynomial padding. All parameters have
polynomial rational encoding length. We decode each possible answer.

\paragraph{Output (HO1).}
Suppose that $p_H(P(x))\geq p_H(x)-\delta_{ H}$.
If $r(x)\leq\varepsilon_1$, return $x$ as an \textnormal{(N1)}
answer. Otherwise, since $p_H=r/D$ and $k<1$, we obtain
\begin{equation*}
\begin{aligned}
r(P(x))
\geq r(x)-D\delta_{ H}
=k\,r(x)+(1-k)(r(x)-\varepsilon_1)+\varepsilon_2
>k\,r(x)+\varepsilon_2.
\end{aligned}
\end{equation*}
This inequality gives an \textnormal{(N2)} witness: choose
$y:=P(x)$ and $z:=x$. Then
$u=\Pi(x,z)=y$ and $v=\Pi(y,u)=P(y)$, so
$d(u,v)=r(P(x))>k\,d(x,y)+\varepsilon_2$.
Thus $(x,y,z)$ is a valid \textnormal{(N2)} answer.

\paragraph{Output (HO2).}
Since $\lambda_H\geq L_P$, a H\"older violation for $f_H=P$
immediately gives (N3).

\paragraph{Output (HO3).}
Suppose
$|p_H(x)-p_H(y)|>\lambda_H\|x-y\|_1^\beta$.
If $x=P(x)$, the definition of residual gives $r(x)=d(x,x)=0$,
so return $x$ as (N1). The same holds if $y=P(y)$ (returning $y$ as (N1)).
Otherwise, both $(x,P(x))$ and $(y,P(y))$ are off-diagonal.
Let $t=\|x-y\|_1>0$. If
$\|P(x)-P(y)\|_1>L_Pt^\beta$, then we obtain a (N3) answer $(x,y)$.
If not, take $t\leq n$ and $t\leq nt^\beta$ to obtain
\begin{equation}
\begin{aligned}
|d(x,P(x))-d(y,P(y))|
&=D|p_H(x)-p_H(y)|\\
&>D\lambda_Ht^\beta
\geq L_d(n+L_P)t^\beta\\
&\geq L_d\bigl(t+\|P(x)-P(y)\|_1\bigr).
\end{aligned}
\label{eq:pn-regularity-decoder}
\end{equation}
Thus $(x,P(x),y,P(y))$ is an (N4) answer.

Clearly, every part of the reduction is computable in polynomial time and respects the polynomial witness size. Together with Theorem~\ref{thm:holder-localopt-cls-complete}, this proves the \textsf{CLS}-membership.
\end{proof}

\begin{remark}[Bypassing the metric promise]
\label{rem:pn-diagonal}
The proof above uses the promise that $d$ is a metric only through $d(x,x)=0$ in
the \textnormal{(HO3)} decoder. Without it, suppose $x=P(x)$. If
$r(x)\le\varepsilon_1$, then $x$ is an \textnormal{(N1)} answer. Otherwise, take
all three points of \textnormal{(N2)} equal to $x$. In this case, $u=\Pi(x,x)=x$ and
$v=\Pi(x,u)=x$, so $d(u,v)=d(x,x)=r(x)$, and
$(1-k)\,r(x)>(1-k)\,\varepsilon_1>\varepsilon_2$ give
$d(u,v)>k\,d(x,x)+\varepsilon_2$. Hence, $(x,x,x)$ is an \textnormal{(N2)} answer.
The same applies to $y$. No other step uses symmetry, the triangle inequality,
metric space completeness, or $\Pi(x,y)\in F(x)$. The ranges $P(X)\subseteq X$ and
$d(X,X) \subseteq  [0,D]$ can be enforced syntactically by coordinate-wise clipping with $\min $ and $\max$ gates,
which preserves well-behavedness and $1$-Lipschitzness. Therefore, subject to the numerical condition \(0<\varepsilon_2<(1-k)\varepsilon_1\) and the (enforced) range conditions \(P(X)\subseteq X\) and \(d(X, X)\subseteq[0,D]\), the membership reduction and the totality of $\beta$-\textsc{Projected-Nadler} hold for
arbitrary well-behaved circuits.
\end{remark}

\subsection{The \textsf{CLS}-hardness of $\beta$-\textsc{Projected-Nadler}}
\label{subsec:clo-to-nadler}

For the reverse reduction, we encode the local-search potential directly into
the metric and use the update circuit as a singleton correspondence.

\begin{lemma}
\label{lem:clo-to-nadler}
For every fixed rational $\beta\in(0,1]$, \textsc{Continuous-LocalOpt} admits a polynomial-time reduction
to $\beta$-\textsc{Projected-Nadler}.
\end{lemma}

\begin{proof}
Let $(f,p,\delta,\lambda)$ be a \textsc{Continuous-LocalOpt} instance
on $X=[0,1]^n$. If $\delta\geq1$, use a fixed target instance of the
singleton form below and decode every answer as the origin.
Assume $0<\delta<1$.
Use the following ultrametric construction:
\begin{equation}
d_p(x,y)=
\begin{cases}
0,&x=y,\\
1+\max\{p(x),p(y)\},&x\neq y.
\end{cases}
\label{eq:pn-potential-ultrametric}
\end{equation}
Set $d=d_p$ and $\Pi(x,z)=f(x)$, independently of $z$, and choose
\begin{equation}
D=2,\qquad
k=1-\frac\delta2,\qquad
\varepsilon_1=\frac12,\qquad
\varepsilon_2=\frac\delta8,\qquad
L_P=n\lambda,\qquad
L_d=\lambda.
\label{eq:pn-hardness-parameters}
\end{equation}
In particular,
$0<\varepsilon_2=\delta/8<\delta/4=(1-k)\varepsilon_1$.
All circuits
and parameters have polynomial description length.

For completeness, $d_p$ is a genuine complete ultrametric.
For pairwise distinct $x,y,z$, the maximum of $d_p(x,z)$ and
$d_p(z,y)$ is at least both $1+p(x)$ and $1+p(y)$, hence at least
$d_p(x,y)$. The cases with repeated points are immediate.
Every nonzero distance is within $[1,2]$, so every $d_p$-Cauchy sequence is
eventually constant. Consequently, the singleton correspondence
$F_f(x)=\{f(x)\}$ has nonempty closed bounded values, and $\Pi$ is
its exact projection in $d_p$.
The contraction parameter remains a claim that can be violated.
We now decode all four output types.

\paragraph{Output (N1).}
Since $P=f$ and every nonzero $d_p$-distance is at least one,
$d_p(x,f(x))\leq1/2$ implies $x=f(x)$.
Thus $p(f(x))=p(x)$ and $x$ satisfies (CO1).

\paragraph{Output (N2).}
The projection circuit ignores its second argument, so the answer gives
$d_p(f(x),f(y))>k d_p(x,y)+\varepsilon_2$.
In particular, $x\neq y$ and $f(x)\neq f(y)$.
If neither endpoint satisfied (CO1), then
$p(f(x))<p(x)-\delta$ and $p(f(y))<p(y)-\delta$, which yields
\begin{equation}
\begin{aligned}
d_p(f(x),f(y))
&=1+\max\{p(f(x)),p(f(y))\}\\
&<1+\max\{p(x),p(y)\}-\delta\\
&=d_p(x,y)-\delta\\
&\leq(1-\delta/2)d_p(x,y)
=k d_p(x,y).
\end{aligned}
\label{eq:pn-hardness-contraction}
\end{equation}
The last inequality uses $d_p(x,y)\leq2$.
This contradicts the buffered certificate.
Evaluate the two source potential inequalities and return an
endpoint satisfying (CO1).

\paragraph{Output (N3).}
Let $t=\|x-y\|_1\in(0,n]$; then we have $t\leq nt^\beta$.
Therefore
$\|f(x)-f(y)\|_1>n\lambda t^\beta\geq\lambda t$,
so the same pair $(x,y)$ satisfies (CO2).

\paragraph{Output (N4).}
Suppose that the returned points satisfy $a\neq b$, $a'\neq b'$, and
\begin{equation*}
|d_p(a,b)-d_p(a',b')|
>
\lambda\bigl(\|a-a'\|_1+\|b-b'\|_1\bigr).
\end{equation*}
Since the two arguments in each metric evaluation are distinct,
both evaluations use the formula
$d_p(u,v)=1+\max\{p(u),p(v)\}$. The common additive constant $1$
therefore cancels upon subtraction, giving
\begin{equation*}
\begin{aligned}
|d_p(a,b)-d_p(a',b')|
&=
\bigl|\max\{p(a),p(b)\}-\max\{p(a'),p(b')\}\bigr|\\
&\leq
\max\bigl\{|p(a)-p(a')|,\ |p(b)-p(b')|\bigr\}.
\end{aligned}
\end{equation*}
Here we use the fact that the change in a maximum is at most the
larger change in its two arguments.

We claim that at least one of the pairs $(a,a')$ and $(b,b')$
violates the $\lambda$-Lipschitz bound for $p$. Otherwise, the
preceding estimate would imply
\begin{equation*}
|d_p(a,b)-d_p(a',b')|
\leq
\lambda\max\{\|a-a'\|_1,\|b-b'\|_1\}
\leq
\lambda\bigl(\|a-a'\|_1+\|b-b'\|_1\bigr),
\end{equation*}
contradicting the returned \textnormal{(N4)} certificate.
The decoder, therefore, checks these two pairs and returns one
satisfying $|p(u)-p(v)|>\lambda\|u-v\|_1$ as an
\textnormal{(CO3)} answer.

Every decoder uses polynomially many rational operations and respects
the polynomial witness size. This proves the reduction.
\end{proof}

Now, the proof of Theorem \ref{thm:nadler-clo-equivalence} naturally follows from Lemma \ref{lem:clo-to-nadler} and \ref{lem:nadler-to-holder-localopt}.

\section{Background and Omitted Proofs for Section \ref{sec:converse-nadler}}
\label{sec:appendix-converse-nadler}

\subsection{Fryszkowski's problem and set-theoretic remetrization}
\label{subsec:fryszkowski}
Classical converses to Banach's contraction principle ask when a self-map can
be made contractive by replacing the ambient metric. Fryszkowski proposed the
corresponding remetrization problem for set-valued maps, as recorded by
Jachymski \cite{Jachymski2000}. For a nonempty set $\Omega$, write
$2^\Omega_\ast:=2^\Omega\setminus\{\emptyset\}$. Given a correspondence
$F:\Omega\to 2^\Omega_\ast$, define its induced action on nonempty subsets by
$\widehat F(A):=\bigcup_{x\in A}F(x)$, and write $\widehat F^{\,n}$ for its
$n$-fold iterate.
\begin{definition}[Fryszkowski's problem]
Given $\alpha\in(0,1)$, a nonempty set $\Omega$, and a correspondence
$F:\Omega\to 2^\Omega_\ast$, determine necessary and/or sufficient conditions
for the existence of a complete metric $d$ on $\Omega$ under which $F$ is a
Nadler $\alpha$-contraction:
\begin{equation}
    H_d(F(x),F(y))
    \leq
    \alpha d(x,y)
    \qquad\text{for all }x,y\in\Omega.
    \label{eq:fryszkowski-target}
\end{equation}    
\end{definition}
Here $H_d$ denotes the generalized Hausdorff distance, which may take
the value $+\infty$ on arbitrary subsets. This is a set-theoretic
remetrization problem: no metric or topology on $\Omega$ is fixed in advance,
and the goal is to identify conditions on the dynamics of $F$ that permit
some complete contracting metric.

The results below concern correspondences with a distinguished
\emph{endpoint} $z$, meaning that $F(z)=\{z\}$. This is stronger than the
usual fixed-point condition $z\in F(z)$. The stronger condition is natural in
a converse theorem that requires every branch of the correspondence to move
toward the same terminal state: once the process reaches $z$, no admissible
update may leave it.

To answer Fryszkowski's problem, Andrei Com\u{a}neci developed a
Lyapunov-type formulation that serves as the starting point for his partial
answers to this problem \cite{Comaneci2017}. The following characterization
applies to every $\alpha\in(0,1)$ and to complete metrics that are not
required to be bounded. It replaces the search for a metric by the search for
a scalar rank function that vanishes exactly at the endpoint and decreases
geometrically along every admissible successor.

\begin{proposition}[Com\u{a}neci's rank-function characterization
    \cite{Comaneci2017}]
\label{prop:comaneci-rank}
Let $\alpha\in(0,1)$, let $\Omega$ be nonempty, and suppose that
$F:\Omega\to 2^\Omega_\ast$ has an endpoint $z$. The following statements are
equivalent.
\begin{enumerate}
    \item There exists a complete metric $d$ on $\Omega$ such that $F$ is a
    Hausdorff $\alpha$-contraction with respect to $d$.

    \item There exists a function $\phi:\Omega\to[0,\infty)$ such that
    $\phi^{-1}(0)=\{z\}$ and
    \begin{equation}
        \sup_{u\in F(x)}\phi(u)
        \leq
        \alpha\phi(x)
        \qquad\text{for every }x\in\Omega.
        \label{eq:comaneci-rank}
    \end{equation}
\end{enumerate}
\end{proposition}

The supremum in \eqref{eq:comaneci-rank} is important: the rank must decrease
for every $u\in F(x)$, rather than only along one selected trajectory. Thus,
Com\u{a}neci's characterization isolates the dynamical content of the
remetrization problem in a single Lyapunov-type condition that controls all
branches of the set-valued iteration.

Com\u{a}neci then used this characterization to obtain the first of the two
partial answers highlighted in his paper: a criterion for the stronger case
in which the witnessing metric is required to be bounded. The relevant
condition is expressed in terms of the descending sequence of global images
\begin{equation}
    \Omega
    =
    \widehat F^{\,0}(\Omega)
    \supseteq
    \widehat F^{\,1}(\Omega)
    \supseteq
    \widehat F^{\,2}(\Omega)
    \supseteq
    \cdots.
    \label{eq:nested-global-images}
\end{equation}
The set $\widehat F^{\,n}(\Omega)$ contains all states that can be reached
after $n$ admissible updates from some initial state in $\Omega$. Hence, the
intersection of these sets consists of the states that remain globally
reachable at arbitrarily large depths.

Com\u{a}neci proved that, for $\alpha\in(0,1/2)$, collapse of these global
images to the endpoint is equivalent to the existence of both a bounded rank
function and a complete bounded contracting metric. Later, Luchian removed the restriction $\alpha<1/2$ and established the same characterization for the full contraction range $\alpha\in(0,1)$ \cite{Luchian2018}.

\begin{theorem}[Com\u{a}neci--Luchian bounded remetrization
    \cite{Comaneci2017,Luchian2018}]
\label{thm:comaneci-luchian}
Let $\alpha\in(0,1)$, let $\Omega$ be nonempty, and let
$F:\Omega\to 2^\Omega_\ast$ have an endpoint $z$. The following statements
are equivalent.
\begin{enumerate}
    \item The nested global images collapse to the endpoint:
    \begin{equation}
        \bigcap_{n\geq0}\widehat F^{\,n}(\Omega)
        =
        \{z\}.
        \label{eq:nested-images-collapse}
    \end{equation}

    \item There exists a bounded function $\phi:\Omega\to[0,\infty)$ such
    that $\phi^{-1}(0)=\{z\}$ and
    \begin{equation}
        \sup_{u\in F(x)}\phi(u)
        \leq
        \alpha\phi(x)
        \qquad\text{for every }x\in\Omega.
        \label{eq:bounded-comaneci-rank}
    \end{equation}

    \item There exists a complete bounded metric $d$ on $\Omega$ such that
    $F$ is a Hausdorff $\alpha$-contraction with respect to $d$.
\end{enumerate}
Com\u{a}neci proved the equivalence for $\alpha\in(0,1/2)$, and Luchian
extended it to every $\alpha\in(0,1)$.
\end{theorem}

The theorem gives a complete set-theoretic criterion for bounded
remetrization under the endpoint assumption: the endpoint must be the only
state that survives every level of the global iteration. Com\u{a}neci's
second partial answer treats finite spaces. He showed that, when $\Omega$ is
finite and $\{z\}$ is the unique nonempty subset fixed by $\widehat F$, a
complete metric making $F$ an $\alpha$-contraction exists for every
$\alpha\in(0,1)$ \cite{Comaneci2017}. Luchian further studied the stronger
equality form of contractivity, namely Nadler $\alpha$-similarities, under an
additional non-overlap assumption \cite{Luchian2018}.

The Com\u{a}neci--Luchian results answer Fryszkowski's question at a purely
mathematical level. They do not begin with a reference metric, and the
constructed metric is not required to generate a prescribed topology or to
be quantitatively comparable with any pre-existing notion of distance.
Consequently, approximation or convergence in the remetrized space need not
translate into approximation or convergence in the metric relevant to an
iterative algorithm. Section~\ref{sec:converse-nadler} strengthens the
remetrization goal by requiring both topological equivalence and
quantitative error transfer.

\subsection{Proof of equivalence relations}
\label{subsec:condition-equivalence}
We prove the equivalence between global-image collapse, uniform convergence
over all initial states, and the two convergence conditions
\eqref{eq:pointwise-robust-convergence} and
\eqref{eq:local-uniform-robust-convergence} together with the endpoint
condition. This result need not be novel and can be covered by more general results (e.g., see {\cite[Theorem 4.5]{ChebanMammana2005}}). Here, we give an independent and self-contained proof.


\begin{lemma}[Equivalent convergence conditions on compact spaces]
\label{lem:equiv-relation}
Let $(X,d)$ be a nonempty compact metric space, let
$F:X\to\mathcal{K}(X)$ be Hausdorff-continuous, and fix $p\in X$.
For $A\subseteq X$, define
$\widehat F(A):=\bigcup_{x\in A}F(x)$, and write $F^{[t]}(x):=\widehat F^{\,t}(\{x\}),$
where $\widehat F^{\,0}$ is the identity. Then the following
conditions are equivalent:

    


    


    


    



\begin{enumerate}
    \item[(i)] The global reachable sets collapse to $\{p\}$:
    \[
    \bigcap_{t\ge 0} \widehat F^{\,t}(X)=\{p\}.
    \]
    \item[(ii)] The reachable sets converge to $\{p\}$ uniformly
    over all initial states:
    \[
    \lim_{t\to\infty}
    \sup_{x\in X}H_d\bigl(F^{[t]}(x),\{p\}\bigr)=0.
    \]

    \item[(iii)] The point $p$ is an endpoint, meaning
    $F(p)=\{p\}$, and
    \[
    \lim_{t\to\infty}H_d\bigl(F^{[t]}(x),\{p\}\bigr)=0
    \qquad\text{for every }x\in X.
    \]
    Moreover, there is an open neighborhood $U$ of $p$ such that
    \[
    \lim_{t\to\infty}
    \sup_{x\in U}H_d\bigl(F^{[t]}(x),\{p\}\bigr)=0.
    \]
\end{enumerate}

\end{lemma}

\begin{proof}
We first establish the compactness and continuity properties used below,
and then explain why the global reachable sets form a nested sequence.

\paragraph{Compactness and continuity.}
For every nonempty compact set $A\subseteq X$, the set
$\widehat F(A)$ is nonempty and compact.
To see this, consider
\[
\Gamma_A:=\{(x,u)\in A\times X:u\in F(x)\}.
\]
Suppose that $(x_j,u_j)\in\Gamma_A$ converges to $(x,u)\in A\times X$.
Since $u_j\in F(x_j)$,
\[
\operatorname{dist}_d(u,F(x))
\le d(u,u_j)+H_d(F(x_j),F(x))
\longrightarrow 0.
\]
The set $F(x)$ is closed, so $u\in F(x)$. Thus $\Gamma_A$ is closed
in the compact space $A\times X$, and hence compact.
Its projection onto the second coordinate is precisely
$\widehat F(A)$, proving the claim. Moreover, compactness of $X$ implies that $F$ is uniformly
Hausdorff-continuous. Define
\[
\omega(r):=
\sup\{H_d(F(x),F(y)):x,y\in X,\ d(x,y)\le r\}.
\]
Then $\omega(r)\to0$ as $r\downarrow0$, and
\[
H_d\bigl(\widehat F(A),\widehat F(B)\bigr)
\le \omega\bigl(H_d(A,B)\bigr)
\qquad
\text{for all }A,B\in\mathcal K(X).
\]
Indeed, for every $x\in A$, compactness of $B$ provides $y\in B$
with $d(x,y)\le H_d(A,B)$. Consequently, every point of $F(x)$
has distance at most $\omega(H_d(A,B))$ from
$F(y)\subseteq\widehat F(B)$. Exchanging the roles of $A$ and $B$
gives the displayed Hausdorff bound.
Therefore, $\widehat F$ is continuous on $\mathcal K(X)$.
Since $x\mapsto\{x\}$ is continuous, every finite iterate
$x\mapsto F^{[t]}(x)$ is also Hausdorff-continuous.

\paragraph{Nestedness of the global reachable sets.}
For each integer $t\ge0$, let
\[
K_t:=\widehat F^{\,t}(X).
\]
The preceding argument shows that every $K_t$ is nonempty and compact.
We now prove that
\[
X=K_0\supseteq K_1\supseteq K_2\supseteq\cdots.
\]
First, $\widehat F$ preserves set inclusion: if $A\subseteq B$, then
\[
\widehat F(A)
=\bigcup_{x\in A}F(x)
\subseteq
\bigcup_{x\in B}F(x)
=\widehat F(B).
\]
Also, $F(x)\subseteq X$ for every $x\in X$, so
$K_1=\widehat F(X)\subseteq X=K_0$.
Now suppose that $K_t\subseteq K_{t-1}$ for some $t\ge1$.
Applying inclusion preservation gives
\[
K_{t+1}
=\widehat F(K_t)
\subseteq
\widehat F(K_{t-1})
=K_t.
\]
The desired nestedness follows by induction.

Equivalently, a point reachable after $t+1$ updates from some point
of $X$ is also reachable after $t$ updates: discard the first update
and regard the resulting state, which still belongs to $X$, as the
initial point. This argument uses the fact that the initial state
in the definition of $K_t$ may be any point of $X$. Since $\widehat F$ preserves unions, we also have
\[
K_t=\bigcup_{x\in X}F^{[t]}(x).
\]
For every nonempty compact set $A\subseteq X$,
$H_d(A,\{p\})=\sup_{a\in A}d(a,p)$. Hence
\[
H_d(K_t,\{p\})
=
\sup_{x\in X}H_d(F^{[t]}(x),\{p\}).
\]
Thus condition \textnormal{(ii)} is equivalent to
$H_d(K_t,\{p\})\to0$.

\paragraph{\textnormal{(i)} $\Rightarrow$ \textnormal{(ii)}.}
Assume that $\bigcap_{t\ge0}K_t=\{p\}$.
Suppose, for a contradiction, that
$H_d(K_t,\{p\})$ does not converge to zero.
Then there exist $\varepsilon>0$, integers $t_j\to\infty$,
and points $z_j\in K_{t_j}$ such that
\[
d(z_j,p)\ge\varepsilon
\qquad\text{for every }j.
\]
By compactness of $X$, we may pass to a subsequence and assume
that $z_j\to z\in X$.

Fix any integer $m\ge0$. For all sufficiently large $j$, we have
$t_j\ge m$, and nestedness gives $K_{t_j}\subseteq K_m$.
Thus $z_j\in K_m$ eventually. Since $K_m$ is closed, its limit
$z$ also belongs to $K_m$. As $m$ was arbitrary,
\[
z\in\bigcap_{m\ge0}K_m=\{p\}.
\]
However, continuity of the distance gives
$d(z,p)\ge\varepsilon$, a contradiction. Therefore
$H_d(K_t,\{p\})\to0$, proving \textnormal{(ii)}.

\paragraph{\textnormal{(ii)} $\Rightarrow$ \textnormal{(i)} and
\textnormal{(iii)}.}
Assume \textnormal{(ii)}. The sets $K_t$ are nested, nonempty,
and compact, so their intersection is nonempty.
If $z$ belongs to this intersection, then, for every $t$,
\[
d(z,p)\le H_d(K_t,\{p\})\longrightarrow0.
\]
Hence $z=p$, proving
$\bigcap_{t\ge0}K_t=\{p\}$ and therefore \textnormal{(i)}.

We next show that $p$ is an endpoint. Since $p\in K_t$ for every $t$,
inclusion preservation yields
\[
F(p)=\widehat F(\{p\})
\subseteq\widehat F(K_t)
=K_{t+1}
\qquad\text{for every }t\ge0.
\]
Consequently,
\[
F(p)\subseteq\bigcap_{t\ge0}K_{t+1}=\{p\}.
\]
Nonemptiness of $F(p)$ implies $F(p)=\{p\}$.
Finally, global uniform convergence implies pointwise convergence
and gives the local uniform convergence required in
\textnormal{(iii)} by taking $U=X$.

\paragraph{\textnormal{(iii)} $\Rightarrow$ \textnormal{(ii)}.}
Assume \textnormal{(iii)}, and let $U$ be the open neighborhood
of $p$ on which convergence is uniform.
For each $x\in X$, pointwise convergence provides an integer
$n_x\ge0$ such that
\[
F^{[n_x]}(x)\subseteq U.
\]
Since $F^{[n_x]}(x)$ is compact and $U$ is open, there exists
$r_x>0$ such that
\[
\{z\in X:
  \operatorname{dist}_d(z,F^{[n_x]}(x))<r_x\}
\subseteq U.
\]
Hausdorff continuity of the finite iterate
$y\mapsto F^{[n_x]}(y)$ therefore provides an open neighborhood
$V_x$ of $x$ such that
\[
F^{[n_x]}(y)\subseteq U
\qquad\text{for every }y\in V_x.
\]
The sets $V_x$ cover $X$. By compactness, choose a finite subcover
$V_{x_1},\ldots,V_{x_m}$, and write
\[
n_i:=n_{x_i},
\qquad
M:=\max_{1\le i\le m}n_i.
\]
Thus every initial state has all its reachable states inside $U$
at one of the finitely many times $n_1,\ldots,n_m$. Fix $\varepsilon>0$. Uniform convergence on $U$ provides an integer
$T\ge0$ such that
\[
\sup_{z\in U}H_d(F^{[s]}(z),\{p\})
<\frac{\varepsilon}{2}
\qquad\text{for every }s\ge T.
\]
Now let $y\in X$ and $t\ge M+T$.
Choose $i$ with $y\in V_{x_i}$. Then
$F^{[n_i]}(y)\subseteq U$ and $t-n_i\ge T$.
Splitting the iteration at time $n_i$ gives
\[
F^{[t]}(y)
=
\bigcup_{z\in F^{[n_i]}(y)}F^{[t-n_i]}(z).
\]
Every set in this union starts from a point of $U$ and has evolved
for at least $T$ further updates. Therefore
\[
H_d(F^{[t]}(y),\{p\})
\le
\sup_{z\in U}H_d(F^{[t-n_i]}(z),\{p\})
<\frac{\varepsilon}{2}.
\]
Since $M$ and $T$ do not depend on $y$, it follows that
\[
\sup_{y\in X}H_d(F^{[t]}(y),\{p\})
\le\frac{\varepsilon}{2}<\varepsilon
\qquad\text{for every }t\ge M+T.
\]
This proves global uniform convergence and hence
\textnormal{(ii)}.
\end{proof}

\subsection{Outline for proof of Theorem \ref{thm:daskalakis-nadler-converse}}

We first give an outline for our proof. The proof adapts the remetrization scheme of Meyers and of Daskalakis,
Tzamos, and Zampetakis
\cite{Meyers1967,DaskalakisTzamosZampetakis2018} to a set-valued iteration. In the
single-valued setting, one follows a single orbit of the update map. Here,
every part of the construction must instead control the full reachable sets
and every admissible successor. The proof is organized into five steps.
Step~1 turns the convergence assumptions into a nested global filtration of
the state space. Step~2 constructs an equivalent metric under which the
correspondence is nonexpansive. Steps~3 and~4 convert progress through the
filtration into strict Hausdorff contraction. Step~5 shows that small
distance in the constructed metric still yields a useful approximation
guarantee in the original metric.

\paragraph{Step 1: Uniform collapse of the global reachable sets.}
We first upgrade the pointwise convergence of all branches to a uniform
statement about the global images of the correspondence. The local
uniformity assumption is used to construct an open neighborhood $W$ of the
endpoint $x^*$ that is forward invariant under $F$. Hausdorff continuity
implies that the property of entering $W$ after finitely many steps is
stable under small perturbations of the initial point, and compactness then
turns the resulting point-dependent entrance times into one common entrance
time for all of $X$. Consequently, the sets
$K_t:=\widehat F^{\,t}(X)$ form a nested sequence of nonempty compact sets
satisfying
$H_d(K_t,\{x^*\})\to0$ and
$\bigcap_{t\geq0}K_t=\{x^*\}$. These sets provide the dynamical layers used
later to measure how far a point has progressed toward the endpoint.

\paragraph{Step 2: A Hausdorff-nonexpansive equivalent metric.}
Starting from $d$, we recursively enlarge the metric to incorporate
the Hausdorff distance between the outputs of $F$ measured using the metric constructed at the previous stage, and then take the monotone supremum $\overline d$
of these auxiliary metrics. This is the
set-valued analogue of the nonexpansive closure in the converse of
Daskalakis, Tzamos, and Zampetakis \cite[Theorem 1]{DaskalakisTzamosZampetakis2018}.
Lemma~\ref{lem:monotone-hausdorff-limit} allows the supremum to pass through
the Hausdorff construction, giving
$H_{\overline d}(F(x),F(y))\leq\overline d(x,y)$. Thus, $F$ is nonexpansive
with respect to $\overline d$. The uniform collapse obtained in Step~1 is
then used to show that the infinite metric closure does not change the
topology. Hence $\overline d$ is topologically equivalent to $d$, while
$\overline d\geq d$ preserves quantitative control over the original notion
of distance.

\paragraph{Step 3: Dynamical rank and geodesic distance.}
Choose $N$ so that the global image $K_N$ lies inside the $d$-ball of radius
$\varepsilon$ around $x^*$. For $x\neq x^*$, define the depth $\tau(x)$ to
be the largest index for which $x\in K_{\tau(x)}$, and shift this depth by
setting $\ell(x):=\tau(x)-N$. Since
$\widehat F(K_t)=K_{t+1}$, every admissible transition $u\in F(x)$ increases
the rank by at least one. We encode this progress in the rank-weighted
function
$\rho(x,y):=c^{\min\{\ell(x),\ell(y)\}}\overline d(x,y)$. The function $\rho$ need not satisfy the triangle inequality, so we replace
it by its geodesic distance $D$.

We prove that $D$ separates distinct points by obtaining a positive
lower bound on the total $\rho$-length of every connecting sequence.
We consider separately sequences that avoid a suitable deeper set
$K_r$ and those that enter it. For the latter, we bound the length
accumulated up to the first entry. The upper bound
$D\leq c^{-N}\overline d$, together with compactness, then implies
that $D$ generates the same topology as $\overline d$ and $d$.

\paragraph{Step 4: Establish contraction under the metric $D$.}
The nonexpansiveness established in Step~2 implies that, for every
$u\in F(x)$, one can choose $v\in F(y)$ such that
$\overline d(u,v)\leq \overline d(x,y)$. Since every point in $F(x)$ has
rank at least one greater than that of $x$, the rank weighting strengthens
this estimate to $\rho(u,v)\leq c\rho(x,y)$.

We apply this estimate successively to any finite sequence
$x=x_0,\ldots,x_m=y$ considered in the definition of $D(x,y)$. Starting
from an arbitrary $u_0\in F(x)$, we choose $u_i\in F(x_i)$ so that
$\rho(u_{i-1},u_i)\leq c\rho(x_{i-1},x_i)$ for every $i$. Summing these
inequalities and taking the infimum over all such sequences shows that
every point of $F(x)$ has $D$-distance at most $cD(x,y)$ from $F(y)$.
Interchanging $x$ and $y$ gives the same estimate in the opposite
direction. Therefore,
\[
    H_D(F(x),F(y))\leq cD(x,y).
\]
\paragraph{Step 5: Transfer approximation guarantees to the original
metric.}
Let $S:=K_N$. By construction, $S$ is forward invariant and is contained
in $B_d(x^*,\varepsilon)$. Every point outside $S$ has negative shifted
rank. Consequently, before a finite sequence of intermediate points first
reaches $S$, each term $\rho(z_{i-1},z_i)$ in the definition of $D$ is no
smaller than $\overline d(z_{i-1},z_i)$. Thus, a small value of
$D(x,x^*)$ is possible only if $x$ already belongs to $S$ or is close to
$S$ under $\overline d$.

The same reasoning applies to two points $x$ and $y$. Any finite sequence
used in the definition of $D(x,y)$ either avoids $S$, in which case its
total $\rho$-value is at least $\overline d(x,y)$, or reaches $S$, in
which case its initial portion is at least the $\overline d$-distance from
one endpoint to $S$. Hence, a small value of $D(x,y)$ forces either $x$
and $y$ to be close to each other or one of them to be close to $S$.

Finally, because $S\subseteq B_d(x^*,\varepsilon)$ and
$\overline d\geq d$, closeness to $S$ implies closeness to $x^*$ in the
original metric. This yields the endpoint and pairwise transfer bounds in
\eqref{eq:endpoint-transfer}--\eqref{eq:pairwise-transfer}, ensuring that
approximation under the constructed metric also gives a quantitative
guarantee under $d$.

\subsection{Topological facts}
\label{}
\begin{lemma}[Continuity properties of compact-valued correspondences]
\label{lem:appendix-conotunity-Fhat}
Let $(X,d)$ be a compact metric space, and let
$F\colon X\to\mathcal{K}(X)$ be Hausdorff-continuous. For every
$A\in\mathcal{K}(X)$, define
$\widehat F(A):=\bigcup_{x\in A}F(x)$. Then the following statements
hold.
\begin{enumerate}
    \item[(i)] The set $\widehat F(A)$ is nonempty and compact for every
    $A\in\mathcal{K}(X)$. Moreover, the induced map
    $\widehat F\colon\mathcal{K}(X)\to\mathcal{K}(X)$ is continuous
    with respect to the Hausdorff metric. Consequently, the map
    $x\mapsto F^{[t]}(x)$ is Hausdorff-continuous for every
    $t\geq 0$.

    \item[(ii)] Let $G\colon X\to\mathcal{K}(X)$ be Hausdorff-continuous,
    and let $V\subseteq X$ be open. Then the set
    $\{x\in X:G(x)\subseteq V\}$ is open in $X$.
\end{enumerate}
\end{lemma}

\begin{proof}
We first prove part~(i). Fix $A\in\mathcal{K}(X)$. Since both $A$ and every value $F(x)$ are nonempty, the union $\widehat F(A)$ is nonempty. Because $X$ is compact, it is enough to show that $\widehat F(A)$ is closed.

Let $(y_n)_{n\geq 1}$ be a sequence in $\widehat F(A)$ converging to
some $y\in X$. For each $n$, choose $x_n\in A$ such that
$y_n\in F(x_n)$. Since $A$ is compact, there is a subsequence
$(x_{n_j})_{j\geq 1}$ converging to some $x\in A$. Hausdorff
continuity of $F$ gives
$H_d(F(x_{n_j}),F(x))\to 0$. Since $y_{n_j}\in F(x_{n_j})$, we have
\begin{equation}
    \operatorname{dist}_d(y,F(x))
    \leq
    d(y,y_{n_j})
    +
    \operatorname{dist}_d(y_{n_j},F(x))
    \leq
    d(y,y_{n_j})
    +
    H_d(F(x_{n_j}),F(x)).
\end{equation}
The right-hand side converges to zero, and hence
$\operatorname{dist}_d(y,F(x))=0$. Because $F(x)$ is compact, and
therefore closed, it follows that $y\in F(x)$. Since $x\in A$, we
conclude that $y\in\widehat F(A)$. Thus $\widehat F(A)$ is closed in
the compact space $X$, so it is compact.

We next prove that $\widehat F$ is Hausdorff-continuous. Since $F$ is
a continuous map from the compact space $X$ into the metric space
$(\mathcal{K}(X),H_d)$, it is uniformly continuous by the Heine–Cantor theorem. Thus, for every
$\varepsilon>0$, there exists $\delta>0$ such that
$d(a,b)<\delta$ implies
$H_d(F(a),F(b))<\varepsilon/2$. Let $A,B\in\mathcal{K}(X)$ satisfy $H_d(A,B)<\delta$. Fix
$u\in\widehat F(A)$. There is some $a\in A$ such that $u\in F(a)$.
Because $B$ is compact, there exists $b\in B$ satisfying
$d(a,b)=\operatorname{dist}_d(a,B)$. In particular,
$d(a,b)<\delta$. Therefore,
\begin{equation}
    \operatorname{dist}_d(u,\widehat F(B))
    \leq
    \operatorname{dist}_d(u,F(b))
    \leq
    H_d(F(a),F(b))
    <
    \frac{\varepsilon}{2}.
\end{equation}
Taking the supremum over $u\in\widehat F(A)$ bounds the directed
Hausdorff distance from $\widehat F(A)$ to $\widehat F(B)$ by
$\varepsilon/2$. Interchanging $A$ and $B$ gives the same bound in
the opposite direction. Hence
$H_d(\widehat F(A),\widehat F(B))\leq\varepsilon/2<\varepsilon$.
This proves that $\widehat F$ is Hausdorff-continuous.

Finally, define $\iota\colon X\to\mathcal{K}(X)$ by
$\iota(x):=\{x\}$. This map is an isometry because
$H_d(\{x\},\{y\})=d(x,y)$. By the definition of the reachable sets,
$F^{[t]}(x)=\widehat F^{\,t}(\iota(x))$. Since $\widehat F$ and
$\iota$ are continuous, the map $x\mapsto F^{[t]}(x)$ is
Hausdorff-continuous for every $t\geq 0$.

We now prove part~(ii). Let
$E:=\{x\in X:G(x)\subseteq V\}$, and fix $x_0\in E$. We show that
$x_0$ is an interior point of $E$. If $V=X$, then $E=X$, so the conclusion is immediate. Suppose
therefore that $V\neq X$. Since $G(x_0)$ is compact,
$X\setminus V$ is closed, and these two sets are disjoint, their
distance
$\eta:=\operatorname{dist}_d(G(x_0),X\setminus V)$ is strictly
positive.

By Hausdorff continuity of $G$ at $x_0$, there exists a neighborhood
$U$ of $x_0$ such that
$H_d(G(x),G(x_0))<\eta/2$ for every $x\in U$. Fix $x\in U$ and
$y\in G(x)$. Then
$\operatorname{dist}_d(y,G(x_0))<\eta/2$. If $y\notin V$, then
$y\in X\setminus V$, and the definition of $\eta$ would imply
$\operatorname{dist}_d(y,G(x_0))\geq\eta$, which is a contradiction.
Therefore $y\in V$. Since this holds for every $y\in G(x)$, we have
$G(x)\subseteq V$. Thus, $U \subseteq E$, so $x_0$ is an interior point of $E$. Since
$x_0\in E$ was arbitrary, the set $E$ is open.
\end{proof}

\begin{lemma}[Monotone Hausdorff-limit identity]
\label{lem:monotone-hausdorff-limit}


Let $(X,\tau)$ be a compact topological space, and let
$(d_m)_{m\geq0}$ be metrics on $X$ with a common finite upper bound.
Assume that $d_m\leq d_{m+1}$ pointwise and that every
$d_m:X\times X\to\mathbb R$ is continuous for the product topology
$\tau\times\tau$.
Suppose that
\[\overline d(x,y):=\sup_{m\geq 0}d_m(x,y)\] is finite for every
$x,y\in X$. Then, for every pair of nonempty compact sets
$A,B\subseteq X$,
\[H_{\overline d}(A,B)=\sup_{m\geq 0}H_{d_m}(A,B).\]
\end{lemma}

\begin{proof}
Continuity of $d_0$ makes the identity map from $(X,\tau)$ to
$(X,d_0)$ continuous. It is a bijection from a compact space to a
Hausdorff space, and hence a homeomorphism. Thus $\tau$ is metrizable,
and every $\tau$-compact subset is sequentially compact. All compactness
and continuity arguments below use this common topology.

We first prove the corresponding identity for point-to-set distances.
Fix $a\in X$ and a nonempty compact set $B\subseteq X$. Set
$\delta_m:=\operatorname{dist}_{d_m}(a,B)$ and
$\delta:=\operatorname{dist}_{\overline d}(a,B)$. Since $d_m\leq\overline d$ pointwise, one has
$\delta_m\leq\delta$ for every $m$. Hence,
$\sup_{m\geq 0}\delta_m\leq\delta$.

We then prove the reverse inequality. Because $d_m\leq d_{m+1}$,
the sequence $(\delta_m)_{m\geq 0}$ is nondecreasing. Let
$L:=\sup_{m\geq 0}\delta_m$, so that $\delta_m\to L$ as
$m\to\infty$. For every $m$, the function
$b\mapsto d_m(a,b)$ is continuous on the compact set $B$.
Therefore, there exists $b_m\in B$ such that
$d_m(a,b_m)=\delta_m$.

By compactness of $B$, the sequence $(b_m)_{m\geq 0}$ has a
subsequence $(b_{m_j})_{j\geq 1}$ converging to some $b^*\in B$.
Fix an arbitrary $r\geq 0$. Since $m_j\to\infty$, one has
$m_j\geq r$ for all sufficiently large $j$. Monotonicity of the
metrics then gives
$d_{m_j}(a,b_{m_j})\geq d_r(a,b_{m_j})$. Consequently,
\begin{equation}
L
=
\lim_{j\to\infty}d_{m_j}(a,b_{m_j})
\geq
\lim_{j\to\infty}d_r(a,b_{m_j})
=
d_r(a,b^*),
\end{equation}
where the final equality follows from continuity of $d_r$.

Since $r$ was arbitrary, it follows that
$L\geq\sup_{r\geq 0}d_r(a,b^*)=\overline d(a,b^*)$.
Because $b^*\in B$, we also have
$\overline d(a,b^*)\geq\operatorname{dist}_{\overline d}(a,B)
=\delta$. Therefore $L\geq\delta$. Together with the previously
proved inequality $L\leq\delta$, this gives
\[
\operatorname{dist}_{\overline d}(a,B)
=
\sup_{m\geq 0}\operatorname{dist}_{d_m}(a,B).
\]

For a metric $\rho$, write
$h_\rho(A,B):=\sup_{a\in A}\operatorname{dist}_\rho(a,B)$ for the
directed Hausdorff distance from $A$ to $B$. Applying the
point-to-set identity to each $a\in A$ gives
\begin{equation}
h_{\overline d}(A,B)
=
\sup_{a\in A, m\ge 0}
\operatorname{dist}_{d_m}(a,B)
=
\sup_{m\geq 0}h_{d_m}(A,B).
\end{equation}
The interchange of the two suprema is valid because both sides are
the supremum of
$\operatorname{dist}_{d_m}(a,B)$ over
$(a,m)\in A\times\mathbb N_0$. By the same argument,
$h_{\overline d}(B,A)=\sup_{m\geq 0}h_{d_m}(B,A)$.

Finally, using the definition of the Hausdorff distance,
\begin{equation}
H_{\overline d}(A,B)
=
\max\left\{
\sup_{m\geq 0}h_{d_m}(A,B),
\sup_{m\geq 0}h_{d_m}(B,A)
\right\}
=
\sup_{m\geq 0}H_{d_m}(A,B).
\end{equation}
Indeed, for any two real sequences $(\alpha_m)$ and $(\beta_m)$,
\[\max\left\{\sup_m\alpha_m,\sup_m\beta_m\right\}
=\sup_m\max\{\alpha_m,\beta_m\}.\] This proves the identity.
\end{proof}

\subsection{Proof of Theorem \ref{thm:daskalakis-nadler-converse}}
For readability, write $D=D_{c,\varepsilon}$. The construction has three
components: a Hausdorff-nonexpansive closure of the original metric, a
dynamical rank that increases under every admissible transition, and the
geodesic distance of the resulting rank-weighted distance.

\paragraph{Step 1: Uniform collapse of the global reachable sets.}
By Lemma~\ref{lem:appendix-conotunity-Fhat}, we know that every finite iterate $x\mapsto F^{[t]}(x)$ is Hausdorff-continuous.

We first construct a forward-invariant neighborhood of $x^*$ analogous to the construction by Daskalakis and Meyers \cite{DaskalakisTzamosZampetakis2018, Meyers1967}. Choose an
open set $V$ whose closure is contained in $U$, and choose $r_0>0$ such that
$B_d(x^*,r_0)\subseteq V$. By
\eqref{eq:local-uniform-robust-convergence}, there is $T_0\geq1$ such that
$\widehat F^{\,T_0}(U)\subseteq B_d(x^*,r_0)\subseteq V$. Define
\begin{equation}
    W
    :=
    \bigcap_{j=0}^{T_0-1}
    \left\{
        x\in X:
        F^{[j]}(x)\subseteq V
    \right\}.
    \label{eq:robust-trapping-neighborhood}
\end{equation}

By Lemma~\ref{lem:appendix-conotunity-Fhat}, each set $\{x : F^{[j]}(x)\subseteq V\}$ is open, so $W$ is open as a finite intersection of open sets. Also, $x^*\in W$ since $F^{[j]}(x^*)=\{x^*\}\subseteq V$ for every $j$. Observe first that $W\subseteq V\subseteq U$, since the $j=0$ term in \eqref{eq:robust-trapping-neighborhood}
gives $\{x\}\subseteq V$. Let $x\in W$ and $u\in F(x)$. Because $\{u\}\subseteq F(x)$ and
$\widehat{F}$ is monotone with respect to inclusion,
\[
F^{[j]}(u)=\widehat{F}^{\,j}(\{u\})\subseteq\widehat{F}^{\,j}(F(x))=F^{[j+1]}(x)
\qquad\text{for every } j\ge 0 .
\]
For $0\le j\le T_0-2$, the right-hand side lies in $V$ by the definition of $W$. For
$j=T_0-1$, we use $x\in U$ and the choice of $T_0$:
\[
F^{[T_0]}(x)\subseteq\widehat{F}^{\,T_0}(U)\subseteq B_d(x^*,r_0)\subseteq V .
\]
Hence $F^{[j]}(u)\subseteq V$ for all $0\le j\le T_0-1$, which implies $u\in W$. This proves
$\widehat{F}(W)\subseteq W$.

For every $x\in X$, robust global convergence gives an integer $m_x$ such
that $F^{[m_x]}(x)\subseteq W$. By applying (ii) of Lemma~\ref{lem:appendix-conotunity-Fhat} to $G=F^{[m_x]}$ with the open set $W$ in the role of $V$, the inclusion $F^{[m_x]}(x')\subseteq W$ holds for all $x'$ in a neighborhood $O_x$ of $x$. Compactness provides a
finite subcover $O_{x_1},\ldots,O_{x_s}$. Let
$M:=\max_i m_{x_i}$. Since $W$ is forward invariant,
$\widehat F^{\,M}(X)\subseteq W$. It follows from the local uniform
convergence on $W\subseteq U$ that, for
$K_t:=\widehat F^{\,t}(X)$,
\begin{equation}
    H_d(K_t,\{x^*\})\longrightarrow0.
    \label{eq:global-image-collapse}
\end{equation}
The sets $K_t$ are nonempty and compact, satisfy
$K_{t+1}\subseteq K_t$, and obey
$\bigcap_{t\geq0}K_t=\{x^*\}$.

\paragraph{Step 2: A Hausdorff-nonexpansive equivalent metric.}
Define a sequence of metrics by $d_0:=d$ and
\begin{equation}
    d_{m+1}(x,y)
    :=
    \max\left\{
        d(x,y),
        H_{d_m}(F(x),F(y))
    \right\},
    \qquad
    \overline d(x,y):=\sup_{m\geq0}d_m(x,y).
    \label{eq:behavioral-metric-construction}
\end{equation}
The pullback $(x,y)\mapsto H_{d_m}(F(x),F(y))$ is a pseudometric, so taking
maximum with $d$ creates a metric. Moreover, $d_m\leq d_{m+1}$ and
$d_m(x,y)\leq\operatorname{diam}_d(X)$ for all $m,x,y$. Hence
$\overline d$ is a finite metric satisfying $\overline d\geq d$.

We show by induction on $m$ that each $d_m$ is continuous on $X\times X$ (with the
product topology of $d$). This holds for $d_0=d$. Assume $d_m$ is continuous. Since
$X\times X$ is compact, $d_m$, which maps from $(X\times X, d(a,a')+d(b,b'))$ to $(\mathbb{R}, |\cdot|)$, is uniformly continuous. Namely, for every $\eta>0$, there is $\delta>0$ such that
\[
d(a,a')+d(b,b')<\delta\ \Longrightarrow\ |d_m(a,b)-d_m(a',b')| < \eta\qquad \text{for all } a,a',b,b' \in X.
\]
Also, $d_m(b,b)=0$. Now if we let $a'=b'=b$, for every
$\eta>0$, there is $\delta>0$ such that
\begin{equation}\label{eq:dm-small}
d(a,b)<\delta\ \Longrightarrow\ d_m(a,b)=|d_m(a,b)-d_m(b,b)|<\eta .
\end{equation}
Let $A,B\in\mathcal{K}(X)$ with $H_d(A,B)<\delta$. For $a\in A$, compactness of $B$
provides $b\in B$ with $d(a,b)=\operatorname{dist}_d(a,B)<\delta$, so
$\operatorname{dist}_{d_m}(a,B)\le d_m(a,b)<\eta$ by \eqref{eq:dm-small}.
Interchanging $A$ and $B$ gives
\begin{equation}\label{eq:Hd-to-Hdm}
H_d(A,B)<\delta\ \Longrightarrow\ H_{d_m}(A,B)\le\eta .
\end{equation}
Set $\Phi_m(x,y):=H_{d_m}(F(x),F(y))$. Since $H_{d_m}$ is a pseudometric on
$\mathcal{K}(X)$, the triangle inequality gives
\[
|\Phi_m(x,y)-\Phi_m(x',y')|\le H_{d_m}(F(x),F(x'))+H_{d_m}(F(y),F(y')).
\]
If $(x',y')\to(x,y)$, then $H_d(F(x'),F(x))\to 0$ and $H_d(F(y'),F(y))\to 0$ by
Hausdorff continuity of $F$, so both terms on the right tend to $0$ by
\eqref{eq:Hd-to-Hdm}. Thus $\Phi_m$ is continuous, and so is
$d_{m+1}=\max\{d,\Phi_m\}$. In particular, as $d\le d_m$, each $d_m$ is topologically
equivalent to $d$, so $\mathcal{K}(X)$ is the same family under $d$ and $d_m$ and
$H_{d_m}$ is a genuine metric on it.

By Lemma \ref{lem:monotone-hausdorff-limit}, for all nonempty compact $A,B\subseteq X$, 
\begin{equation}
    H_{\overline d}(A,B)
    =
    \sup_{m\geq0}H_{d_m}(A,B).
    \label{eq:hausdorff-monotone-limit}
\end{equation}
As a consequence, we can conclude that
\begin{equation}
    H_{\overline d}(F(x),F(y))
    \leq
    \overline d(x,y)
    \qquad\text{for all }x,y\in X.
    \label{eq:hausdorff-nonexpansive}
\end{equation}

It remains to compare $\overline d$ with $d$. Fix $\gamma>0$. By
\eqref{eq:global-image-collapse}, choose $q$ such that
$\operatorname{diam}_d(K_q)\leq\gamma$. Since $K_q$ is forward invariant, an
induction on $m$ gives
$\operatorname{diam}_{d_m}(K_q)\leq\gamma$ for every $m$.
We claim that for every $j\in\{0,\dots,q\}$, all $x,y\in K_j$, and all $p\ge 0$,
\begin{equation}\label{eq:descent}
d_{q-j+p}(x,y)\le\max\{d_{q-j}(x,y),\gamma\}.
\end{equation}
Since $K_0=X$, the case $j=0$ is
\begin{equation}\label{eq:16}
d_{q+p}(x,y)\le\max\{d_q(x,y),\gamma\}\qquad\text{for all } x,y\in X,\ p\ge 0 .
\end{equation}
We prove \eqref{eq:descent} by induction on $j$. For $j=q$, it is the
diameter bound: $x,y\in K_q$ gives $d_p(x,y)\le\operatorname{diam}_{d_p}(K_q)\le\gamma$.
Assume the claim for $j+1$ and let $x,y\in K_j$ and $p\ge 1$ (for $p=0$ there is
nothing to prove). By definition $\widehat{F}(K_j)=K_{j+1}$, so
$F(x),F(y)\subseteq K_{j+1}$, and the induction hypothesis (with the same $p$) gives,
for all $u\in F(x)$ and $v\in F(y)$,
\[
d_{q-j-1+p}(u,v)\le\max\{d_{q-j-1}(u,v),\gamma\}.
\]
Therefore, we can also conclude that
\[
H_{d_{q-j-1+p}}(F(x),F(y))\le\max\bigl\{H_{d_{q-j-1}}(F(x),F(y)),\gamma\bigr\},
\]
by taking the infimum over $v\in F(y)$, then the supremum over
$u\in F(x)$, and repeating with $x$ and $y$ interchanged. Consequently, we have
\begin{align*}
 d_{q-j+p}(x,y)
&=\max\bigl\{d(x,y),H_{d_{q-j-1+p}}(F(x),F(y))\bigr\}\\
&\le\max\bigl\{d(x,y),H_{d_{q-j-1}}(F(x),F(y)),\gamma\bigr\}\\
&=\max\{d_{q-j}(x,y),\gamma\},   
\end{align*}
which is \eqref{eq:descent} for $j$.
Since $d_q$ is continuous and $d_q(x,x)=0$,
\eqref{eq:16} implies that
$d(x,y)\to0$ entails $\overline d(x,y)\to0$. Together with
$d\leq\overline d$, this proves that $d$ and $\overline d$ generate the same
topology. Thus $(X,\overline d)$ is compact and complete.

\paragraph{Step 3: Dynamical rank and geodesic distance.}
Choose $N$ sufficiently large that
$K_N\subseteq B_d(x^*,\varepsilon)$. For $x\neq x^*$, define its
global-image depth by
\begin{equation}
    \tau(x)
    :=
    \max\left\{
        n\in\mathbb{N}_0:
        x\in K_n
    \right\},
    \qquad
    \tau(x^*):=+\infty,
    \qquad
    \ell(x):=\tau(x)-N.
    \label{eq:dynamical-rank}
\end{equation}
The construction of $\tau$ is motivated by the construction of $n_x$ in \cite[Theorem 3.3]{Luchian2018} by Luchian. The maximum is finite for $x\neq x^*$ because
$\bigcap_nK_n=\{x^*\}$. Also, $\ell(x)\geq-N$ for every $x$. Most
importantly, every admissible transition increases the rank: if
$u\in F(x)$, then $x\in K_{\tau(x)}$ implies
$u\in\widehat F(K_{\tau(x)})=K_{\tau(x)+1}$, and hence
\begin{equation}
    \ell(u)\geq\ell(x)+1.
    \label{eq:rank-increase}
\end{equation}
Define the symmetric rank-weighted distance
\begin{equation}
    \rho(x,y)
    :=
    c^{\min\{\ell(x),\ell(y)\}}\overline d(x,y),
    \label{eq:rank-weighted-distance}
\end{equation}
with the convention $c^{+\infty}=0$. The function $\rho$ need not satisfy
the triangle inequality. We therefore take its geodesic distance:
\begin{equation}
    D(x,y)
    :=
    \inf_{\substack{m\geq1\\x=x_0,\ldots,x_m=y}}
    \sum_{i=1}^{m}\rho(x_{i-1},x_i).
    \label{eq:geodesic-closure}
\end{equation}
Symmetry and the triangle inequality are immediate for $D$. We next verify positivity of $D$.

Suppose first that $x\neq y$ and
$\tau(x)\leq\tau(y)<+\infty$. Consider an arbitrary chain (or sequence) from $x$ to $y$: $ x=z_0,z_1,\ldots,z_m=y.$ We call each consecutive pair $(z_{i-1},z_i)$ a link of the chain. Its
$\rho$-length is
\[
    \rho(z_{i-1},z_i)
    =
    c^{\min\{\ell(z_{i-1}),\ell(z_i)\}}
    \overline d(z_{i-1},z_i).
\]
The factor
$c^{\min\{\ell(z_{i-1}),\ell(z_i)\}}$ will be called the rank multiplier
of the link. Set $r:=\tau(y)+1$. If $z\notin K_r$, then
$\tau(z)\leq\tau(y)$ and hence $\ell(z)\leq\ell(y)$. Since $0<c<1$,
this implies $c^{\ell(z)}\geq c^{\ell(y)}$.

There are now two cases. First, suppose that the chain does not meet
$K_r$. Then every point $z_i$ lies outside $K_r$, so the rank
multiplier of every link is at least $c^{\ell(y)}$. Therefore,
\[
\begin{aligned}
    \sum_{i=1}^m \rho(z_{i-1},z_i)
    \geq
    c^{\ell(y)}
    \sum_{i=1}^m \overline d(z_{i-1},z_i) 
    \geq
    c^{\ell(y)}\overline d(x,y).
\end{aligned}
\]

Second, suppose that the chain meets $K_r$, and let $j$ be the smallest
index such that $z_j\in K_r$. Then
$z_0,\ldots,z_{j-1}\notin K_r$. Hence every link in the initial segment
$z_0,\ldots,z_j$ has rank multiplier at least $c^{\ell(y)}$, and
\[
\begin{aligned}
    \sum_{i=1}^m \rho(z_{i-1},z_i)
    \geq
    c^{\ell(y)}
    \sum_{i=1}^j \overline d(z_{i-1},z_i) 
    \geq
    c^{\ell(y)}\overline d(x,z_j) 
    \geq
    c^{\ell(y)}
    \operatorname{dist}_{\overline d}
    \!\left(x,K_{\tau(y)+1}\right).
\end{aligned}
\]
Since every chain falls into one of these two cases, taking the infimum
over all chains gives
\[
    D(x,y)
    \geq
    c^{\ell(y)}
    \min\left\{
        \overline d(x,y),
        \operatorname{dist}_{\overline d}
        \!\left(x,K_{\tau(y)+1}\right)
    \right\}.
\]
Finally, $x\neq y$ gives $\overline d(x,y)>0$, while
$\tau(x)\leq\tau(y)$ implies $x\notin K_{\tau(y)+1}$. Since
$K_{\tau(y)+1}$ is closed, its $\overline d$-distance from $x$ is
strictly positive. Thus $D(x,y)>0$. Therefore
\begin{equation}
    D(x,y)
    \geq
    c^{\ell(y)}
    \min\left\{
        \overline d(x,y),
        \operatorname{dist}_{\overline d}
        \!\left(x,K_{\tau(y)+1}\right)
    \right\}
    >0.
    \label{eq:positive-definiteness-finite-rank}
\end{equation}
If $y=x^*$, the same argument with
$K_{\tau(x)+1}$ gives
\begin{equation}
    D(x,x^*)
    \geq
    c^{\ell(x)}
    \operatorname{dist}_{\overline d}
    \!\left(x,K_{\tau(x)+1}\right)
    >0.
    \label{eq:positive-definiteness-endpoint}
\end{equation}
Thus $D$ is a metric.

Now if we take a chain $(x,y)$ with length one and apply the lower bound $\ell\geq-N$, we get
$D(x,y)\leq c^{-N}\overline d(x,y)$. Hence the identity map from the compact
space $(X,\overline d)$ to the Hausdorff space $(X,D)$ is continuous and
bijective. It is therefore a homeomorphism. Thus $D$ is topologically
equivalent to both $\overline d$ and $d$, and $(X,D)$ is compact and
complete.

\paragraph{Step 4: Construct the contraction.}
Fix $x,y\in X$ and $u\in F(x)$. By Step~2, we know that $\overline{d}$ and $d$ generate the same
topology, so $F(y)$ is compact in $(X,\overline{d})$ and the continuous function
$\overline{d}(u,\cdot)$ attains its infimum on $F(y)$ at some $v\in F(y)$. Since
$u\in F(x)$, \eqref{eq:hausdorff-nonexpansive} gives
\[
\overline{d}(u,v)=\operatorname{dist}_{\overline{d}}\bigl(u,F(y)\bigr)
\le\sup_{u'\in F(x)}\operatorname{dist}_{\overline{d}}\bigl(u',F(y)\bigr)
\le H_{\overline{d}}\bigl(F(x),F(y)\bigr)\le\overline{d}(x,y).
\]
Combining this
with \eqref{eq:rank-increase} yields
\begin{equation}
    \rho(u,v)
    \leq
    c\rho(x,y).
    \label{eq:rho-contraction}
\end{equation}
Now fix a chain $x=x_0,\ldots,x_m=y$ and a point $u_0\in F(x_0)$. Applying
\eqref{eq:rho-contraction} successively, we can construct a chain $u_0,\dots, u_m$ by choosing $u_i\in F(x_i)$ such that
$\rho(u_{i-1},u_i)\leq c\rho(x_{i-1},x_i)$ at each $i \in [m]$. Then
\begin{equation}
    D(u_0,u_m)
    \leq
    c\sum_{i=1}^{m}\rho(x_{i-1},x_i).
    \label{eq:lifted-chain-bound}
\end{equation}
Taking the infimum over all chains gives
\[\operatorname{dist}_D(u_0,F(y))\leq cD(x,y).\]
Then taking the supremum over $u_0\in F(x)$ gives
\[
    \sup_{u\in F(x)}
    \operatorname{dist}_D(u,F(y))
    \leq cD(x,y).
\]
Interchanging $x$ and $y$ gives the corresponding estimate from
$F(y)$ to $F(x)$. Taking the maximum of the two directed estimates proves
\eqref{eq:constructed-nadler-contraction}.

\paragraph{Step 5: Transfer to the original metric.}
Let $S:=K_N$, so that $S\subseteq B_d(x^*,\varepsilon)$. If
$x\notin S$, then $\ell(x)<0$. Every chain $x=z_0,z_1,\ldots,z_m=x^*$ must enter $S$, and every link
$(z_{i-1},z_i)$ up to and including the first entrance $z_j\in S$ has
$z_{i-1}\notin S$ (so $\min\{\ell(z_{i-1}),\ell(z_i)\}\leq-1$) and hence has
$\rho$-length at least its $\overline d$-length. Consequently,
\begin{equation}
    D(x,x^*)
    \geq
    \operatorname{dist}_{\overline d}(x,S).
    \label{eq:distance-to-core-bound}
\end{equation}
If $D(x,x^*)\leq\varepsilon$, then either $x\in S$, in which case
$d(x,x^*)\leq\varepsilon$, or
$\operatorname{dist}_{d}(x,S)\leq
\operatorname{dist}_{\overline d}(x,S)\leq\varepsilon$. In the latter case,
\begin{equation}
    d(x,x^*)
    \leq
    \operatorname{dist}_{d}(x,S)
    +
    \sup_{z\in S}d(z,x^*)
    \leq2\varepsilon.
    \label{eq:endpoint-transfer-proof}
\end{equation}
This proves \eqref{eq:endpoint-transfer}.

For the pairwise statement, suppose first that $x,y\notin S$. A chain from
$x$ to $y$ either avoids $S$, in which case its length is at least
$\overline d(x,y)$, or enters $S$, in which case its initial segment has
length at least $\operatorname{dist}_{\overline d}(x,S)$. Hence
\begin{equation}
    D(x,y)
    \geq
    \min\left\{
        \overline d(x,y),
        \operatorname{dist}_{\overline d}(x,S)
    \right\}.
    \label{eq:pairwise-lower-bound}
\end{equation}
Thus $D(x,y)\leq\varepsilon$ implies either $d(x,y)\leq\varepsilon$ or
$d(x,x^*)\leq2\varepsilon$. If one of $x$ and $y$ already belongs to
$S$, then its $d$-distance from $x^*$ is at most $\varepsilon$.
This proves \eqref{eq:pairwise-transfer} and completes the proof.

\begin{remark}
Theorem~\ref{thm:daskalakis-nadler-converse} is a mathematical universality
result: every robustly globally convergent compact-valued iteration admits a
complete equivalent metric under which all branches contract geometrically.
Its construction is not automatically efficient. In particular, it uses
the globally reachable sets $K_t$, the infinite supremum defining
$\overline d$, and the shortest-chain closure defining $D$. These objects
need not have polynomial-size circuit representations. Consequently, a
computational converse suitable for a \textsf{CLS} reduction requires additional
succinctness assumptions, such as an efficiently evaluable potential
that decreases for every $u\in F(x)$. The compact, Hausdorff-continuous
formulation above is nevertheless compatible with the representation of
correspondences used in computational Kakutani problems
\cite{PapadimitriouVlatakisGkaragkounisZampetakis2023}.  
\end{remark}

\subsection{Proof of Corollary \ref{cor:robust-global-convergence}}
Since $F(x^*)=\{x^*\}$, the contraction inequality implies that every
$u\in F(x)$ satisfies $D(u,x^*)\leq cD(x,x^*)$. This gives
\eqref{eq:trajectory-linear-convergence}. It also yields, for every
$u\in F(x)$,
\begin{align}
    D(x,u)
    &\geq
    D(x,x^*)-D(u,x^*)
    \geq
    (1-c)D(x,x^*), \nonumber\\
    D(x,u)
    &\leq
    D(x,x^*)+D(u,x^*)
    \leq
    (1+c)D(x,x^*).
    \label{eq:residual-two-sided-proof}
\end{align}
Taking the infimum over $u\in F(x)$ proves
\eqref{eq:residual-error-bound}. In particular, a fixed point has zero
distance from $x^*$, proving uniqueness. Finally,
\eqref{eq:residual-error-bound} and
\eqref{eq:trajectory-linear-convergence} give
\[
    D(x_t,x^*)
    \leq
    \frac{c^t}{1-c}r_D(x_0).
\]
If $r_D(x_0)>0$, the choice of $t$ in
\eqref{eq:iteration-complexity-bound} makes the right-hand side at most
$\varepsilon$, so \eqref{eq:endpoint-transfer} gives
$d(x_t,x^*)\leq2\varepsilon$. If $r_D(x_0)=0$, the residual lower bound
gives $x_0=x^*$, and every admissible trajectory remains at the endpoint.

\section{Omitted Proofs in Section \ref{sec:large-margin-triplet-nadler}}
\label{sec:appendix-triplet-loss}
\subsection{Technical lemmas}

We first isolate the quantitative facts used by the reduction.
\begin{lemma}[Conversion between approximate FOSP and KKT point]
\label{lem:fosp-to-kkt-conversion}
Let $ K=[0,1]^q,$ and suppose that $\mathcal L$ is continuously differentiable on $K$ and $\nabla\mathcal L$ is $\Lambda$-Lipschitz on $K$ for
some $\Lambda>0$. Let $\delta>0$, and assume that $z\in K$ is a $\delta$-FOSP for $\mathcal{L}$. Then 
\begin{align}
 y:=
    \Pi_K\left(
        z-\frac{1}{\Lambda}\nabla\mathcal L(z)
    \right)
    \label{eq:fosp-to-kkt-projected-step}    
\end{align}
is a $2\sqrt{\Lambda\delta}$-KKT point of $\mathcal L$ over
$K$.
\end{lemma}

\begin{proof}
Let $g:=\nabla\mathcal L(z).$ By the first-order characterization of Euclidean projection,
\eqref{eq:fosp-to-kkt-projected-step} implies $$ z-\frac{1}{\Lambda}g-y\in
    N_K(y).$$ Since $z\in K$, we may test
$ z-\frac{1}{\Lambda}g-y$ against $z-y$. This gives
\begin{align}
    0
    \geq
    \left\langle
        z-\frac{1}{\Lambda}g-y,
        z-y
    \right\rangle
    =
    \|z-y\|_2^2
    -
    \frac{1}{\Lambda}
    \langle g,z-y\rangle.
    \label{eq:projection-displacement-first}
\end{align}
Consequently, $\|z-y\|_2^2
    \leq
    \frac{1}{\Lambda}
    \langle g,z-y\rangle.$ Applying the $\delta$-FOSP condition with $u=y$ yields $\langle g,z-y\rangle \leq \delta.$
Combining this inequality with
\eqref{eq:projection-displacement-first}, we obtain $\|z-y\|_2
    \leq
    \sqrt{\frac{\delta}{\Lambda}}.$
Now define
\[
    n
    :=
    \Lambda(z-y)-g
    \in
    N_K(y).
    \label{eq:approximate-normal-certificate}
\]

The $\Lambda$-Lipschitz continuity of the gradient implies
\begin{equation}
\label{eq:approximate-normal-residual}
\begin{aligned}
    \|\nabla\mathcal L(y)+n\|_2
    &=
    \left\|
        \nabla\mathcal L(y)
        -
        \nabla\mathcal L(z)
        +
        \Lambda(z-y)
    \right\|_2
    \\
    &\leq
    \|\nabla\mathcal L(y)-\nabla\mathcal L(z)\|_2
    +
    \Lambda\|z-y\|_2
    \\
    &\leq
    2\Lambda\|z-y\|_2
    \leq
    2\sqrt{\Lambda\delta}.
\end{aligned}
\end{equation}

It follows in particular that $\|\nabla\mathcal L(y)+n\|_\infty\leq 2\sqrt{\Lambda\delta}.$
It remains to translate
\eqref{eq:approximate-normal-residual} into the coordinatewise KKT
conditions. For the box $K=[0,1]^q$, every $n\in N_K(y)$ satisfies
\[
    n_i
    \begin{cases}
        \leq 0, & y_i=0,\\
        =0, & 0<y_i<1,\\
        \geq 0, & y_i=1.
    \end{cases}
    \label{eq:box-normal-cone-coordinates}
\]
Suppose that $y_i>0$. If $y_i<1$, then $n_i=0$; if
$y_i=1$, then $n_i\geq0$. In either case,
\[
    \partial_i\mathcal L(y)
    \leq
    \partial_i\mathcal L(y)+n_i
    \leq
    \left|
        \partial_i\mathcal L(y)+n_i
    \right|
    \leq
    2\sqrt{\Lambda\delta}.
\]
Similarly, suppose that $y_i<1$. If $y_i>0$, then $n_i=0$; if
$y_i=0$, then $n_i\leq0$. Hence
\[
    \partial_i\mathcal L(y)
    \geq
    \partial_i\mathcal L(y)+n_i
    \geq
    -
    \left|
        \partial_i\mathcal L(y)+n_i
    \right|
    \geq
    -2\sqrt{\Lambda\delta}.
\]

Finally, substituting $\delta=\frac{\varepsilon^2}{4\Lambda}$ gives $ 2\sqrt{\Lambda\delta}
    =
    \varepsilon,$ so an $\varepsilon^2/(4\Lambda)$-FOSP is converted into an
$\varepsilon$-KKT point. This is an analytic conclusion, since continuous
differentiability alone gives no algorithm for evaluating the gradient.
If $\nabla\mathcal L$ can be evaluated exactly in polynomial time on
rational inputs, with polynomially bounded output length, and
$\Lambda>0$ is a rational of polynomial encoding length, then $y$ is
computable in polynomial time from rational $z$. Specifically, we can evaluate the gradient,
take the rational step, and clip each coordinate to $[0,1]$. The rational
quadratic objectives constructed below have this property. The running
time here is measured in both the instance length and the encoding
length of $z$.
\end{proof}

\begin{lemma}[Large margin leads to non-negativity]
\label{lem:triplet-quadratic-structure}

Let
$W:=\sum_{t\in\mathcal C} w_t$, and assume that $q\geq1$, $W>0$,
and $\alpha \ge d_{\mathrm{emb}}\geq1$. Then the following properties hold.
\begin{enumerate}
    \item[\rm(i)]
    For every $z\in K$ and every triplet $t\in\mathcal C$, we have
    \[
        \alpha-d_{\mathrm{emb}}
        \leq
        \|z_{i_t}-z_{j_t}\|_2^2
        -
        \|z_{i_t}-z_{k_t}\|_2^2
        +
        \alpha
        \leq
        \alpha+d_{\mathrm{emb}}.
    \]
    Consequently, $\mathcal L$ agrees on $K$ with the quadratic polynomial
    \begin{equation}
        \mathcal L(z)
        =
        \sum_{t\in\mathcal C} w_t
        \left(
            \|z_{i_t}-z_{j_t}\|_2^2
            -
            \|z_{i_t}-z_{k_t}\|_2^2
            +
            \alpha
        \right).
        \label{eq:large-margin-quadratic-loss}
    \end{equation}
    \item[\rm(ii)]
    For every $z\in K$, $ W(\alpha-d_{\mathrm{emb}})
        \leq
        \mathcal L(z)
        \leq
        W(\alpha+d_{\mathrm{emb}}).$ Define $R:=2d_{\mathrm{emb}}W$ and
    \begin{equation}
        p(z)
        :=
        \frac{
            \mathcal L(z)-W(\alpha-d_{\mathrm{emb}})
        }{R}.
        \label{eq:normalized-triplet-potential}
    \end{equation}
    Then $p(z)\in[0,1]$ for every $z\in K$.
\end{enumerate}
\end{lemma}

\begin{proof}
For any $a,b\in[0,1]^{d_{\mathrm{emb}}}$, one has
$0\leq\|a-b\|_2^2\leq d_{\mathrm{emb}}$. Hence 
\[
\|z_{i_t}-z_{j_t}\|_2^2
        -
        \|z_{i_t}-z_{k_t}\|_2^2
        +
        \alpha
\]
lies in
$[\alpha-d_{\mathrm{emb}},\alpha+d_{\mathrm{emb}}]$ and is non-negative. This proves \eqref{eq:large-margin-quadratic-loss}. The same bound
implies
$W(\alpha-d_{\mathrm{emb}})\leq\mathcal L(z)\leq
W(\alpha+d_{\mathrm{emb}})$, so \eqref{eq:normalized-triplet-potential}
takes values in $[0,1]$.
\end{proof}


Expanding the squared distances in
\eqref{eq:large-margin-quadratic-loss} gives a rational quadratic in the embeddings, which gives the following bounds. 
\begin{lemma}[Lipschitz  bound for quadratic formula]
\label{lem:lipchitz-quadratic-bound}
Let $$\mathcal L(z)= \frac{1}{2}z^\top H z+h^\top z+c,$$
 where $H\in\mathbb{Q}^{q\times q}$ is symmetric and $c$ is constant. Define
\[B_H=\max_{\ell\in[q]}\sum_{j=1}^{q}|H_{\ell j}|,\quad\Lambda=\max\{1,B_H\},\quad\,G=\max\{1,\|h\|_\infty+B_H\}, \quad L_p=\frac{G}{R}.\]   
Then, for all $x,y,z\in K$, $\nabla \mathcal L$ is $\Lambda$-Lipschitz
with respect to $\ell_2$, with $\|\nabla\mathcal L(z)\|_\infty \leq G,$ and 
$$|p(x)-p(y)| \leq L_p\|x-y\|_1.$$ The matrix $H$, the vector $h$, and all constants introduced above can be computed from the triplet instance in polynomial time and have polynomial
encoding length.
\end{lemma}

\begin{proof}
For $x,y \in K$, we can show that $\nabla \mathcal{L}(z) = Hz+h$ and 
\[
\nabla \mathcal L(x)-\nabla \mathcal L(y) = H(x-y).
\]
Because $H$ is symmetric, its maximum absolute row sum equals its maximum
absolute column sum. Therefore,
\begin{equation}
    \|H\|_2
    \leq
    \sqrt{\|H\|_1\|H\|_\infty}
    =
    B_H
    \leq
    \Lambda.
    \label{eq:triplet-hessian-bound}
\end{equation}

Thus $\nabla\mathcal L$ is $\Lambda$-Lipschitz in $\ell_2$. Since
$\|z\|_\infty\leq1$ on $K$, we also have
$\|Hz+h\|_\infty\leq B_H+\|h\|_\infty\leq G$. The mean-value theorem then
gives, for all $x,y\in K$,
\begin{equation}
    |p(x)-p(y)|
    \leq
    \frac{1}{R}
    \sup_{z\in K}\|\nabla\mathcal L(z)\|_\infty
    \|x-y\|_1
    \leq
    L_p\|x-y\|_1.
    \label{eq:triplet-potential-lipschitz}
\end{equation}

Expanding the squared distances in \eqref{eq:large-margin-quadratic-loss} into quadratic form will produce $H$, $h$, and $c$ using polynomially many rational arithmetic operations, which proves the final claim.
\end{proof}

\begin{lemma}[Projected-gradient certificate]
\label{lem:triplet-projected-gradient-certificate}
Under the assumptions of Lemma~\ref{lem:triplet-quadratic-structure} and \ref{lem:lipchitz-quadratic-bound}, define
\begin{equation}
    T(z)
    :=
    \Pi_K\!\left(
        z-\frac{1}{\Lambda}\nabla\mathcal L(z)
    \right),
    \label{eq:triplet-projected-gradient-map}
\end{equation}
where $\Pi_K$ is Euclidean projection onto $K$, and set
\begin{equation}
    \delta
    :=
    \min\left\{
        \frac12,
        \frac{\varepsilon^2}{8qR\Lambda}
    \right\}.
    \label{eq:triplet-delta}
\end{equation}

The map $T$ is represented by a polynomial-size arithmetic circuit and is
$2$-Lipschitz with respect to $\ell_1$. Furthermore, if $x\in K$ satisfies
$p(T(x))\geq p(x)-\delta$, then $T(x)$ is an $\varepsilon$-FOSP of
$\mathcal L$.
\end{lemma}

\begin{proof}
Projection onto the box $K$ is coordinatewise clipping to $[0,1]$, so it is
computable by a polynomial-size circuit and is nonexpansive in $\ell_1$.
Using $\nabla\mathcal L(x)-\nabla\mathcal L(y)=H(x-y)$ and
\[\|H\|_{1,1}=\max_{\ell\in[q]}\sum_{j=1}^{q}|H_{\ell j}|=B_H\] for symmetric matrix $H$, we obtain
\begin{equation}
 \label{eq:triplet-update-lipschitz}   
 \begin{aligned}
    \|T(x)-T(y)\|_1
    &\leq
    \left\|
        (x-y)
        -
        \frac{1}{\Lambda}
        \bigl(\nabla\mathcal L(x)-\nabla\mathcal L(y)\bigr)
    \right\|_1\\
    &\leq
    \left(1+\frac{B_H}{\Lambda}\right)
    \|x-y\|_1
    \leq
    2\|x-y\|_1.
\end{aligned}
\end{equation}
Here $\|\cdot\|_{1,1}$ is the operator norm when both input and output are measured by $\ell_1$ norm. Now fix $x\in K$ with $p(T(x))\geq p(x)-\delta$ and write $y:=T(x)$.
Projection optimality with the comparison point $x$ gives
$\langle\nabla\mathcal L(x),y-x\rangle
\leq-\Lambda\|y-x\|_2^2$. By $\Lambda$-smoothness,
\begin{align}
    \mathcal L(y)
    \leq
    \mathcal L(x)
    +
    \langle\nabla\mathcal L(x),y-x\rangle
    +
    \frac{\Lambda}{2}\|y-x\|_2^2
    \leq
    \mathcal L(x)
    -
    \frac{\Lambda}{2}\|y-x\|_2^2.
    \label{eq:triplet-descent}
\end{align}
Since $\mathcal L=Rp+W(\alpha-d_{\mathrm{emb}})$, the assumed potential
inequality and \eqref{eq:triplet-descent} imply
$\|y-x\|_2\leq\sqrt{2R\delta/\Lambda}$. For any $u\in K$, the variational inequality for the projection in
\eqref{eq:triplet-projected-gradient-map} yields

\begin{equation}
    \langle\nabla\mathcal L(x),u-y\rangle
    \geq
    \Lambda\langle x-y,u-y\rangle.
    \label{eq:triplet-projection-vi}
\end{equation}
Combining \eqref{eq:triplet-projection-vi} with the $\Lambda$-Lipschitz
continuity of the gradient and using $\|u-y\|_2\leq\sqrt q$, we obtain
\begin{align}
    \langle u-y,\nabla\mathcal L(y)\rangle
    \geq
    -2\Lambda\|x-y\|_2\|u-y\|_2
    \geq
    -\sqrt{8qR\Lambda\delta}
    \geq
    -\varepsilon.
    \label{eq:triplet-vi-stationarity}
\end{align}
Thus $y=T(x)$ satisfies Definition~\ref{def:large-margin-triplet-fosp}.
\end{proof}

\begin{lemma}[Potential ultrametric]
\label{lem:triplet-potential-ultrametric}
Let $p:K\to[0,1]$ be the potential from
\eqref{eq:normalized-triplet-potential}, and define

\begin{equation}
    d_p(x,y)
    :=
    \begin{cases}
        0, & x=y,\\[1mm]
        1+\max\{p(x),p(y)\}, & x\neq y.
    \end{cases}
    \label{eq:triplet-potential-metric}
\end{equation}

Then $d_p$ is a complete ultrametric on $K$ with diameter at most two. In
addition, for any off-diagonal pairs $(a,b)$ and $(a',b')$,

\begin{equation}
    |d_p(a,b)-d_p(a',b')|
    \leq
    L_p\bigl(\|a-a'\|_1+\|b-b'\|_1\bigr).
    \label{eq:triplet-metric-lipschitz}
\end{equation}
\end{lemma}

\begin{proof}
If $x,y,z$ are pairwise distinct, then
$d_p(x,z)\geq1+p(x)$ and $d_p(z,y)\geq1+p(y)$, so
$d_p(x,y)\leq\max\{d_p(x,z),d_p(z,y)\}$. The remaining cases are immediate.
Thus $d_p$ is an ultrametric. Every nonzero distance belongs to $[1,2]$;
consequently, every $d_p$-Cauchy sequence is eventually constant. Hence
$(K,d_p)$ is complete and has diameter at most two.

For off-diagonal pairs, the constant term in
\eqref{eq:triplet-potential-metric} cancels. The elementary Lipschitz property
of the maximum and Lemma~\ref{lem:lipchitz-quadratic-bound} give
\begin{align*}
    |d_p(a,b)-d_p(a',b')|
    &\leq
    \max\left\{
        |p(a)-p(a')|,
        |p(b)-p(b')|
    \right\}\\
    &\leq
    L_p\bigl(\|a-a'\|_1+\|b-b'\|_1\bigr).
\end{align*}
This proves \eqref{eq:triplet-metric-lipschitz}.
\end{proof}

\subsection{Proof of Theorem \ref{thm:large-margin-triplet-to-nadler} and Corollary \ref{cor:large-margin-triplet-cls-complete}}

\begin{proof}[Proof of Theorem \ref{thm:large-margin-triplet-to-nadler}]
   If $V=\emptyset$, the reduction maps the input to any fixed valid
   $\beta$-\textsc{Projected-Nadler} instance (for $\beta=1$) and returns the empty tuple
   as the exact FOSP. If $V\neq\emptyset$ and $W=\sum_{t\in\mathcal C} w_t=0$, the objective
   is constant; in this case, we also map the input to any fixed valid instance and decode
   every answer as $z=0$, which is an exact FOSP. Assume that $W>0$ and $V\neq\emptyset$, and
construct $p$ from Lemma \ref{lem:triplet-quadratic-structure}, $T, \delta$ from Lemma \ref{lem:triplet-projected-gradient-certificate}, and $d_p$ from
Lemma~\ref{lem:triplet-potential-ultrametric}. Define the singleton correspondence
$F_{\mathcal L}(x):=\{T(x)\}$ and its projection circuit by
$\Pi_{F_{\mathcal L}}(x,z):=T(x)$. Since every value is a singleton, this is
an exact projection circuit under $d_p$. The target instance uses
$$D=2, \quad k=1-\delta/2, \quad \varepsilon_1=1/2, \quad \varepsilon_2=\delta/8, \quad L_P=2, \quad L_d=L_p.$$    
Lemma~\ref{lem:triplet-potential-ultrametric} proves completeness of the metric space and the
diameter bound, and $\delta\leq1/2$ ensures $k\in(0,1)$. Also, $(1-k)\varepsilon_1 = \delta/4 > \delta/8 = \varepsilon_2 > 0$, which satisfies the promise on accuracy.

Suppose first that the solver returns a point $x$ of type
\textnormal{(N1)}. Then $d_p(x,T(x))\leq1/2$. Because every nonzero
$d_p$-distance is at least one, $x=T(x)$. In particular,
$p(T(x))=p(x)$, so Lemma~\ref{lem:triplet-projected-gradient-certificate}
shows that $x$ is an $\varepsilon$-FOSP.

Now suppose that the solver returns a certificate of type \textnormal{(N2)}.
Since the projection circuit ignores its second argument, either branch
of the certificate gives
\begin{equation*}
    d_p(T(x),T(y))
    > k\,d_p(x,y)+\varepsilon_2
    > k\,d_p(x,y).
\end{equation*}
In particular, $T(x)\neq T(y)$, which also implies $x\neq y$.
We show that at least one $v\in\{x,y\}$ satisfies
$p(T(v))\geq p(v)-\delta$.
Suppose otherwise, so that
$p(T(x))<p(x)-\delta$ and $p(T(y))<p(y)-\delta$.
Since both pairs $(x,y)$ and $(T(x),T(y))$ consist of distinct points,
the off-diagonal formula for $d_p$ gives
\begin{equation}
\label{eq:triplet-n2-contradiction}
\begin{aligned}
    d_p(T(x),T(y))
    &= 1+\max\{p(T(x)),p(T(y))\} \\
    &< 1+\max\{p(x),p(y)\}-\delta \\
    &= d_p(x,y)-\delta \\
    &\leq \left(1-\frac{\delta}{2}\right)d_p(x,y) \\
    &= k\,d_p(x,y),
\end{aligned}
\end{equation}
where the last inequality follows from $d_p(x,y)\leq 2$.
This contradicts the returned certificate.
Thus, the decoder checks the potential-decrease condition at $x$ and $y$,
selects a point $v$ satisfying $p(T(v))\geq p(v)-\delta$,
and returns $T(v)$.
By Lemma~\ref{lem:triplet-projected-gradient-certificate},
this point is an $\varepsilon$-FOSP.

Finally, equation \eqref{eq:triplet-update-lipschitz} rules out type
\textnormal{(N3)} for $\beta=1$, and equation \eqref{eq:triplet-metric-lipschitz} rules out
type \textnormal{(N4)}. The matrices, constants, and arithmetic circuits used
above have polynomial encoding length; the metric circuit uses two
evaluations of $p$, a maximum gate, and comparison gates for equality, and
the decoder evaluates at most two projected-gradient updates. Therefore both
the instance map and the solution decoder run in polynomial time.
\end{proof}

\begin{proof}[Proof of Corollary \ref{cor:large-margin-triplet-cls-complete}]
The \textsf{CLS}-membership follows from Theorems~\ref{thm:large-margin-triplet-to-nadler}
and Theorem~\ref{thm:nadler-clo-equivalence}. For \textsf{CLS}-hardness, we reduce from
\textsc{QuadraticProgram-KKT} (Definition~\ref{def:quadratic-program-kkt}), which remains
\textsf{CLS}-complete when the accuracy $\varepsilon>0$ is given in binary
\cite{Fearnley2025KKT}. Let $(Q,b,\varepsilon)$ be an instance on $[0,1]^n$, and write
$p(x)=x^\top Qx+b^\top x$.

We use the construction in the proof of
\cite[Theorem~4.1]{YanEtAlTriplet} with embedding dimension $d_{\mathrm{emb}}=1$, pivots
$a_A=0$ and $a_B=1/2$, margin $\alpha=1=d_{\mathrm{emb}}$, and one movable point per
variable. Thus, we have $K=[0,1]^n$ and $z_i=x_i$. Each coefficient of $p$ is realized by a
triplet and its dual, with nonnegative rational weights whose difference is a prescribed
rational linear combination of the entries of $Q$ and $b$. This gives
$\mathcal L(z)=p(z)+C$ on $K$ for some constant $C$, and consequently $\nabla\mathcal L=\nabla p$ on $K$.
Since $\alpha\ge d_{\mathrm{emb}}$, this is an instance of
Definition~\ref{def:large-margin-triplet-fosp}. Let $\Lambda$ be the constant of
Lemma~\ref{lem:lipchitz-quadratic-bound} for $\mathcal L$, and set the target accuracy to
$\varepsilon_T:=\varepsilon^2/(4\Lambda)$. All weights, $\Lambda$ and $\varepsilon_T$ have
polynomial encoding length.

Given an $\varepsilon_T$-FOSP $z\in K$, output
$y:=\Pi_K\bigl(z-\nabla\mathcal L(z)/\Lambda\bigr)$, which is computable in polynomial time
by coordinatewise clipping. By Lemma~\ref{lem:fosp-to-kkt-conversion}, $y$ is an
$\varepsilon$-KKT point of $\mathcal L$ over $K$. Since $\nabla\mathcal L=\nabla p$ on $K$,
$y$ is an $\varepsilon$-KKT point of~\eqref{eq:quadratic-program-kkt}.
\end{proof}

\end{document}